\documentclass [reqno] {amsart}

\usepackage{ulem}
\usepackage{graphicx}
\usepackage{fancyhdr}
\usepackage{amssymb}
\usepackage{amsmath}
\usepackage{amsthm}
\usepackage{xspace}
\usepackage{amscd}
\usepackage{amsfonts, epsfig}

\usepackage{xcolor}
\usepackage{pstricks,pst-node}
\usepackage{tikz}
\usetikzlibrary{backgrounds}
\usetikzlibrary{patterns,fadings}
\usetikzlibrary{arrows,decorations.pathmorphing}
\usetikzlibrary{shapes}
\usetikzlibrary{calc}

\usepackage{fancybox}

\definecolor{brightcerulean}{rgb}{0.11, 0.67, 0.84}
\definecolor{cerulean}{rgb}{0.0, 0.48, 0.65}
\definecolor{Gray}{rgb}{0.5, 0.5, 0.5}

\usepackage{enumerate}
\usepackage{url}

\usepackage{array}
\usepackage{multirow}

\usepackage{float}

\newcommand{\N}{\mathbb{N}}                                              
\newcommand{\Z}{\mathbb{Z}}   
 \newcommand{\ZZ}{\mathbb{Z}}

\newcommand{\R}{\mathbb{R}}     
\newcommand{\RR}{\mathbb{R}}                                        
\newcommand{\C}{\mathbb{C}}
\newcommand{\T}{\mathbb{T}}
\renewcommand{\P}{\mathbb{P}}
\newcommand{\E}{\mathbb{E}}
\newcommand{\EE}{\mathbb{E}}

\newcommand{\A}{\mathcal{A}}

\newcommand{\bE}{\mathbf{E}}
\newcommand{\bV}{\mathbf{V}}
\newcommand{\bQ}{\mathbf{Q}}

\renewcommand{\S}{\mathcal{S}}
\newcommand{\cL}{\mathcal{L}}
\newcommand{\cE}{\mathcal{E}}

\newcommand{\cH}{\mathcal{H}}

\newcommand{\cD}{\mathcal{D}}
\newcommand{\cF}{\mathcal{F}}
\newcommand{\cC}{\mathcal{C}}

\renewcommand{\le}{\leq}
\renewcommand{\ge}{\geq}

\newcommand\bx{{\mathbf x}}
\newcommand\by{{\mathbf y}}
\newcommand\be{{\mathbf e}}

\newcommand{\benum}{\begin{enumerate}}
\newcommand{\eenum}{\end{enumerate}}
\newcommand{\bitem}{\begin{itemize}}
\newcommand{\eitem}{\end{itemize}}

\newcommand{\barray}{\begin{array}}
\newcommand{\earray}{\end{array}}

\newcommand{\vertiii}[1]{{\left\vert\kern-0.25ex\left\vert\kern-0.25ex\left\vert #1 
    \right\vert\kern-0.25ex\right\vert\kern-0.25ex\right\vert}}

\newtheorem{theo}{\textsc{Theorem}}[section]
\newtheorem{cor}[theo]{\textsc{Corollary}}
\newtheorem{lem}[theo]{\textsc{Lemma}}
\newtheorem{prop}[theo]{\textsc{Proposition}}

\newtheorem{de}{\textsc{Definition}}[section]

\theoremstyle{remark}

\newtheorem{rem}{\textsc{Remark}}[section]

\newcommand{\mc}[1]{{\mathcal #1}}
\newcommand{\mf}[1]{{\mathfrak #1}}

\newcommand{\bb}[1]{{\mathbb #1}}

\usepackage[colorlinks=true,linkcolor=blue,citecolor=magenta]{hyperref}

\author{C\'edric Bernardin}
\address{Universit\'e C\^ote d'Azur, CNRS, LJAD\\
Parc Valrose\\
06108 NICE Cedex 02, France}
\email{{\tt cbernard@unice.fr}}

\author{Patr\'{\i}cia Gon\c{c}alves}
\address{\noindent Center for Mathematical Analysis,  Geometry and Dynamical Systems \\
Instituto Superior T\'ecnico, Universidade de Lisboa\\
Av. Rovisco Pais, no. 1, 1049-001 Lisboa, Portugal}
\email{patricia.goncalves@math.tecnico.ulisboa.pt}

\author{ Milton Jara}
\address{\noindent Instituto de Matem\'atica Pura e Aplicada\\ Estrada Dona Castorina 110\\ 22460-320 Rio De Janeiro, Brazil.}
\email{mjara@impa.br}

\author{Marielle Simon}
\address{\noindent Inria Lille -- Nord Europe \\ 40 avenue du Halley \\ 59650 Villeneuve d'Ascq, France\\ \textit{and} Laboratoire Paul Painlev\'e, UMR CNRS 8524 \\ Cit\'e Scientifique \\ 59655 Villeneuve d'Ascq, France}
\email{marielle.simon@inria.fr}

\begin{document}

\title[]{Nonlinear perturbation of a noisy Hamiltonian lattice field model: universality persistence}

\begin{abstract}
In \cite{BGJ} it has been proved that a linear Hamiltonian lattice field with two conservation laws, perturbed by a conservative stochastic noise, belongs to the $\frac32$-L\'evy/Diffusive universality class in the nonlinear fluctuating theory terminology \cite{S}, i.e.~energy superdiffuses like an asymmetric stable $\frac32$-L\'evy process and volume like a Brownian motion. According to this theory this should remain valid at zero tension if the harmonic potential is replaced by an even potential. In this work we consider a quartic anharmonicity and show that the result obtained in the harmonic case persists up to some small critical value of the anharmonicity.  
\end{abstract}

\maketitle


\section{Introduction}

During the last two decades there has been a strong regain of interest  in the understanding of anomalous diffusion in asymmetric one dimensional systems with several conservation laws, whose typical examples are given by chains of coupled oscillators \cite{LLN}. During several years contradictory numerical simulations have been performed and their accuracy has been strongly debated without a clear consensus between specialists. Recently important progresses have been obtained with the development of the so-called \textit{nonlinear fluctuating hydrodynamics theory} developed by Spohn \cite{S}. The theory identifies precisely the universality classes describing the form of the anomalous diffusion in terms of macroscopic thermodynamical quantities associated to the microscopic system and also explains why so many numerics provided so different conclusions. Roughly Spohn's approach consists to start with the hyperbolic system of conservation laws governing the macroscopic evolution of the empirical conserved quantities, then add diffusion and dissipation to this system of coupled PDEs and linearize the system at second order w.r.t.~equilibrium values of the conserved quantities. In the calculations a fundamental role is played by the normal modes, i.e.~the eigenvectors of the linearized equation, called \textit{heat mode} and \textit{sound modes} in \cite{S}. These modes evolve with different velocities in different time scales and may be described by different forms of anomalous superdiffusion  or by a standard diffusion.  

On the other hand a rigorous justification of Spohn's predictions is lacking and some of them are in contradiction\footnote{The contradiction exists in particular because the kinetic predictions are done by using the kinetic equations outside the time scale where they are expected to be valid.} with kinetic theory \cite{LS,MM}. Until now the nonlinear fluctuating hydrodynamics predictions have been fully justified only for {\textit{linear}} Hamiltonian lattice field models perturbed by a noise conserving the energy and one or two extra quantities \cite{BGJ, JKO2}. The universality classes identified in these works are described by a skew or symmetric $3/4$-fractional diffusion equation for the heat mode (i.e.~the energy, with zero velocity) and a normal diffusion for the sound modes (with non-zero velocity; volume in \cite{BGJ}; stretch and momentum in \cite{JKO}). 

In the Hamiltonian lattice field model, if linear interactions are replaced by nonlinear interactions, some new universality classes, such as the famous Kardar-Parisi-Zhang (KPZ)  universality class, may appear. Proving any result confirming this picture is of course a highly challenging problem, even in the case of systems with a single conservation law. For the latter only two universality classes are possible: the KPZ universality class and the Edwards-Wilkinson (diffusive) class. Only few one dimensional stochastic asymmetric models (e.g.~the exclusion process) with one conserved quantity have been proved to belong to the KPZ universality class (see e.g.~\cite{C,QS} and references therein). For models with several conserved quantities the question is completely open and even more interesting. Indeed, a special feature of models with several conservation laws is the fact that different time scales coexist in the same model, which never occurs for systems with only one conserved field.

In this work we consider a small quartic nonlinear perturbation of the linear Hamiltonian lattice field model with conservative noise considered in \cite{BGJ}. In the absence of nonlinearities, as mentioned above, the model belongs to the universality class described by a skew $3/4$-fractional diffusion equation for the energy and a normal diffusion for the volume. According to Spohn's theory (see \cite{SS}) if the nonlinear perturbation is driven by an even potential and if the tension is null, the model still belongs to the same universality class. The purpose of this work is to show rigorously that it is the case for very small nonlinear perturbations. Despite our results remain quite limited, they are the first results of this type for nonlinear interactions.

The strategy of the proof follows the general scheme introduced in \cite{BGJ}. The success of this strategy, in the linear case, is due to the fact that the $n$-point correlation functions form a complicated but in any case, a closed system.  Dealing with nonlinear potentials the problem is much harder since this last property is lost and we have to manage the control of a hierarchy. The paper quantifies the intensity with which we can perturb the linear system in order to be able to cut the hierarchy as if we considered only the linear system. The control of the error terms produced by this cut-off requires several standard techniques of interacting particle systems as well as some ad-hoc estimates. Observe also that we are only considering the case of a perturbation given by an even potential with a zero tension. If one of these conditions is not respected we expect to reach a different universality class.  

The paper is organized as follows. In Section \ref{sec:model} we introduce the model we study and in Section \ref{sec:mr} we state our main results. Some technical material is introduced in Section \ref{sec:esti} while the proofs of the two main theorems are given in Section \ref{sec:volumeevolution} and \ref{sec:te-eff}. In Appendix \ref{app:prediction} we explore the nonlinear fluctuating hydrodynamics predictions. The other three appendices contain technical computations. \\

{\bf{Notations:}} Given two real-valued functions $f$ and $g$ depending on the variable $u \in \bb R^d$ we will write $f(u) \approx g(u)$ if there exists a constant $C>0$ which does not depend on $u$ such that for any $u$, $C^{-1} f(u) \le g(u) \le C f(u)$ and $ f(u) \lesssim g(u)$ if for any $u$, $f(u) \le C g(u)$. We write $f =\mc O (g)$ (resp. $f=o(g)$) in the neighborhood of $u_0$ if $| f| \lesssim | g|$ in the neighborhood of $u_0$ (resp. $\lim_{u \to u_0} f(u)/g(u) =0$). Sometimes it will be convenient to precise the dependence of the constant $C$ on some extra parameters and this will be done by the standard notation $C(\lambda)$ if $\lambda$ is the extra parameter. Finally, we denote by $\cC_c^\infty(\R^d)$ the space of infinitely differentiable functions $f : \R^d \to \R$ with compact support.

\section{Model and notations}
\label{sec:model}

\subsection{Perturbed Hamiltonian lattice field model}


We consider a linear Hamiltonian lattice field model \cite{BS} at equilibrium perturbed by an energy conserving noise. This is a Markov process defined on the state space $\Omega=\RR^{\ZZ}$. A typical configuration $\omega \in \Omega$ is denoted by $\omega:=\{\omega_x\,; \, x \in \ZZ\}$. We then perturb it by adding a small anharmonicity which is  regulated by the small parameter $\gamma >0$. Let us define the infinitesimal generator $\cL_{\gamma}$ of the model as $\cL_{\gamma}:=\A_{\gamma}+\S$ where for any $\gamma >0$ we denote 

\[
\A_\gamma:=\sum_{x\in\Z}\Big\{ (\omega_{x+1}-\omega_{x-1}) + \gamma (\omega_{x+1}^3-\omega_{x-1}^3)\Big\} \frac{\partial}{\partial\omega_x}
\]
and for all local\footnote{A function $g:\Omega \to \RR$ is \textit{local} if it depends on the variable $\omega \in \Omega$ only through a finite number of $\{\omega_x \; ; \; x \in \Z\}$.} functions $\varphi: \Omega \to \R$
\[
(\S \varphi)(\omega) :=\sum_{x \in \Z} \Big\{\varphi(\omega^{x,x+1})-\varphi(\omega)\Big\}.
\]
 Above we denote by $\omega^{x,x+1}$ the configuration that is obtained from $\omega$ by exchanging $\omega_x$ and $\omega_{x+1}$, keeping the other values identical, namely: 
 \[(\omega^{x,x+1})_z = \begin{cases} \omega_{x+1} & \text{ if } z=x; \\ \omega_x & \text{ if } z=x+1; \\ \omega_z & \text{ otherwise.} \end{cases} 
 \]
 The Liouville operator $\mc A_\gamma$ is the usual generator associated to the Hamiltonian dynamics of an infinite number of coupled oscillators, where the one-site energy is the sum of the kinetic energy and the potential energy, as follows: for any $u \in \R$ and  $\gamma >0$ let us introduce the local energy
\[
e_\gamma (u):= \frac{u^2}{2}+\gamma\; \frac{u^4}{4}.
\] 
For each $x \in \Z$, the energy of the atom $x$ is simply $e_\gamma(\omega_x)$. When $\gamma=0$, the energy is \textit{harmonic}, whereas $\gamma \to 0$ is the \textit{weakly anharmonic} case. The operator $\mc S$ is the stochastic noise that acts on configurations by exchanging nearest neighbour variables $\omega_x$ and $\omega_{x+1}$ at random Poissonian  times. 

The existence of a Markov process $\{\omega^\gamma (t) \; ; \; t \geq 0\}$ with state space $\Omega$ and generator $\mc L_\gamma$ is provided by usual techniques (see \cite{BS} and references therein). It has a family of invariant measures, called \textit{Gibbs equilibrium measures}, given by 
\[
d\nu_{\beta,\tau,\gamma}(\omega):= \prod_{x \in \Z}\frac{e^{-\beta e_{\gamma}(\omega_x) - \lambda \omega_x}}{Z_\gamma(\beta,\lambda)} d\omega_x,\qquad \lambda=\tau \beta  \in \R, \beta >0,
\]
that are associated to the two conserved quantities, called \textit{volume} and \textit{energy}, formally given by
\[
\mathcal{V}:=\sum_{x\in \Z} \omega_x, \qquad \mathcal{E}_\gamma:=\sum_{x \in \Z} e_{\gamma} (\omega_x).
\]
Above, $Z_\gamma(\beta,\lambda)$ is the normalization constant. The parameters $\beta^{-1} >0$ and $\tau \in \R$ are called, respectively, \textit{temperature} and \textit{tension}. 

For any $(\beta,\tau,\gamma)\in (0,\infty)\times \R\times (0,\infty)$, let us denote by $\langle \varphi \rangle_{\beta,\tau,\gamma}$ the average of $\varphi: \Omega\to \bb R$ with respect to $\nu_{\beta,\tau,\gamma}$ and by $\langle \cdot, \cdot \rangle_{\beta,\tau,\gamma}$ the corresponding ${\bb L}^2$-scalar product. Let us define 
\begin{align}
\mf e_\gamma (\beta,\tau)&:=\big\langle e_{\gamma} (\omega_0) \big\rangle_{\beta,\tau,\gamma}=-\frac{\partial}{\partial\beta}\big(\log Z_\gamma (\beta,\lambda)\big), \label{eq:energy}\\
\mf v_\gamma (\beta,\tau)&:= \big\langle \omega_0 \big\rangle_{\beta,\tau,\gamma}=-\frac{\partial}{\partial\lambda}\big(\log Z_\gamma (\beta,\lambda)\big). \label{eq:volume}
\end{align}
From here on, we consider the dynamics described by the accelerated generator $n^a \mc L_{\gamma_n}$ (therefore the system evolves on the time scale $tn^a$ for some $a>0$), where $\gamma_n$ now depends also on the scaling parameter in such a way that $\lim_{n\to \infty} \gamma_n =0$.  The dependence of $\gamma_n$ with respect to the scaling parameter $n$ will be precised later. We assume that the dynamics starts from the Gibbs equilibrium  measure $\nu_{\beta,0,\gamma_n}$ at  temperature $\beta^{-1}$ and tension $\tau=0$, and we look at its evolution during a time interval $[0,T]$, where $T>0$ is fixed. The law of the resulted process \[\Big\{\omega^{\gamma_n}_x(tn^a)\; ; \; t \in [0,T]\,,\, x\in\bb Z\Big\}\] is simply denoted by $\P$, and the expectation with respect to $\P$ is denoted by  $\E$. For the sake of readability, from now on we denote $\omega_x^{\gamma_n}(tn^a)$ simply by $\omega_x(tn^a)$.

\subsection{Energy and volume fluctuation fields}

 We define, for any test function $f\in\cC_c^\infty(\R)$ and $\omega \in \Omega$,
\begin{equation*}
\cE^n(f, \omega):=\frac{1}{\sqrt n}\sum_{x\in\Z}f\Big(\frac{x}{n}\Big)\Big\{e_{\gamma_n} (\omega_x)-\mf e_{\gamma_n} (\beta,0)\Big\},
\end{equation*}
and the dynamical \textit{energy fluctuation field} by
\[
\cE_t^n(f):=\cE^n (f,\omega (tn^a)).
\]
Similarly,  we define, for any test function $f\in\cC_c^\infty(\R)$ and $\omega \in \Omega$
\[
\mc V^n(f,\omega):=\frac{1}{\sqrt n}\sum_{x\in\Z}f\Big(\frac{x}{n}\Big)\Big\{\omega_x-\mf v_{\gamma_n}(\beta,0)\Big\},
\]
and the dynamical  \textit{volume fluctuation field} by
\[ \mc V_t^n(f):=\mc V^n(f,\omega(tn^a)).\]
 Let $g \in \cC_c^\infty(\R)$ be a fixed function. The goal of this paper is to study the behavior as $n \to \infty$ of the \textit{correlation energy and volume fields} given for any test function $f\in\cC_c^\infty(\R)$ by
\begin{align*}
\bE_t^n(f):&=\E\big[ \cE_0^n(g) \; \cE_t^ n(f)\big]\\
&=\E\bigg[ \cE_0^n(g)  \; \bigg\{\frac{1}{4\sqrt n}\sum_{x\in\Z}f\Big(\frac{x}{n}\Big) \Big(\big[2\omega_x^2+\gamma_n\omega_x^4\big](tn^a)-4\mf e_{\gamma_n}(\beta,0)\Big)\bigg\}\bigg],\\
\bV_t^n(f):&=\E\big[ \mc V_0^n(g) \; \mc V_t^ n(f)\big]\\
&=\E\bigg[\mc V_0^n(g)\bigg\{\frac{1}{\sqrt n}\sum_{x\in\Z}f\Big(\frac{x}{n}\Big)\Big(\omega_x(tn^a)-\mf v_{\gamma_n}(\beta,0)\Big)\bigg\}\bigg].
\end{align*}

We show in Appendix \ref{app:prediction} that we have the following expansions  when $\gamma_n \to 0$:
\begin{align}
{\mf v}_{\gamma_n}(\beta,0) & = o(\gamma_n) \label{eq:expv} \\
{\mf e}_{\gamma_n}(\beta,0) & = \frac{1}{2\beta} - \frac{3\gamma_n}{4\beta^2} + o(\gamma_n). \label{eq:expe}
\end{align}

\subsection{Notations and definitions}

For any $f \in \cC_c^\infty(\R)$ and $h\in\cC_c^\infty(\R^2)$ we introduce two $\ell^2$-norms defined as follows:  
\begin{equation}
\|f\|^2_{2,n}:=\tfrac{1}{n}\sum_{x\in\bb Z} f^2\big(\tfrac{x}{n}\big), \qquad 
\big(N_n^{\neq}(h)\big)^2:=\tfrac{1}{n^2} \sum_{x \neq y} h^2\big(\tfrac{x}{n},\tfrac{y}{n}\big). \label{eq:normh}
\end{equation}
The discrete gradient and discrete Laplacian of $f$ are defined as usual by
\begin{equation*}
(\nabla_n f) \big(\tfrac{x}{n}\big):=n \big[ f \big(\tfrac{x+1}{n}\big) -f  \big(\tfrac{x}{n}\big) \big], \quad \;  (\Delta_n f) \big(\tfrac{x}{n}\big):=n^2 \big[ f \big(\tfrac{x+1}{n}\big) +f  \big(\tfrac{x-1}{n}\big) -2  f\big(\tfrac{x}{n}\big)\big]
\end{equation*} 
where $x\in\bb Z$ and $n\ge 1$ is the inverse of the mesh of the discretization\footnote{The reader will notice that the previous definitions depend in fact only on the values of the functions $f$ and $h$ respectively on $\frac1n \Z$ and $\frac1n \Z \times \frac1n \Z$ so that they can be generalized to functions defined only on these sets.}. 
For our purpose, we also need to define three Fourier transforms:
\begin{itemize}
\item \textit{Fourier transform of integrable functions -- } If $f: \R \to \R$ is an integrable function, we define its Fourier transform $\mathcal{F} (f):\R\to\C$ as
\begin{equation}
\label{fou tranf cont}
 \mathcal{F}(f) (\xi):=\int_\R f(u) e^{2i\pi\xi u} \, d u, \qquad \xi \in \R.\end{equation}

\item \textit{Fourier transform of square summable sequences -- } If $h:\Z\to\R$ is square summable, we define its Fourier transform $\widehat{h}:\T\to\C$ in $\mathbb{L}^2(\T)$ as
\begin{equation}
\label{fou tranf}
 \widehat{h}(\theta):=\sum_{x\in\Z} h(x) e^{2i\pi \theta x}, \qquad \theta \in \T.\end{equation}

\item\label{FF3} \textit{Discrete Fourier transform of integrable functions -- } If $g: \R \to \R$ is an integrable function, we define its discrete Fourier transform $\cF_n(g):\R\to\C$ as
\begin{equation*}
\cF_n(g) (\xi) = \tfrac{1}{n} \sum_{x \in \Z} g \big(\tfrac{x}{n}\big) e^{2i \pi x \frac{\xi}{n}}, \qquad \xi\in \R.
\end{equation*}
\end{itemize}
These definitions can easily be extended for $d$-dimensional spaces, $d\geq 1$. 
Finally, given some parameters $(\beta,\tau,\gamma)$ and  $\varphi \in \bb L^2 (\nu_{\beta,\tau,\gamma})$ belonging to the domain of $\S$, the  \textit{Dirichlet form} of $\varphi$ is given by 
\begin{equation} \label{eq:diric}
\mc D(\varphi):=\big\langle\varphi\; , \;  (-\S) \varphi\big\rangle_{\beta,\tau,\gamma}.
\end{equation}
Observe that we do not precise the dependence on $(\beta,\tau,\gamma)$ in \eqref{eq:diric}. Any time we will use the Dirichlet form, there will be no confusion regarding the values of the parameters. Similarly, for any $\varphi \in \bb L^2 (\nu_{\beta,\tau,\gamma})$ we introduce, for any $z>0$, the $\mc H_{-1,z}$ norm given by 
\begin{align}
\big\|\varphi\big\|_{-1,z}^2:&=\Big\langle\varphi\; ,\;  \big(z-\S\big)^{-1}\varphi\Big\rangle_{\beta,\tau,\gamma}\notag\\ &= \sup_{g} \Big\{2\;\big\langle \varphi\;,\;  g\big\rangle_{\beta,\tau,\gamma}-z\big\langle g^2\big\rangle_{\beta,\tau,\gamma}- \big\langle g\;,\; (-\S)g\big\rangle_{\beta,\tau,\gamma}\Big\}
\label{eq:norm-1}
\end{align}
where $g: \Omega \to \RR$ belongs to the set of local  bounded functions.

\section{Statement of the main results}
\label{sec:mr}

Let us assume that 
\[ \gamma_n = \frac{c}{n^b} \qquad \text{for some } c,b>0,\]
and recall that the time scale is $tn^a$, with $a>0$. Our main convergence theorems depend on the range of the parameters $(a,b)$. Recall that $\tau=0$. In the nonlinear fluctuating theory framework this choice implies in particular the identification of the sound mode with the volume, and the heat mode with the energy. 

\subsection{Macroscopic fluctuations}

\begin{theo}[Volume fluctuations in the time scale $tn^a$ with  $a \le 1$]\label{theo:volume}

\quad 

Let us fix $f,g \in \cC_c^\infty(\R)$, and $t>0$. The macroscopic volume behavior follows the dichotomy:
\begin{enumerate}[1.]
\item If $a<1$, then   
\[ \lim_{n\to \infty} {\bf V}_t^n(f) = \lim_{n\to\infty} {\bf V}_0^n(f)=\frac{1}{\beta}\iint_{\bb R^2} f(u)g(v)   \; du dv.\]
\item If $a=1$, then 
\[ \lim_{n\to \infty} {\bf V}_t^n(f) = \frac{1}{\beta}\iint_{\bb R^2} f(u)g(v) \hat P_t(u-v)  \; du dv,
\] where $\{\hat P_t \, ; \, t\ge 0\}$ is the semi-group generated by the transport operator $-2\nabla$. 
\end{enumerate}
\end{theo}

In order to study the fluctuations in the time scale $a>1$, we need first to recenter the fluctuation field in a frame moving with some specific velocity. Let us denote  $\chi_n:= \langle \omega_0^2 \rangle_{\beta,0,\gamma_n}$ which satisfies $\chi_n \to \beta^{-1}$ as $n\to\infty$ (see Appendix \ref{app:prediction}). We define \[{\bf c}_n:=-2-6\chi_n\gamma_n\] which  is essentially the \textit{sound mode velocity} $c(\beta, 0, \gamma_n)$ (defined in Appendix \ref{app:prediction}) at first order in $\gamma_n$. We now introduce the new volume fluctuation field 
$\widetilde{\bV}_t^n(f)$, which is defined on a moving reference frame as follows:
\begin{equation}
\widetilde{\bV}_t^n(f):=\E\bigg[\mc V_0^n(g)\; \bigg\{\frac{1}{\sqrt n}\sum_{x\in\Z}f\Big(\frac{x -{\bf c}_n tn^a}{n}\Big)\big(\omega_x(tn^a)-\mf v_n(\beta,0)\big)\bigg\}\bigg]. \label{eq:vtilde}
  \end{equation}

\begin{theo}[Volume fluctuations in the time scale $tn^a$ with $a>1$] \label{theo:volume2}

\quad

 Let us fix $f,g \in \cC_c^\infty(\R)$, and $t>0$. 
Let \[a^* ( b)= \begin{cases} 1+ 2b\; , & b\in [0,\frac12], \\ 2\; , & b \ge \frac12.\end{cases}\] For any $b> 0$ two cases hold:
\begin{enumerate}[1.]
\item If $a < a^*(b)$, then
\[ \lim_{n\to \infty} \widetilde{\bf V}_t^n(f) = \lim_{n\to\infty} \widetilde{\bf V}_0^n(f)=\frac{1}{\beta}\iint_{\bb R^2} f(u)g(v)   \; du dv.\]
\item If $b > \frac12$ and $a=a^*(b)=2$ then
 \[ \lim_{n\to \infty} \widetilde{\bf V}_t^n(f) = \frac{1}{\beta}\iint_{\bb R^2} f(u)g(v) \tilde P_t(u-v)  \; du dv,
\] where $\{\tilde P_t\, ; \,t\ge 0\}$ is the semi-group generated by the Laplacian operator $\Delta$. 
\end{enumerate}
\end{theo}

These two theorems establish the following picture which is also summarized in Figure  \ref{fig:summer1}:

\begin{itemize} 
\item In the time scale $tn^a$, $a<1$, the volume field does not evolve.

 \smallskip
 
\item In the hyperbolic time scale $tn$ ($a=1$), the initial fluctuations of the volume field are transported with velocity $-2$.

\smallskip

\item  We then define a new volume field in a frame moving at velocity ${\mathbf c}_n=-2 -6\chi_n \gamma_n$ which takes into account the first order term in $\gamma_n$ of the sound velocity. The new field does not evolve up to time scale $tn^{2}$ for $b>\frac12$ and up to the time scale $tn^{2b+1}$ for $0\le b \le \frac12$. 
\smallskip

\item For $b>\frac12$, in the diffusive time scale, the evolution is driven by a heat equation.
\end{itemize}

\begin{figure}[H]e
\begin{tikzpicture}[scale=0.2]
\draw[->,>=latex] (0,0) -- (23,0);
\draw[->,>=latex] (0,0) -- (0,24);
\draw (0,23) node[left]{$a$};
\draw (22,0) node[below]{$b$};
\draw (0,0) node[below]{$0$};
\draw (0,0) node[left]{$0$};
\draw (5,0) node[below]{$\frac12$};
\draw (0,20) node[left]{$2$};
\draw (0,10) node[left]{$1$};
\draw[-,=latex,red,line width=2] (5,20) -- (20, 20) node[midway,above,sloped] {\bf{heat eq.}};
\fill[Gray] (0,0) -- (0,10) -- (5,20) -- (20,20) -- (20,0) -- cycle;
\draw (12,6.67) node{{\textbf{{{No evolution}}}}};
\draw[-,=latex, dashed] (0,20) -- (5,20);
\draw[-,=latex, dashed] (5,0) -- (5,20);
\draw[-,=latex, dashed] (0,20) -- (5,20) node[midway,above,sloped] {\bf{???}};
\fill[Gray!60] (0,10) -- (0,20) -- (5,20) -- cycle;
\end{tikzpicture}
\caption{Volume fluctuations: value of the time scale exponent $a$ as a function of the anharmonicity exponent $b$. \label{fig:summer1}}
\end{figure}
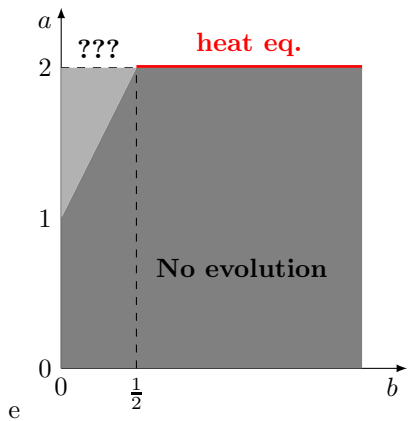

\begin{rem} 
\label{rem:vf}
For $0 \le b \le \frac12$ we conjecture in fact that the evolution is trivial up to the time scale $tn^{2}$ (a proof that there is no evolution in the light gray zone is thus missing). Our conjecture is supported by the following consideration : for $b=0$, i.e.~$\gamma_n$ of order one, according to Spohn's nonlinear fluctuating hydrodynamics theory \cite{S,SS} and the computations of Appendix \ref{app:prediction}, the fluctuations of the volume\footnote{Recall that since $\tau=0$ the sound mode coincides with the volume.} field should still belong to the diffusive universality class. Therefore, at $b=0$, the time scale for which the sound evolution takes place should be $a=2$. Assuming that the exponent $a:=a^*(b)$ of the time scale on which evolution of the sound mode occurs is continuous and linear in $b\in[0,\frac12]$, we would get that $a^*(b)=2$ for $b\in [0,\frac12]$. 
\end{rem}

\begin{theo}[Energy fluctuations]\label{theo:energy}

\quad

 Let us fix $f,g \in \cC_c^\infty(\R)$, and $t>0$. We have the following two cases:
\begin{enumerate}[1.]
\item If  $a< \frac32$ and $b> \frac14$, then the macroscopic energy fluctuation field does not evolve:  
\[ \lim_{n\to \infty} \bE_t^n(f) = \lim_{n\to\infty} \bE_0^n(f)=\frac{2}{\beta^2}  \iint_{\R^2} f(u)g(v)  \; du dv.\]
\item If $a = \frac32$ and $b > \frac14$, then
\[ \lim_{n\to\infty} \bE_t^n(f) = \frac{2}{\beta^2}  \iint_{\R^2} f(u)g(v)  P_t(u-v) \; du dv,\]
where $\{P_t\,;\, t\geq 0\}$ is the semi-group generated by the infinitesimal generator of an asymmetric $3/2$-stable L\'evy process
\begin{equation}\label{eq:opL}
{\bf L}:=-\tfrac{1}{\sqrt 2} \big( (-\Delta)^{\frac34} - \nabla (-\Delta)^{\frac14} \big).
\end{equation}
\end{enumerate}
\end{theo}

Note that the operator $\bf L$ in \eqref{eq:opL} is the same as in \cite{BGJ}, which corresponds to the case $\gamma_n=0$. This theorem shows that if the nonlinearity is sufficiently weak, i.e. $\gamma_n= o(n^{-1/4})$, then the energy fluctuation field starts to evolve only in the time scale $tn^{3/2}$ and that in this time scale its evolution is the same as in the linear case ($\gamma_n=0$). Similarly to what we explained in Remark  \ref{rem:vf} we expect that the result remains valid for $b\in [0,\frac14]$, see Figure \ref{fig:winter0}.

\begin{figure}[H]
\begin{tikzpicture}[scale=0.2]
\draw[->,>=latex] (0,0) -- (23,0);
\draw[->,>=latex] (0,0) -- (0,24);
\draw (0,23) node[left]{$a$};
\draw (22,0) node[below]{$b$};
\draw (0,0) node[below]{$0$};
\draw (0,0) node[left]{$0$};
\draw (5,0) node[below]{$\frac14$};
\draw (0,20) node[left]{$3/2$};
\draw[-,=latex,red,line width=2] (5,20) -- (20, 20) node[midway,above,sloped] {\bf{frac. heat eq.}};
\fill[Gray!60] (0,0) -- (5,0) --(5,20)--(0,20) -- cycle;
\fill[Gray] (5,0) -- (5,20) -- (20,20) -- (20,0)  -- cycle;
\draw (12,6.67) node{{\textbf{{{No evolution}}}}};
\draw[-,=latex, dashed] (0,20) -- (5,20);
\draw[-,=latex, dashed] (5,0) -- (5,20);
\draw[-,=latex, dashed] (0,20) -- (5,20) node[midway,above,sloped] {\bf{???}};
\end{tikzpicture}

\caption{Energy fluctuations: value of the time scale exponent $a$ as a function of the anharmonicity exponent $b$.}\label{fig:winter0} 
\end{figure}
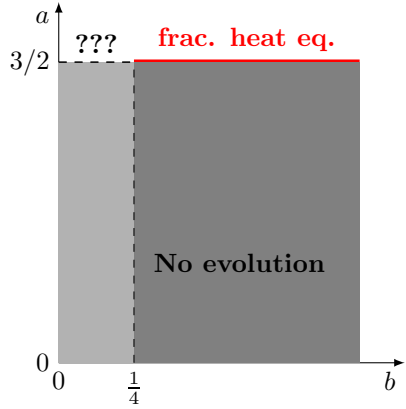

\subsection{Introduction of auxiliary fields}

In order to explain the proofs of our main theorems we need to introduce some auxiliary fields. From now on, for the sake of simplicity we assume $\beta=1$. The general case $\beta>0$ can be easily deduced from this by performing a change of variables. 

 Let $\kappa_n:=\kappa(\gamma_n)$ be the constant given ahead by \eqref{eq:kappann}, which satisfies $\kappa_n \to 3$ as $\gamma_n\to 0$ (see Section \ref{ssec:hermite}). Fix $g \in \cC_c^\infty(\R)$. First we define the bidimensional correlation fields, for any $h\in\cC_c^\infty(\R^2)$, as follows
\begin{align*}
\bQ_t^{2,n}(h) & := \E\bigg[ \cE_0^n(g)  \; \bigg\{\frac{1}{n}\sum_{x\neq y}h\Big(\frac{x}{n},\frac{y}{n}\Big) \omega_x\omega_y(tn^a)\bigg\}\bigg], \\
\bQ_t^{4,n}(h) & := \E\bigg[ \cE_0^n(g)  \; \bigg\{\frac{1}{n}\sum_{x\neq y}h\Big(\frac{x}{n},\frac{y}{n}\Big) (\omega_x^3-\kappa_n\omega_x)\omega_y(tn^a)\bigg\}\bigg] ,\\
\bQ_t^{6,n}(h) & := \E\bigg[ \cE_0^n(g)  \; \bigg\{\frac{1}{n}\sum_{x\neq y}h\Big(\frac{x}{n},\frac{y}{n}\Big) (\omega_x^3-\kappa_n\omega_x)(\omega^3_y-\kappa_n\omega_y)(tn^a)\bigg\}\bigg].
\end{align*}
%
We  point out that $\bQ_t^{2,n}, \bQ_t^{4,n}, \bQ_t^{6,n} $ do not depend on the values of $h$ at the diagonal $\{ (x,y)\in\bb Z^2\,;\, x = y \}$. 
We also define, for $f\in \cC_c^\infty(\R)$, the auxiliary field
\begin{equation*}
\bE_t^{4,n}(f)  :=
\E\bigg[ \cE_0^n(g)  \; \bigg\{\frac{1}{\sqrt n}\sum_{x\in\Z}f\Big(\frac{x}{n}\Big) \big( \omega_x^4(tn^a) -\big\langle \omega_x^4 \big\rangle_{\beta,0,\gamma_n} \big)\bigg\}\bigg].
\end{equation*}
 Let us introduce another one dimensional  field, related to the evolution of the volume correlation field as follows: it is defined for $f\in\cC_c^\infty(\R)$ as
\begin{equation}
\mathbf {V}_t^{3,n}(f)  := \E\bigg[ \mathcal{V}_0^n(g)  \; \bigg\{\frac{1}{\sqrt n}\sum_{x\in \bb Z }f\Big(\frac{x}{n}\Big) \;  \omega_x^3 (tn^a)\bigg\}\bigg]. \label{eq:v3}
\end{equation}
Finally, similarly to \eqref{eq:vtilde}, we define $\widetilde{\bf V}_t^{3,n}(f)$ in a reference frame which moves at velocity ${\bf c}_n$. 
 
\section{Estimate tools} \label{sec:esti}

\subsection{Two major inequalities}

We state  here two  inequalities, that are going to be largely used in what follows, in order to estimate the limit behavior of the main correlation fields ${\bf E}_t^n$, ${\bf V}_t^n$ and ${\bf Q}_t^{2,n}$. 

\subsubsection{Cauchy-Schwarz inequality} The following \textit{a priori} bounds  are consequences of the Cauchy-Schwarz inequality and stationarity of the process: for any $f\in\cC_c^\infty(\R)$  and $h \in \cC_c^\infty(\R^2)$,
\begin{align}
|\bE_t^n(f) | & \leq C(g)\; \|f\|_{2,n} \label{eq:CS-S}\\
|\bE_t^{4,n}(f) | & \leq C(g)\; \|f\|_{2,n} \label{eq:CS-S4}\\
|\bQ_t^{\lambda,n}(h) | & \leq C_\lambda(g)\; N_n^{\neq}(h), \qquad \text{for any } \lambda \in \{2,4,6\},\label{eq:CS-Q}\\
|\mathbf{V}_t^n(f) | & \leq C(g)\; \|f\|_{2,n} \label{eq:CS-V}\\
|\mathbf{V}_t^{3,n}(f) | & \leq C(g)\; \|f\|_{2,n}. \label{eq:CS-V3}
\end{align} 
where $C_\lambda, C$ are positive constants that only depend on the fixed test function $g$, and $\|\cdot\|_{2,n}$, $N_n^{\neq}(\cdot)$ are the norms defined in \eqref{eq:normh}. 


\subsubsection{Kipnis-Varadhan inequality and $\mathcal{H}_{-1,z}$ norms} A more refined bound is provided by \cite[Lemma 2.4]{MR2952852} as follows: for any function $\psi \in \mathbb{L}^2(\nu_{\beta,\tau,\gamma_n})$ we have
\begin{equation}\E\bigg[\bigg(\int_0^T \psi\big(\omega(tn^a)\big) dt\bigg)^2\bigg] \leq CTn^{-a} \; \big\|\psi(\cdot)\big\|^2_{-1, n^{-a}T^{-1}}
\label{eq:estim}
\end{equation}
where $\|\cdot\|_{-1,z}^2$ is defined in \eqref{eq:norm-1}.
 Then we have, for instance,
\begin{align}
\bb E\bigg[\bigg(\int_0^T \bE_t^n(f) dt \bigg)^2\bigg] & \leq C(g) Tn^{-a} \big\| {\mc E}^n (f) \big\|^2_{-1, n^{-a}T^{-1}}\label{eq:H-1-S}
\end{align} 
The goal of the next section is to present a general method to compute the $\mc H_{-1,z}$-norms that appear at the right hand side of \eqref{eq:estim} and \eqref{eq:H-1-S},  and to estimate them with the sharpest possible bounds.

\subsection{Orthogonal polynomials and $\mathcal{H}_{-1,z}$ norms}
\label{ssec:hermite}

Recall that for simplicity, we assume that $\beta=1$. In this section we drop the index $(\beta,\tau)$ from the notations. We denote $\nu_\gamma:=\nu_{1,0,\gamma}$ and the scalar product on ${\bb L}^2 (\nu_{1,0,\gamma})$ is simply denoted by $\langle \cdot ,\cdot \rangle$.

\subsubsection{Construction of the orthogonal polynomials}

Let $\{H_n\,;\, n\in \bb N\}$ be the sequence of orthogonal polynomials with respect to the following probability measure  on $\bb R$:
\[dW (u):= Z_\gamma^{-1}  (1,0) \exp ( - e_{\gamma} (u)) du\]
obtained by a Gram-Schmidt procedure from the basis $(1,u,u^2, \ldots)$. The average of a function $f(u)$ with respect to $W$ is denoted by $\langle f\rangle_W$. The first polynomials are given by 
\[\begin{array}{ll}  H_0 (u)=1, & \displaystyle H_1 (u) =u, \\
\displaystyle H_2 (u) = u^2 - \langle u^2 \rangle_W, &\displaystyle H_3 (u)= u^3 - \kappa(\gamma)\; u, 
\end{array}\]
where 
\begin{equation}
\label{eq:kappann}
\kappa(\gamma):= \frac{\langle u^4 \rangle_W}{\langle u^2 \rangle_W}.
\end{equation} 
Observe that $\kappa(\gamma) \to 3$ as $\gamma \to 0$.

We use here some ideas of \cite[Appendix 2]{BGJSS}.  Let us construct  a basis of $\mathbb{L}^2(\nu_{\gamma})$ constituted by multivariate polynomials by tensorization of the $H_k$'s. 
We denote by $\Sigma$  the set composed of configurations  $\sigma=\{\sigma_x\}_{x\in \Z} \in \N^{\Z}$ such that $\sigma_{x} \ne 0$ only for a finite number of $x$ and  \[\Sigma_k = \Big\{ \sigma \in \Sigma \; ; \; \sum_{x \in \Z} \sigma_x=k\Big\}.\] On the set of $k$-tuples $\bx:=(x_1, \ldots,x_k)$ of $\Z^k$, we introduce the equivalence relation $\bx \sim \by$ if there exists a permutation $p$ on $\{1, \ldots,k\}$ such that $x_{p(i)} =y_i$ for all $i \in \{1, \ldots,k\}$. The class of $\bx$ for the relation $\sim$ is denoted by $[\bx]$ and its cardinal by $c({\bf x})$.  Then the set of configurations of $\Sigma_k$ can be identified with the set of $k$-tuples classes for $\sim$ by the one-to-one application:
\begin{equation*}
[{\bf x}]=[(x_1,\ldots,x_k)] \in \Z^k/ \sim \; \rightarrow \sigma^{[{\bf x}]} \in \Sigma_k
\end{equation*}
where for any $y \in \Z$, $(\sigma^{[\bf x]})_y= \sum_{i=1}^k {\bf 1}_{y=x_i}$. We shall identify $\sigma \in \Sigma_k$ with the occupation number of  a configuration with $k$ particles, and $[\bf x]$ will correspond to  the positions of those $k$ particles. To any $\sigma \in \Sigma$, we associate the polynomial function $H_{\sigma}$ given by
\begin{equation*}
H_{\sigma} (\omega) = \prod_{x \in \Z} H_{\sigma_x} (\omega_x).
\end{equation*}
Then, the family $\left\{ H_{\sigma} \; ; \; \sigma \in \Sigma \right\}$ forms an orthogonal basis of ${\bb L}^2 (\nu_{\gamma})$ such that
\begin{equation}
\label{eq:prod hsigma}
\int H_\sigma (\omega)  \, H_{\sigma'} (\omega) \, d\nu_{\gamma} (\omega)  = {N}_{\gamma} (\sigma) \delta_{\sigma=\sigma'},\end{equation}
where $N_\gamma$ is a real-valued function and $\delta$ denotes the Kronecker function, i.e. $\delta_{\sigma=\sigma'}=1$ if $\sigma=\sigma'$ and zero otherwise.

A function $\Phi:\Sigma \to \R$ such that $\Phi(\sigma)=0$ if $\sigma \notin \Sigma_k$ is called a degree $k$ function. Thus, such a function is sometimes considered as a function defined only on $\Sigma_k$. A local function $\phi \in {\mathbb L}^2 (\nu_{\gamma})$ whose decomposition on the orthogonal basis $\{ H_{\sigma} \, ; \, \sigma \in \Sigma \}$ is given by $\phi=\sum_{\sigma} \Phi(\sigma)  H_{\sigma}$ is called of degree $k$ if and only if $\Phi$ is of degree $k$. A function $\Phi: \Sigma_k \to \R$ is nothing but a symmetric function $\Phi:\Z^k \to \R$ through the identification of $\sigma$ with $[\bx]$. We denote, with some abuse of notation,  by $\langle \cdot, \cdot \rangle$ the scalar product on $\oplus {\mathbb L}^2 (\Sigma_k)$, each $\Sigma_k$ being equipped with the counting measure. Hence, if $\Phi,\Psi:\Sigma \to \R$, we have
\begin{equation*}
\langle \Phi, \Psi \rangle = \sum_{k\geq 0} \sum_{\sigma \in \Sigma_k} \Phi_k (\sigma) \Psi_k (\sigma) = \sum_{k \geq 0} \sum_{\bx \in \Z^k} \frac{1}{c({\bf x})} \,  \Phi_k (\bx) \Psi_k (\bx),
\end{equation*}
with $\Phi_k, \Psi_k$ the restrictions of $\Phi,\Psi$ to $\Sigma_k$. 

The nice property of the generator $\mathcal{S}$ of the stochastic noise is that it can be nicely decomposed on the basis. If a local function $\phi \in {\bb L}^{2} (\nu_{\gamma})$ is written in the form $\phi =\sum_{\sigma \in \Sigma} \Phi(\sigma) H_{\sigma}$ then we have
\begin{equation*}
({\mc S} \phi) (\omega) = \sum_{\sigma \in \Sigma} ({\mf S} \Phi)(\sigma) H_{\sigma} (\omega)
\end{equation*}
with
\begin{equation}
\label{eq:SHermitepol}
({\mf S} \Phi)(\sigma) = \sum_{x \in \Z} ( \Phi(\sigma^{x,x+1}) - \Phi(\sigma)),
\end{equation}
where $\sigma^{x,x+1}$ is obtained from $\sigma$ by exchanging the occupation numbers $\sigma_x$ and $\sigma_{x+1}$.

\subsubsection{Estimates of $\mc H_{-1,z}$ norms}
Here we prove the following   lemma:


\begin{lem} \label{lem:two}
Let $F:\Z^2 \to \R$ be square-summable, namely $\sum_{x,y} F^2(x,y) < +\infty$, and assume that $F$ vanishes along the diagonal: $F(x,x)=0$ for any $x\in\Z$. Then, there exists $C >0$ such that, for any $(p,q) \in \bb N^2$ with $p\neq q$,  any $z >0$,   and any $\gamma \leq 1$,
\begin{equation}
\label{eq:h-1-two}
\bigg\|\sum_{x\neq y} F(x,y) H_p(\omega_x)H_q(\omega_y) \bigg\|_{-1,z}^2 \leq C  \iint_{[-\frac12,\frac12]^2} \frac{|\widehat{F}(k,\ell)|^2}{z+4\sin^2(\pi k)+4\sin^2(\pi \ell)} dk d\ell,
\end{equation}
where \[ \widehat{F}(k,\ell)=\sum_{(x,y)\in \Z^2} F(x,y) e^{2i\pi (k x+\ell y)}.\]
\end{lem}

\begin{proof}
Let $p,q \geq 1$ be fixed. We define the subset $\chi_{p,q}$ included in $\Sigma_{p+q}$ as 
\[
\chi_{p,q}:=\big\{\sigma \in \Sigma \; ; \; \exists \; x\neq y, \; \sigma=p\delta_x + q\delta_y\big\}.
\]
Then, the function $\phi\in\mathbb{L}^2(\nu_\gamma)$  under interest in the left hand side of \eqref{eq:h-1-two} is written in the form $\phi=\sum_{\sigma \in \chi_{p,q}} \Phi(\sigma)H_\sigma$ where $\Phi(p \delta_x+q \delta_y)=F(x,y)$.
The set $\chi_{p,q}$ has the following stability property: for any $\sigma \in \chi_{p,q}$ and $x\in\bb Z$, we have $\sigma^{x,x+1}\in\chi_{p,q}$. {For any local function  $\psi= \sum_{\sigma \in \Sigma} \Psi(\sigma) H_{\sigma}$ we write $\psi=\psi_{p,q} + \psi'$ where
$$\psi_{p,q} = \sum_{\sigma \in \chi_{p,q}} \Psi(\sigma) H_{\sigma}, \quad \psi'= \sum_{\sigma \notin \chi_{p,q}}\Psi(\sigma) H_{\sigma}.$$
We have
\begin{equation*}
({\mc S} \psi) (\omega) = \sum_{\sigma \in \Sigma} ({\mf S} \Psi)(\sigma) H_{\sigma} (\omega) 
\end{equation*}
where ${\mf S}$ is defined by \eqref{eq:SHermitepol}. Recall that ${\mc D} (\psi)$ has been defined in \eqref{eq:diric} as  the Dirichlet form of $\psi$. A simple computation based on the orthogonality of the polynomials $H_\sigma$ and the stability property of $\chi_{p,q}$ shows that 
\begin{equation*}
\begin{split}
&{\mc D} (\psi) = \langle \psi  ,  \; -{\mc S} \psi\rangle= \langle \Psi  , \; {- \mf S} \Psi\rangle  \\
& =\tfrac{1}{2} \,  \sum_{x \in \Z} \sum_{\sigma \in \chi_{p,q}} N_{\gamma} (\sigma) \; \left[ \Psi(\sigma^{x,x+1}) - \Psi(\sigma) \right]^2 \\&+ \tfrac{1}{2} \sum_{x \in \Z} \sum_{\sigma \notin \chi_{p,q}} N_{\gamma} (\sigma) \; \left[ \Psi(\sigma^{x,x+1}) - \Psi(\sigma) \right]^2\\
&= {\mc D} (\psi_{p,q}) +{\mc D} (\psi')
\end{split}
\end{equation*}
where $N_\gamma$ has been defined in \eqref{eq:prod hsigma}. Moreover we have that
\begin{equation}
\langle \phi, \psi \rangle = \langle \phi, \psi_{p,q} \rangle, \quad \langle \psi , \psi\rangle = \langle \psi_{p,q}, \psi_{p,q} \rangle + \langle \psi',\psi' \rangle.  
\end{equation}
Recall from \eqref{eq:norm-1} that
\begin{equation*}
\begin{split}
\big\langle \phi , \;  \big(z-\S\big)^{-1}\phi\big\rangle & =\sup_{\psi} \Big\{2\big\langle \phi \;, \;   \psi\big\rangle-z\big\langle \psi^2\big\rangle- {\mc D} (\psi) \Big\}\\
&= \sup_{\psi_{p,q}, \psi'} \Big\{2\big\langle \phi \;, \;   \psi_{p,q}\big\rangle-z\big\langle \psi_{p,q}^2\big\rangle- {\mc D} (\psi_{p,q}) -z\big\langle (\psi')^2\big\rangle- {\mc D} (\psi')\Big\}.
\end{split}
\end{equation*}
 Therefore, we have that
\[
\big\langle\phi , \;  \big(z-\S\big)^{-1}\phi\big\rangle=\sup_{\psi} \Big\{2\big\langle \phi \;,\;   \psi\big\rangle-z\big\langle \psi^2\big\rangle- \big\langle \psi \; (-\S)\psi\big\rangle\Big\},
\]
where the supremum is now restricted to local functions $\psi$  which are on the form $\psi=\sum_{\sigma\in\chi_{p,q}} \Psi(\sigma)H_\sigma$. }As a result, for any $z >0$ and as $\gamma\to 0$,
\begin{align}
\big\langle & \psi\; , \;  \big(z-\S\big)^{-1}\psi\big\rangle \notag \\
&= \sup_{\Psi:\chi_{p,q} \to\R} \bigg\{2\sum_{\sigma \in\chi} \Phi(\sigma)\Psi(\sigma)  N_{\gamma} (\sigma) \;-z\sum_{\sigma\in\chi}\Psi^2(\sigma)  N_{\gamma} (\sigma) \notag\\
&\hspace{3cm}  - \sum_{x \in \Z} \sum_{\sigma \in \chi_{p,q} }  N_{\gamma} (\sigma) \; \left(\Psi(\sigma^{x,x+1}) -\Psi(\sigma) \right)^2\bigg\} \notag \\
&\approx \; \sup_{\Psi:\chi_{p,q} \to\R} \bigg\{2\sum_{\sigma \in\chi} \Phi(\sigma)\Psi(\sigma)-z\sum_{\sigma\in\chi}\Psi^2(\sigma)- \sum_{x \in \Z} \sum_{\sigma \in \chi_{p,q}} \left(\Psi(\sigma^{x,x+1}) -\Psi(\sigma) \right)^2\bigg\} \label{eq:vphisup}
\end{align}
because $N_{\gamma} (\sigma)$ is constant for $\sigma \in \chi$ and \[\lim_{\gamma \to 0} N_{\gamma} (p\delta_0 + q\delta_1)= N_0 (p\delta_0 +q\delta_1)>0,\] and the last estimate in \eqref{eq:vphisup} follows by an explicit computation: let us define \begin{align*}
\Delta_0\; &= \left\{ (x,x) \, ; \, x \in \Z \right\},\\
\Delta_{+} & = \left\{ (x,x+1) \in \Z^2 \, ; \,  x\in\Z\right\},\qquad \Delta_+^+ = \left\{ (x,y) \in \Z^2 \, ; \, y \geq x+1\right\},\\
\Delta_{-} & = \left\{ (x,x-1) \in \Z^2 \, ; \,  x\in\Z\right\},\qquad \Delta_-^-  = \left\{ (x,y) \in \Z^2 \, ; \, y \leq x-1\right\}.
\end{align*}
For any $\psi=\sum_{\sigma \in \chi_{p,q}} \Psi (\sigma) H_\sigma$, we define $G:=G_{\psi}: \Z^2 \to \R$ as
\begin{equation} {G}(x,y):=\begin{cases} \Psi (p\delta_x+q\delta_y) & \text{ if } x\neq y,\\
~\\
 \displaystyle \frac{1}{4}\sum_{\substack{\be \in \Z^2\\|\be|=1}} G ((x,y)+\be) & \text{ if } x=y.
\end{cases}\label{eq:hatG}\end{equation}
A straightforward computation shows that
\begin{align*}
\sum_{x \in \Z} \sum_{\sigma \in \chi_{p,q}} \left(\Psi(\sigma^{x,x+1}) -\Psi(\sigma) \right)^2  = & \sum_{\substack{|\be|=1\\y\notin\{x-1,x,x+1\}}} \left( G((x,y)+\be) -\ G(x,y) \right)^2\\
 & + \sum_{(x,y)\in\Delta_-}\sum_{\substack{|\be|=1\\(x,y)+\be \in \Delta_-^-}}\left( G((x,y)+\be) - G(x,y) \right)^2\\
 & + \sum_{(x,y)\in\Delta_+}\sum_{\substack{|\be|=1\\(x,y)+\be \in \Delta_+^+}}\left( G((x,y)+\be) -G(x,y) \right)^2\\
 & + \sum_{(x,y)\in\Delta_-\cup\Delta_+}\left( G(x,y) - G(y,x) \right)^2.
\end{align*}
We denote by ${\bb D}$ the Dirichlet form of a symmetric simple random walk on $\Z^2 -\Delta_0$ where jumps from $u\in\Delta_{+}$ (resp. $\Delta_-$) to its symmetric $\overline u \in \Delta_-$ (resp. $\Delta_+$) with respect to $\Delta_0$  have been added. Let $\bb D_0$ be defined for every function ${G}:\Z^2\to\R$ as
\[
\bb D_0({G})= \sum_{(x,y)\in\Z^2} \sum_{|\be|=1} \left( G((x,y)+\be) -\ G(x,y) \right)^2.
\]
It has been proved in \cite{BG14} the following 
\begin{lem}\label{lem:dirich}
There exist $C,C'>0$ such that, for any $\psi=\sum_{\sigma \in \chi_{p,q}} \Psi (\sigma) H_\sigma$,
\[
\mc D(\psi) \geq C \; \bb D({G_\psi}) \geq C' \;\bb D_0({G_\psi}),
\]
where ${G_\psi}$ is defined in \eqref{eq:hatG}.
\end{lem}

From \eqref{eq:vphisup} and Lemma \ref{lem:dirich}, we have 
\begin{multline*}
\big\langle \phi\; , \; \big(z-\S\big)^{-1} \phi\big\rangle \lesssim \sup_{G} \bigg\{  2\sum_{(x,y) \in \Z^2} F(x,y) G(x,y)  - z \sum_{(x, y) \in \Z^2} G^2 (x,y)  \\
  - \sum_{|e| =1}\sum_{(x,y) \in \Z^2 }  \big( G((x,y)+e) - G(x,y) \big)^2\bigg\}
\end{multline*}
where the supremum is now taken over all local functions $G:\Z^2 \to \R$. Then, by Fourier transform, the last supremum is equal to 
\begin{equation}
\label{eq:intpsi}
\iint_{[-\frac12,\frac12]^2} \frac{|{\widehat F}(k,\ell)|^2}{z + 4  \sin^{2} (\pi k) + 4 \sin^{2} (\pi \ell) } dkd\ell.
\end{equation}
\end{proof}

\section{Proof of the macroscopic fluctuations for the volume field}
\label{sec:volumeevolution}

In this section we establish  Theorems \ref{theo:volume} and  \ref{theo:volume2}.

We are going to write in a convenient way the differential equations governing the evolution of the fluctuation fields. Recall that $\mathbf{V}_t^{3,n}$ has been defined in \eqref{eq:v3}. 


\begin{prop}\label{prop:vol_decomposition}
 For any function $f\in\cC_c^\infty(\R)$, 
\begin{equation*}
\frac{d}{dt} \mathbf{V}_t^n(f)  =  \mathbf{V}_t^n\Big( 2 n^{a-2} \Delta_n f   -  2n^{a-1}\;\nabla _n f   \Big)+ {\gamma_n}\mathbf{V}^{3,n}_t \Big( n^{a-2}\Delta_n f \; -\; 2 n^{a-1} \nabla_n f\Big).
\end{equation*}
\end{prop}
The proof of this proposition is given in Appendix \ref{app:equadiff}. 

\subsection{Fluctuations in the  time scale $a\le 1$}

 From Proposition \ref{prop:vol_decomposition}, by using \eqref{eq:CS-V} and \eqref{eq:CS-V3}, we obtain in the case $a<1$ that $\mathbf{V}_t^n(f)=\mathbf{V}_0^n(f)  + o(1)$ and no evolution holds.
 
 Taking the hyperbolic time scale ($a=1$),  we get that the evolution of the volume field is such that for any $t \in [0, T]$
\begin{equation*}
\mathbf{V}_t^n(f)-\mathbf{V}_0^n(f)  =    -2\int_0^t \mathbf{V}_s^n(\nabla _n f) ds+o(\gamma_n).
\end{equation*}
Thus, in the hyperbolic time scale, the initial fluctuations are transported with a velocity $-2$, and Theorem \ref{theo:volume} is proved.  This result seems to indicate that the sound velocity is $-2$. In fact this is not totally correct since a more accurate value of the sound velocity is given in Appendix \ref{app:prediction} and it is equal to $-2$ only at $0$-th order. Taking into account the first order correction in $\gamma_n$ in the sound mode velocity is fundamental in order to establish the next results.

\subsection{Triviality of the fluctuations up to the time scale $tn^a$ with  $a < a^*(b)$}

In this section, we consider the new field $\widetilde{\bf V}_t^{n}$ defined in \eqref{eq:vtilde}. The time evolution equation given by Proposition \ref{prop:vol_decomposition} can be easily rewritten in the new reference frame as:  
\begin{equation*}
\begin{split}
\frac{d}{dt}\widetilde{\bV}_t^n(f) & =  \widetilde{\bV}_t^n\Big( 2 n^{a-2} \Delta_n f   -  2n^{a-1}\;\nabla _n f   -{\bf c}_n n^{a-1} f^{\prime}\Big)\notag\\
& \quad + {\gamma_n}\widetilde{\bV}^{3,n}_t \Big( n^{a-2}\Delta_n f \; -\; 2 n^{a-1} \nabla_n f\Big).\notag
\end{split}
\end{equation*}
Observe that 
\begin{equation*}
\begin{split}
&\nabla_n f \big(\tfrac{x -{\bf c}_n tn^a}{n}\big)=f' \big(\tfrac{x -{\bf c}_n tn^a}{n}\big) +\tfrac{1}{2n} f'' \big(\tfrac{x -{\bf c}_n tn^a}{n}\big) +\mc O (n^{-2}),\\
&\Delta_n f \big(\tfrac{x -{\bf c}_n tn^a}{n}\big)=f'' \big(\tfrac{x -{\bf c}_n tn^a}{n}\big) +\mc O (n^{-1}).
\end{split}
\end{equation*}
Therefore, for $a \le 2$, by using \eqref{eq:CS-V} and \eqref{eq:CS-V3} we get, for $t \in [0, T]$,
\begin{equation}
\label{eq:tildevolumeevol1}
\widetilde{\bV}_t^n(f) - \widetilde{\bV}_0^n(f) = \int_0^t \left\{  \widetilde{\bV}_s^n\Big( n^{a-2} f'' +6 \chi_n \gamma_n n^{a-1} f^{\prime}\Big) -2 n^{a-1} {\gamma_n}\widetilde{\bV}^{3,n}_s (f') \right\} ds +\varepsilon_n (t)
\end{equation}
where we have $\EE \big[(\sup_{t \le T} \varepsilon_n (t) )^2\big]\to 0$ as $n\to \infty$.  
Observe first that by \eqref{eq:CS-V} the term $\int_0^t \widetilde{\bV}_s^n (n^{a-2} f'') ds$ vanishes in ${\bb L}^2(\bb P)$ as soon as $a<2$ but remains if $a=2$.

Let $\tau_x: \Omega \to \Omega$ be the shift operator defined by $(\tau_x \omega)_z= \omega_{x+z}$, $z \in \Z$, whose action is extended to functions $\varphi:\Omega \to \RR$ by $(\tau_x \varphi) (\omega) =\varphi (\tau_x \omega)$.  
Note now that
\begin{align}
\omega_x^3 
& = (-\mc L_{\gamma_n})(\omega_x^2 \omega_{x+1}) \label{eq:omega3-0}\\ 
& \quad +3\omega_x(\omega_{x+1}^2-\chi_n) +  (\omega_x^2-\chi_n) (2\omega_{x+2}-3\omega_{x+1}) + (\omega_{x-1}^2-\chi_n)\omega_{x+1} \label{eq:omega3-11} \\
& \quad - 2 \omega_{x-1} \omega_x \omega_{x+1} \label{eq:omega3-1} \\ 
& \quad + \chi_n(3\omega_x+2\omega_{x+2}-2\omega_{x+1}) \label{eq:omega3-2} \\
& \quad + \gamma_n (\tau_x\psi)(\omega)\label{eq:omega3-3}
\end{align}
where $\psi:\Omega\to\R$ is a local function such that $\sup_{n} \langle \psi \rangle_{\beta,0,\gamma_n} < +\infty$. Using this in \eqref{eq:tildevolumeevol1} to rewrite the term  $n^{a-1} {\gamma_n}\widetilde{\bV}^{3,n}_s(f')$ we obtain several contributions. 

\begin{enumerate}[I)]
\item The contribution of \eqref{eq:omega3-2} to $n^{a-1}\gamma_n\int_0^t {\widetilde{\bf V}}_s^{3,n}(f') ds$ gives
\[ 3n^{a-1}\gamma_n\chi_n\int_0^t {\widetilde {\bf V}}_s^n(f') \; ds + \varepsilon_n (t), \]
with \[\lim_{n \to \infty}  n^{a-1} \gamma_n\; \EE \Big[\sup_{t \le T} \varepsilon_n^2 (t)\Big] =0\] because $a\le 2$. With the constant prefactor $2$ this transport term cancels the term $ 6 \chi_n \gamma_n n^{a-1}\widetilde{\bV}_s^n (  f^{\prime})$ appearing in \eqref{eq:tildevolumeevol1}.

\medskip

\item The contribution of \eqref{eq:omega3-3} to  $n^{a-1}\gamma_n\int_0^t {\widetilde{\bf V}}_s^{3,n}(f') ds$  gives a smaller term, which by the Cauchy-Schwarz inequality \eqref{eq:CS-V3} is at most of order $\gamma_n^2 n^{a-1}$. Therefore, this term vanishes if $\gamma_n=o(n^{\frac{-a+1}2})$, i.e. $a<2b+1$. 

\medskip

\item The contribution of \eqref{eq:omega3-0} is null. Indeed, by Dynkin's formula, we have
\begin{equation*}
\begin{split}
 &n^{a-1}\gamma_n\int_0^t (-\mc L_{\gamma_n} \phi_s) (\omega (sn^a) )\; ds\\
 &= \cfrac{\gamma_n}{n}\; \phi_t (\omega(tn^a)) -  \cfrac{\gamma_n}{n} \; \phi_0 (\omega(0)) - \cfrac{\gamma_n}{n} \int_0^t (\partial_s \phi_s) (\omega (sn^a)) ds + {\mc M}_t^n(\phi)
 \end{split}
\end{equation*} 
with 
\[
\phi_s(\omega)=\frac{1}{\sqrt n}\sum_{x\in\Z} f'\Big(\frac{x -{\bf c}_n s n^a}{n}\Big)\omega_x^2\omega_{x+1}
\]
and ${\mc M}_t^n(\phi)$ a martingale. Observe that 
$$(\partial_s \phi_s)(\omega) = \frac{{\bf c}_n n^{a-1}}{\sqrt n}\sum_{x\in\Z} f^{''}\Big(\frac{x -{\bf c}_n s n^a}{n}\Big)\omega_x^2\omega_{x+1}.$$
The Cauchy-Schwarz inequality and stationarity imply that uniformly in $t \in [0,T]$,
\begin{equation*}
\lim_{n \to \infty} {\mathbb E} \bigg[ \bigg( \cfrac{\gamma_n}{n} \phi_t (\omega(tn^a)) -  \cfrac{\gamma_n}{n} \phi_0 (\omega(0)) - \cfrac{\gamma_n}{n} \int_0^t (\partial_s \phi_s) (\omega (sn^a)) ds\bigg)^2\bigg] =0
\end{equation*}
as soon as $\gamma_n \to 0$ (recall that $a\le 2$). 
The quadratic variation of the martingale is given by  
\begin{equation*}
\EE \big[ \left\langle {\mc M}^n(\phi) \right\rangle_t \big] = \gamma_n^2 n^{a-2} \int_0^t {\mc D} (\phi_s) ds,
\end{equation*}
where the Dirichlet form $\mc D$ has been defined in \eqref{eq:diric}. By the Cauchy-Schwarz inequality and stationarity we have that $\int_0^t {\mc D} (\phi_s) ds$ is uniformly bounded in $t \in [0,T]$ by a constant depending on $f$ and $\beta$. Therefore we have
\begin{equation*}
\lim_{n \to \infty} \EE \big[ \left\langle {\mc M}^n(\phi) \right\rangle_t \big] = 0
\end{equation*}
since $a\le 2$.

\medskip
\item Finally, it remains to treat the sum of two terms, which are of the form ($p\in\{1,2\}$)
    \begin{equation}
   n^{a-1}\gamma_n \int_0^t \frac{1}{\sqrt n} \sum_{x\in\bb Z}  f'\Big(\frac{x -{\bf c}_n sn^a}{n}\Big) (\omega_x^2-\chi_n)\omega_{x+p}(s n^a) ds, \label{eq:12}
\end{equation}
and
    \begin{equation}
   n^{a-1}\gamma_n \int_0^t \frac{1}{\sqrt n} \sum_{x\in\bb Z}  f'\Big(\frac{x -{\bf c}_n sn^a}{n}\Big)(\omega_{x-1}\omega_x\omega_{x+1})(sn^a) ds . \label{eq:13}
\end{equation}
Terms \eqref{eq:12} and \eqref{eq:13} are treated thanks to Lemma \ref{lem:marielle234} below.
\end{enumerate}

\begin{lem} 
\label{lem:marielle234}
Let $h: (s, x) \to h_s (x) \in \RR$ be a smooth bounded test function and $p\in\{1,2\}$. 
If $a<2b+1$ then 
\begin{align}
&\lim_{n \to \infty} {\mathbb E} \bigg[ \bigg( n^{a-1}\gamma_n \int_0^t \frac{1}{\sqrt n} \sum_{x\in\bb Z} h_{sn^{a-1}}\Big(\frac{x}{n}\Big) (\omega_x^2-\chi_n)\omega_{x+p}(sn^a) ds\bigg)^2 \bigg] =0, \label{eq:crossed} \\
&\lim_{n \to \infty} {\mathbb E} \bigg[ \bigg( n^{a-1}\gamma_n \int_0^t \frac{1}{\sqrt n} \sum_{x\in\bb Z} h_{sn^{a-1}}\Big(\frac{x}{n}\Big) ( \omega_{x-1}\omega_x\omega_{x+1})(sn^a) ds \bigg)^2 \bigg] =0. \label{eq:crossed2}  
\end{align}
\end{lem}

\begin{proof} We start with the proof of \eqref{eq:crossed}. Fix $p\in\{1,2\}$. 
First, we have to rewrite the $\mc H_{-1,z}$ estimate \eqref{eq:estim} in the case when the test function $\psi$ also depends on time $t$. More precisely, an easy modification of \cite[Lemma 3.9]{Seth}
gives: for any $\psi \in \bb L^2(\nu_{\beta,\tau,\gamma_n})$, 
\begin{equation}
\label{eq:estime-time}
\bb E\bigg[\bigg(\int_0^T \psi\big(tn^{a-1},\omega(tn^a)\big)dt\bigg)^2\bigg] \leq Cn^{-a}\int_0^T \big\|\psi(tn^{a-1},\cdot)\big\|_{-1,tn^{-a}}^2 \; dt.
\end{equation} 
From \eqref{eq:estime-time} and Lemma \ref{lem:two}, the term under the limit in \eqref{eq:crossed} is bounded  from above by 
\[ Cn^{2a-2} \; \gamma_n^2\; n^{-a} \int_0^t \iint_{[-\frac12,\frac12]^2} \frac{|\widehat{\varphi}_{sn^{a-1}}(k,\ell)|^2}{sn^{-a}+4\sin^2(\pi k)+4\sin^2(\pi \ell)} dk d\ell\; ds \]
where \[ \varphi_s(x,y)= \begin{cases}
0 & \text{ if } |x-y| \neq p, \\
n^{-\frac12} h_s(\frac x n) & \text{ if } |x-y|=p.
\end{cases} \]
Therefore,
$ |\widehat{\varphi_s}(k,\ell)|^2 = n \big|(\mc F_n h_s)(n(k+\ell))\big|^2, $
and we are reduced to estimate 
\begin{multline*}
n^{a-1} \gamma_n^2  \int_0^t \iint_{[-\frac12,\frac12]^2} \frac{\big|(\mc F_n h_{sn^{a-1}})(n(k+\ell))\big|^2}{sn^{-a}+4\sin^2(\pi k)+4\sin^2(\pi \ell)} dk d\ell\; ds\\
\lesssim n^{a-1} \gamma_n^2  \int_0^t \iint_{[-\frac12,\frac12]^2} \frac{1}{sn^{-a}+4\sin^2(\pi k)+4\sin^2(\pi \ell)} dk d\ell\; ds
\end{multline*}
because the discrete Fourier transform of $h_{sn^{a-1}}$ is bounded by a constant independent of $n$. The last integral is of order $n^{a-1} \gamma_n^2 \log n$ and goes to $0$ if $a-2b -1<0$. This proves \eqref{eq:crossed}.

Now we prove \eqref{eq:crossed2}. For that purpose we note that it is enough to bound each one of the following terms:
 \begin{align}
& {\mathbb E} \bigg[ \bigg( n^{a-1}\gamma_n \int_0^t \frac{1}{\sqrt n} \sum_{x\in\bb Z} h_{sn^{a-1}}\Big(\frac{x}{n}\Big) ( \omega_{x-1}\omega_x(\omega_{x+1}-\overrightarrow{\omega}^\ell_{x+1}))(sn^a) ds \bigg)^2 \bigg],\label{five_terms1}\\
& {\mathbb E} \bigg[ \bigg( n^{a-1}\gamma_n \int_0^t \frac{1}{\sqrt n} \sum_{x\in\bb Z} h_{sn^{a-1}}\Big(\frac{x}{n}\Big) (( \omega_{x-1}-\omega_{x-1-L})\omega_x\overrightarrow{\omega}^\ell_{x+1})(sn^a) ds \bigg)^2 \bigg],\notag\\
 &{\mathbb E} \bigg[ \bigg( n^{a-1}\gamma_n \int_0^t \frac{1}{\sqrt n} \sum_{x\in\bb Z} h_{sn^{a-1}}\Big(\frac{x}{n}\Big) ( (\omega_{x-1-L}-\overleftarrow{\omega}_{x-1-L}^L)\omega_x\overrightarrow{\omega}^\ell_{x+1})(sn^a) ds \bigg)^2 \bigg],\notag\\
& {\mathbb E} \bigg[ \bigg( n^{a-1}\gamma_n \int_0^t \frac{1}{\sqrt n} \sum_{x\in\bb Z} h_{sn^{a-1}}\Big(\frac{x}{n}\Big) (\overleftarrow{\omega}_{x-1-L}^L (\omega_x-\overleftarrow{\omega}_x^L)\overrightarrow{\omega}^\ell_{x+1})(sn^a) ds \bigg)^2 \bigg], \notag\\
& {\mathbb E} \bigg[ \bigg( n^{a-1}\gamma_n \int_0^t \frac{1}{\sqrt n} \sum_{x\in\bb Z} h_{sn^{a-1}}\Big(\frac{x}{n}\Big) ( \overleftarrow{\omega}_{x-1-L}^L\overleftarrow{\omega}_x^L\overrightarrow{\omega}^\ell_{x+1})(sn^a) ds \bigg)^2 \bigg], \label{five_terms2}
\end{align}
where for $k \in\mathbb N - \{0\}$ and $z\in\mathbb Z$ we define
$$\overrightarrow{\omega}_z^k=\frac{1}{k}\sum_{y=z+1}^{z+k}\omega_y.
 \quad \quad \overleftarrow{\omega}_z^k=\frac{1}{k}\sum_{y=z-k}^{z-1}\omega_y.$$
From an ad-hoc version of Lemma  6.2 of \cite{BGS}, when $L\geq \ell$,  the sum of the terms in the previous display are bounded from above by a constant times
\begin{equation*}
\gamma_n^2 n^{2a-3} \left\{   \frac{\ell^2}{n^{a-1}}+ \frac{L^2}{\ell n^{a-1}}+ \frac{L}{\ell n^{a-1}}
+ t\frac{n}{\ell L } \right\}\int_0^t\frac{1}{n}\sum_{x\in\mathbb{Z}}h_{sn^{a-1}}\Big(\frac{x}{n}\Big)^2\; ds.
\end{equation*}
We note that in order to bound the first term \eqref{five_terms1} the proof of the  Lemma  6.2 of \cite{BGS} relies on the one-block estimate that, for completeness, we prove in the next lemma. The last term \eqref{five_terms2} is estimated by using Cauchy-Schwarz inequality, stationarity and independence.  
Now, by choosing $\ell =\varepsilon \sqrt{n}$ and $L=\varepsilon n^{3/4}$, since $a<a^*(b)$, the previous expression vanishes as $n\to\infty$ and $\varepsilon\to 0$.
\end{proof}

\begin{lem}[One-block estimate] 
\label{one-block}
Fix $\ell\in \mathbb N$ and let $\varphi: \Omega \to \mathbb R$ be a local function which has mean zero w.r.t.~ $\nu_{\beta, 0, \gamma_n}$, and whose support does not intersect the set of points 
${\{0,\cdots,\ell\}}$.
There exists a constant $C>0$, such that for any $T>0$ and any function $h$ as in Lemma \ref{lem:marielle234}:
\begin{multline*}
\mathbb{E}\bigg[\bigg(\int_{0}^t \, \sum_{x\in\mathbb{Z}} h_{sn^{a-1}} \Big(\frac{x}{n}\Big) \tau_x\varphi (\omega({sn^a}))\big({\omega}_{x+1}-
\overrightarrow{\omega}_{x+1}^{\ell}\big) (sn^a) \, ds \bigg)^2\bigg] \\
\leq
 C\, \frac{\ell^2}{n^{a-1}}\langle\varphi^2\rangle_{\beta, 0,\gamma_n}\; \int_0^t\frac{1}{n}\sum_{x\in\mathbb{Z}}h_{sn^{a-1}}\Big(\frac{x}{n}\Big)^2ds.
 \end{multline*}
and
\begin{multline*}
{\mathbb E} \bigg[ \bigg( \int_0^t \sum_{x\in\bb Z} h_{sn^{a-1}}\Big(\frac{x}{n}\Big) \big( \omega_{x-1} (sn^a) -\omega_{x-1-\ell} (sn^a) \big) (\tau_x \varphi ) (\omega  (sn^a)) ds \bigg)^2 \bigg]\\
\le  C\, \frac{\ell^2}{n^{a-1}}\langle\varphi^2\rangle_{\beta, 0,\gamma_n}\; \int_0^t\frac{1}{n}\sum_{x\in\mathbb{Z}}h_{sn^{a-1}}\Big(\frac{x}{n}\Big)^2ds.
\end{multline*}
\end{lem}
\begin{proof}
We only prove the first display since the proof of the second one is similar. By \eqref{eq:estime-time} and  \eqref{eq:norm-1} we bound the previous expectation from above by a constant times
\begin{multline}\int_0^t\bigg\|\sum_{x\in\mathbb{Z}}h_{sn^{a-1}}(\tfrac{x}{n})\tau_x\varphi(\omega (tn^a) )\big({\omega}_{x+1}-
 \overrightarrow{\omega}^{\ell}_{x+1}\big) (sn^a) \bigg\|_{-1}^2 ds\\
=\int_{0}^t\;\sup_{f }\bigg\{2 \sum_{x\in\mathbb{Z}}h_{sn^{a-1}}(\tfrac{x}{n}) \left\langle \tau_x\varphi(\omega)\big({\omega}_{x+1}-
\overrightarrow{\omega}^{\ell}_{x+1}\big)f(\omega) \right\rangle_{\beta,0,\gamma_n} -n^a\mathcal{D}(f)\bigg\} ds, \label{varia}
\end{multline}
where the supremum is carried over local functions and $\mathcal D$ is the Dirichlet form defined in \eqref{eq:diric}. 
The term ${\omega}_{x+1}-\overrightarrow{\omega}^{\ell}_{x+1}$ can be written  as a sum of gradients as
\begin{equation*}
{\omega}_{x+1}-\overrightarrow{\omega}^{\ell}_{x+1}=\frac{1}{\ell}\sum_{y=x+2}^{x+1+\ell}\sum_{z=x+1}^{y-1}({\omega}_{z}-{\omega}_{z+1}).
\end{equation*}
By writing the average in \eqref{varia} as twice its half and in one of the terms performing the exchange $\omega$ to $\omega^{z,z+1}$, for which the measure $\nu_{\beta,0,\gamma_n}$ is invariant,  we write the term inside the supremum in \eqref{varia} as
\[
\sum_{x\in\mathbb{Z}}h_{sn^{a-1}}(\tfrac{x}{n})\; \frac{1}{\ell}\sum_{y=x+2}
^{x+1+\ell}\sum_{z=x+1}^{y-1}\; \left\langle  \tau_x\varphi(\omega)({\omega}_z-{\omega}_{z+1})(f(\omega)-f(\omega^{z,z+1})) \right\rangle_{\beta,0,\gamma_n}.\notag
\]
Now, applying  Young's inequality in the term inside brackets in the previous expression, for any choice of positive constants $B_x$, it is bounded from above by
\begin{align}&\frac{1}{\ell}\sum_{x\in\mathbb{Z}}\sum_{y=x+2}
^{x+1+\ell}\sum_{z=x+1}^{y-1} \big| h_{sn^{a-1}}(\tfrac{x}{n}) \big| \; \frac{B_x}{2}\left\langle (\tau_x\varphi(\omega))^2(\omega_z-\omega_{z+1})^2\right\rangle_{\beta,0,\gamma_n} \label{first}\\
+&\frac{1}{\ell}\sum_{x\in\mathbb{Z}}\sum_{y=x+2}
^{x+1+\ell}\sum_{z=x+1}^{y-1}\frac{ \big| h_{sn^{a-1}}(\tfrac{x}{n}) \big|}{2B_x} \; \left\langle \big(f(\omega)-f(\omega^{z,z+1})\big)^2\right\rangle_{\beta,0,\gamma_n}. \label{second}
\end{align}
Let $\varepsilon>0$. For the choice  $2B_x=\varepsilon^{-1} \ell n^{-a} |h_{sn^{a-1}}(\tfrac{x}{n})|$ and independence,   \eqref{first} is bounded from above by a constant times
\begin{equation*}
\frac{\varepsilon^{-1} }{n^a}\sum_{x\in\mathbb{Z}}\sum_{y=x+2}
^{x+1+\ell}\sum_{z=x+1}^{y-1}\big(h_{s n^{a-1}}(\tfrac{x}{n})\big)^2\langle\varphi^2\rangle_{\beta, 0,\gamma_n}\lesssim \, \cfrac{\langle\varphi^2\rangle_{\beta, 0,\gamma_n}}{\varepsilon} \, \frac{\ell^2}{n^{a-1}}\frac{1}{n}\sum_{x\in\mathbb{Z}}\big(h_{s n^{a-1}}(\tfrac{x}{n})\big)^2.
\end{equation*}
Finally, a simple computation shows that for the choice of $B_x$ that we have fixed above, the term  \eqref{second} 
is bounded from above by a constant times $\varepsilon n^a {\mc D} (f)$. By choosing $\varepsilon$ sufficiently small, this term counterbalances with  the term $n^a {\mc D} (f)$ in \eqref{varia}. This  ends the  proof.
\end{proof}

\subsection{Fluctuations in the diffusive time scale $tn^2$ with $a=2$ and $b \in (\frac12, + \infty)$}

The estimates above show that starting from \eqref{eq:tildevolumeevol1} and using the previous estimates with $a=2$, if $b>\frac12$, only the term $\int_0^t \widetilde{\bV}_s^n ( f'') ds$ survives. All the other terms give a zero contribution. This concludes the proof of Theorem \ref{theo:volume2}.

\section{Proof of the macroscopic fluctuations for the energy field}

\label{sec:te-eff}
In this section we prove Theorem \ref{theo:energy}.  
We   need to introduce some operators which are defined as follows. Let $f\in\cC_c^\infty(\R)$ and $h\in\cC_c^\infty(\R^2)$, then 
\begin{enumerate}[(i)]
\item  $\nabla_n f \otimes \delta :  \frac{1}{n} \Z^2 \to \R$ approximates the distribution $f\,'(x) \delta(x=y)$ as
\begin{equation*}
\label{eq:3.9}
\big(\nabla_n f \otimes \delta\big) \big( \tfrac{x}{n}, \tfrac{y}{n}\big) =
\begin{cases} \vphantom{\Big\{}\frac{n^2}{2}\big\{f\big(\frac{x+1}{n}\big) - f\big(\frac{x}{n}\big)\big\}; & y =x+ 1\\ \vphantom{\Big\{} \frac{n^2}{2}\big\{f\big(\frac{x}{n}\big) - f\big(\frac{x-1}{n}\big)\big\}; & y =x- 1\\
0; & \text{ otherwise.}
\end{cases}
\end{equation*}
\item  $\Delta_n h : \R^2 \to \R$ approximates the $2-d$ Laplacian of $h$ as
\begin{equation*}
\Delta_n h\big( \tfrac{x}{n}, \tfrac{y}{n}\big) = n^2\big\{ h\big( \tfrac{x+1}{n}, \tfrac{y}{n}\big)+h\big( \tfrac{x-1}{n}, \tfrac{y}{n}\big) +h\big( \tfrac{x}{n}, \tfrac{y+1}{n}\big)+ h\big( \tfrac{x}{n}, \tfrac{y-1}{n}\big) - 4 h\big( \tfrac{x}{n}, \tfrac{y}{n}\big)\big\}.
\end{equation*}
\item  $\nabla_n h : \R^2 \to \R$ approximates the gradient of $h$ along the diagonal\footnote{The reader will notice that we used also the notation $\nabla_n f$ to denote the usual discrete gradient of the function $f:\RR \to \RR$. No confusions are possible since the latter acts on functions defined on $\RR$.} as
\[
\nabla_n h\big( \tfrac{x}{n}, \tfrac{y}{n}\big) = \begin{cases} \vphantom{\Big\{}\frac n 2\big\{ h\big( \frac{x+1}{n}, \frac{x+1}{n}\big) -  h\big( \frac{x}{n}, \frac{x}{n}\big)\big\}; & y =x+1\\
\vphantom{\Big\{} \frac n 2 \big\{h\big(\frac{x}{n}, \frac{x}{n}\big)-h\big(\frac{x-1}{n}, \frac{x-1}{n}\big)\big\}; & y=x-1\\
0; & \text{ otherwise.}
\end{cases}
\]
\item  $A_n h: \R \to \R$ approximates the directional derivative $(-2,-2) \cdot \nabla h$ as
\begin{equation*}
A_n h\big( \tfrac{x}{n}, \tfrac{y}{n}\big) = n\; \big\{h\big( \tfrac{x}{n}, \tfrac{y-1}{n}\big)+ h\big( \tfrac{x-1}{n}, \tfrac{y}{n}\big)- h\big( \tfrac{x}{n}, \tfrac{y+1}{n}\big)-h\big( \tfrac{x+1}{n}, \tfrac{y}{n}\big)\big\}.
\end{equation*}
\item  $\cD_n h : \frac{1}{n} \Z \to \R$ approximates the directional derivative of $h$ along the diagonal as
\begin{equation*}
\cD_n h\big( \tfrac{x}{n} \big) = n\; \big\{ h \big(\tfrac{x}{n}, \tfrac{x+1}{n}\big) - h \big( \tfrac{x-1}{n}, \tfrac{x}{n} \big) \big\}.
\end{equation*}
\item  $\widetilde {\cD}_n h :\frac{1}{n} \Z^2 \to \R$ approximates the distribution $\partial_y h(x,x) \otimes \delta(x=y)$ as
\begin{equation*}
\widetilde{\cD}_n h \big(\tfrac{x}{n},\tfrac{y}{n}\big) =
\begin{cases}\vphantom{\Big\{} n^2 \big\{h\big(\frac{x}{n}, \frac{x+1}{n}\big)-h\big(\frac{x}{n}, \frac{x}{n}\big)\big\}; & y =x+1\\
\vphantom{\Big\{} n^2 \big\{h\big(\frac{x-1}{n}, \frac{x}{n}\big)-h\big(\frac{x-1}{n}, \frac{x-1}{n}\big)\big\}; & y=x-1\\
0; & \text{ otherwise.}
\end{cases}
\end{equation*}
%
%
\item  $B_n h:\frac{1}{n} \Z^2 \to \R$ is defined as
\begin{equation}
\label{eq:Bnh}
 B_n h\big(\tfrac{x}{n},\tfrac{y}{n}\big): = \sqrt{n}\;  \big\{h\big(\tfrac{x-1}{n},\tfrac{y}{n}\big)-h\big(\tfrac{x+1}{n},\tfrac{y}{n}\big) +\big({\bf 1}_{y=x+1} -{\bf 1}_{y=x-1}  \big) h\big(\tfrac{y}{n},\tfrac{y}{n}\big)\big\}.
 \end{equation}
 
 \end{enumerate}
In the following, we consider a function $h\in\cC_c^\infty(\R^2)$ which is symmetric, namely that satisfies $h(u,v)=h(v,u)$ for any $u,v \in \R$.

\begin{prop}
\label{prop:equadiff} For any function $f\in\cC_c^\infty(\R)$, and any symmetric function $h \in \cC_c^\infty(\R^2)$,
\begin{align}
\frac{d}{dt} \bE_t^n(f) & =      \bE_t^n\Big( n^{a-2}\;\Delta_n f\Big) - (1+\gamma_n\kappa_n)^2 \; \bQ_t^{2,n}\Big(n^{a-\frac32}\; \nabla_n f \otimes \delta\Big)\notag\\
& \quad  - 2  \gamma_n(1+\gamma_n\kappa_n)\;  \bQ_t^{4,n}\Big( n^{a-\frac32}\;\nabla_n f \otimes \delta\Big)  \notag \\
& \quad -   \; \gamma_n^2 \; \bQ_t^{6,n}\Big( n^{a-\frac32}\;\nabla_n f \otimes \delta\Big) ,\label{eq:S} \\
~\notag\\
\frac{d}{dt} \bQ_t^{2,n}(h) &= \bQ_t^{2,n}\Big(\mathbf{L}_n h \Big) -4  \bE_t^n\Big(n^{a-\frac32}\; \cD_n h\Big)  - \gamma_n \; \bE_t^{4,n}\Big(n^{a-\frac32}\; \cD_n h\Big) \notag \\ 
& \quad + 2 \bQ_t^{2,n}\Big(n^{a-2}\;\widetilde{\cD}_n h\Big)   +2 \gamma_n\;  \bQ_{t}^{4,n}\Big(n^{a-\frac12} B_n h\Big) \notag \\
 & \quad +  \gamma_n  \kappa_n \; \bQ_t^{2,n}\Big(n^{a-1}\nabla_n h\Big)\label{eq:Q}\end{align}
where the operator $\mathbf{L}_n$ is defined by
\begin{equation}\mathbf{L}_n :=n^{a-1}(1+\kappa_n \gamma_n) A_n + n^{a-2}  \Delta_n.\label{eq:operator}\end{equation}
%
%
%
\end{prop}

The proof of Proposition \ref{prop:equadiff} is given in Appendix \ref{app:equadiff}.  We remark that since the underlying model is nonlinear the time evolution of the pair (energy field ; quadratic field) is not closed and we have to deal with some hierarchy. This is the main difference with previous studies (\cite{BGJ,BGJSS,BGJS,JKO2}) whose success was very dependent of this closeness due to the linear interactions.
%

\subsection{Strategy of the proof of Theorem \ref{theo:energy}}

Assume $a=\frac32$.  The expressions \eqref{eq:S} and \eqref{eq:Q} can be written, respectively,  as 
\begin{align}
\frac{d}{dt} \bE_t^n(f) = &-(1+\gamma_n\kappa_n)^2\; \bQ_t^{2,n}\Big(\nabla_n f \otimes \delta\Big) +   \bE_t^n\Big( n^{-\frac12}\;\Delta_n f\Big) \notag\\
& - 2\gamma_n(1+\gamma_n\kappa_n) \; \bQ_t^{4,n}\Big(\nabla_n f \otimes \delta\Big) - \gamma_n^2 \; \bQ_t^{6,n}\Big( \nabla_n f \otimes \delta\Big),\label{eq:Sprime} \\
~\notag\\
\frac{d}{dt} \bQ_t^{2,n}(h) =& \bQ_t^{2,n}\Big(\mathbf{L}_n h \Big) -4  \bE_t^n\Big( \cD_n h\Big) - \gamma_n \bE_t^{4,n}\Big(\cD_n h\Big) + 2 \bQ_t^{2,n}\Big(n^{-\frac12}\;\widetilde{\cD}_n h\Big) \notag\\
&  +2\gamma_n \bQ_{t}^{4,n}\Big(n B_n h\Big) + \gamma_n   \kappa_n \bQ_t^{2,n}\Big(\sqrt n\,\nabla_n h\Big). \label{eq:Q1prime}
\end{align}
Let $h_n:\tfrac1n \Z \times \tfrac1n \Z \to \R$ be the symmetric solution of the \textit{Poisson equation}
\begin{equation}\label{eq:poisson}
\mathbf{L}_n h_n=(1+\gamma_n\kappa_n)^2 \; \nabla_n f\otimes \delta, \ \text{ with } \ \mathbf{L}_n=\sqrt n\; (1+\kappa_n\gamma_n) A_n + n^{-\frac12}\Delta_n.
\end{equation}
Then, summing equations \eqref{eq:Sprime} and \eqref{eq:Q1prime}, and integrating in time between $0$ and $T>0$ fixed, we obtain 
\begin{align}
 \bE_T^n(f) = \bE_0^n(f)  & - 4 \int_0^T  \bE_t^n\big( \cD_n h_n\big) dt \label{eq:decomp1} \\
 &  + \bQ_0^{2,n}(h_n) - \bQ_T^{2,n}(h_n) -\gamma_n\; \int_0^T \bE^{4,n}_t\big(\cD_n h_n\big) dt \label{eq:decomp2}\\
 & +   \int_0^T \bE_t^n\big( n^{-\frac12}\;\Delta_n f\big) dt + \gamma_n  \kappa_n \int_0^T   \bQ_t^{2,n}\big(\sqrt n \, \nabla_n h_n\big)\; dt \label{eq:decomp3}\\
 & - \gamma_n^2 \; \int_0^T \bQ_t^{6,n}\big( \nabla_n f \otimes \delta\big) dt\label{eq:decomp4} \\
  & + 2 \int_0^T \bQ_t^{2,n}\big(n^{-\frac12}\;\widetilde{\cD}_n h_n\big) dt \label{eq:decomp5}\\
& + 2\gamma_n \int_0^T \bQ_t^{4,n}\big(   n B_n h_n- (1+\gamma_n\kappa_n) \nabla_n f \otimes \delta\big) dt.\label{eq:decomp6}
\end{align}
We want to estimate the range of the parameter $\gamma_n$ for which the unique term that contributes to the equality above is \eqref{eq:decomp1}, namely
$
 -4  \bE_t^n\big( \cD_n h_n\big).
$
This term will be replaced by 
$
 \bE_t^n\big( {\bf L} f \big)
$
thanks to the next proposition which is proved in Section \ref{ssec:main}.

\begin{prop}\label{prop:main}
The solution $h_n$ of \eqref{eq:poisson} satisfies 
\begin{equation}
\label{eq:limdh}
\lim_{n\to\infty} \tfrac 1 n \sum_{x \in\Z} \big| \mc D_n h_n \big(\tfrac x n\big) + \tfrac 1 4 {\bf L}f\big(\tfrac x n \big)\big|^2 = 0,
\end{equation}
where $\bf L$ is the operator defined in \eqref{eq:opL}.
\end{prop}

We know two different ways to prove that the other terms vanish as $n \to \infty$: the Cauchy-Schwarz inequality, and the Kipnis-Varadhan inequality, which have been presented in Section \ref{sec:esti}. We start with the easiest term: from \eqref{eq:CS-S} one can directly see that, uniformly in $t\in [0,T]$,  
\[
\Big| \bE_t^n\big( n^{-\frac12}\;\Delta_n f\big)  \Big| \xrightarrow[n\to\infty]{} 0,
\]
independently of $\gamma_n$. The other contributions need some work. In the following proposition, which is proved  in Appendices \ref{sec:poisson} and \ref{app:integral}, we estimate some $\ell^2$-norms that will be used in the Cauchy-Schwarz argument.

\begin{prop}[$\ell^2$-norms involving the solution of the Poisson equation \eqref{eq:poisson}]\label{prop:estimate2}
If $h_n:\frac1n \Z \times \frac1n \Z \to \R$ is the symmetric solution of 
\eqref{eq:poisson}, then
\begin{align}
&\sum_{x,y\in\Z}h_n^2\Big(\frac{x}{n},\frac{y}{n}\Big) = \mc O\big(n^{\frac32}\big),\label{eq:est0}\\
&\sum_{x\in\Z}h_n^2\Big(\frac{x}{n},\frac{x}{n}\Big) = \mc O\big(n\big),\label{eq:est0x}\\
&\sum_{x\in\Z}\Big[\cD_nh_n\Big(\frac{x}{n}\Big)\Big]^2= \mc O(n),\label{eq:est1}\\
&\sum_{x\in\Z}\Big[h_n\Big(\frac{x+1}{n},\frac{x+1}{n}\Big)- h_n\Big(\frac{x}{n},\frac{x}{n}\Big)\Big]^2= \mc O\big(n^{-1}\big).\label{eq:est4}
\end{align}
\end{prop}

A direct consequence of Proposition \ref{prop:estimate2} and the Cauchy-Schwarz inequality is the following: from \eqref{eq:est1} we have
\[
\Big|\gamma_n\bE^{4,n}_t\big(\cD_n h_n\big)\Big|\xrightarrow[n\to\infty]{} 0,
\]
since it is of order $\gamma_n$. 
From \eqref{eq:est0} we also get
\[
\Big|\bQ_T^{2,n}(h_n)\Big|   \xrightarrow[n\to\infty]{} 0, \qquad
\Big|\bQ_0^{2,n}(h_n)\Big|  \xrightarrow[n\to\infty]{} 0, \] independently of $\gamma_n$.
We know from \eqref{eq:est4} that the $\mathbb{L}^2(\P)$-norm of  
\[
 \gamma_n  \kappa_n \int_0^T   \bQ_t^{2,n}\big(\sqrt n\, \nabla_n h_n\big)\; dt
\]
is of order $\gamma_n$, and then vanishes. To sum up, both terms \eqref{eq:decomp2} and \eqref{eq:decomp3} vanish in $\bb L^2(\bb P)$, as soon as $\gamma_n=o(1)$.  

Moreover, concerning \eqref{eq:decomp4}, observe that by a first order Taylor expansion we have that $
\big(N_n^{\neq}(\nabla_n f\otimes \delta)\big)^2  = \mc O(n).$ Therefore  \eqref{eq:CS-Q} gives that 
\[
\Big|\gamma_n^2\; \bQ_t^{6,n}(\nabla_n f \otimes \delta)\Big| \xrightarrow[n\to\infty]{} 0,\]
if $\gamma_n=o(n^{-\frac14})$, since it is of order $\gamma_n^2 \sqrt n$.

Finally, the  terms that need a refined investigation are \eqref{eq:decomp5} and \eqref{eq:decomp6}. For these terms, the Cauchy-Schwarz inequality is not sharp enough, hence we are going to estimate them using a dynamical argument for \eqref{eq:decomp5} and the Kipnis-Varadhan inequality \eqref{eq:estim},  with some $\cH_{-1,z}$ norms estimates for \eqref{eq:decomp6}. This is the purpose of Section \ref{subsec:blabla1} and Section \ref{subsec:blabla2} respectively, in which we prove the following results:

\begin{prop}\label{prop_for_q_2}{For $h_n$  solution of \eqref{eq:poisson}, we have that}
\[
 \lim_{n\to\infty} \bb E\bigg[\bigg( \int_0^T \bQ_t^{2,n}\big(n^{-\frac12}\widetilde  {\mc D}_n h_n\big) dt \bigg)^2\bigg] =0.
\]
\end{prop}
\begin{prop} \label{prop_for_q_4} {Let $h_n$  be solution of \eqref{eq:poisson}. If $\gamma_n = o(n^{-\frac14})$, then}
\[
 \lim_{n\to\infty} \bb E\bigg[\bigg(\gamma_n \int_0^T \bQ_t^{4,n}\big(  2n B_n h_n-2 (1+\gamma_n\kappa_n) \nabla_n f \otimes \delta\big) dt \bigg)^2\bigg] =0.
\]
\end{prop}

In the next sections we give the details of the proofs of  Propositions \ref{prop:main}, \ref{prop_for_q_2} and \ref{prop_for_q_4}.  

\begin{rem}Note that in the whole argument we used the restriction $\gamma_n=o(n^{-\frac14})$ only twice:
\begin{itemize}
\item first, to make \eqref{eq:decomp4} vanish as $n\to\infty$;
\item second, to prove Proposition \ref{prop_for_q_4}.
\end{itemize} Our conjecture is that in the first case, the limitation could be improved quite easily. Indeed, to treat \eqref{eq:decomp4} recall that we roughly applied the Cauchy-Schwarz inequality, but it is highly probable that a $\mc H_{-1,z}$-norm argument would give a more refined estimate, and therefore would relax the restriction. However, in the second case, we do not see any easy way to improve the result. This is why we believe that the restriction $\gamma_n=o(n^{-\frac14})$ comes from Proposition \ref{prop_for_q_4}, and only from this last result. 
\end{rem}

\subsection{Proof of Proposition \ref{prop:main}}

\label{ssec:main} 

First let us note that the Fourier transform of $h_n$ defined in \eqref{eq:poisson} is explicitly given, for any $(u,v) \in [-\frac{n}{2},\frac{n}{2}]^2$, by
\begin{equation}
\label{eq:fourierofh}
\begin{split}
\cF_n(h_n)(u,v)&=\frac{(1+\kappa_n\gamma_n)^2}{2\sqrt n} \frac{i\Omega(\tfrac{u}{n},\tfrac{v}{n}) \cF_n(f)(u+v)}{(1+\kappa_n\gamma_n) \Lambda(\tfrac{u}{n},\tfrac{v}{n})-i\Omega(\tfrac{u}{n},\tfrac{v}{n})}\\
&= \vphantom{\Bigg(}\frac{(1+\kappa_n\gamma_n)^2}{2\sqrt n}\; \Theta \big( \tfrac{u}{n},\tfrac{v}{n} \big) \, \cF_n(f)(u+v)
\end{split}
\end{equation}
where, for $(k,\ell) \in \big[-\frac{1}{2},\frac{1}{2}\big]^2$,
\begin{align}
\label{eq:lambda}
\Lambda (k,\ell) & := 4  \left[ \sin^2 (\pi k) + \sin^2 (\pi \ell)\right], \\
\Omega (k,\ell)& := 2 \left[ \sin (2\pi k) + \sin (2 \pi \ell)\right],\label{eq:omega}\\
\Theta(k,\ell)&:= \frac{i\Omega(k,\ell)}{(1+\kappa_n\gamma_n)\Lambda(k,\ell)-i\Omega(k,\ell)}.\label{eq:theta}
\end{align}
Let $G_0$ be defined for any $v \in \R$ as 
\[ 
G_0(v):=\frac{1}{2} \big| \pi v\big|^{\frac32} (1+i{\rm sgn}(v)),\]
and let us denote $q:\R \to \R$ the function 
\[ q(u):=\int_\R e^{-2i\pi u v} \; G_0(v) \mc F(f)(v)\; dv = -\frac{1}{4} {\bf L}f(u).\]
The last equality is obtained by computing the Fourier transform of the right hand side and using the inverse Fourier transform. We want to write a similar Fourier identity for $\mc D_n h_n(\frac x n)$.  Let $q_n : {\frac{1}{n}} \Z \to \R$ be the function defined by
\[q_n \big(\tfrac{x}{n}\big) = {{\mc D}}_n h_n \,  \big(\tfrac{x}{n}\big).
\]
Since $\mc F_n(h_n)$ is a symmetric function we can easily see (as in \cite[Lemma D.1]{BGJ}) that
\begin{equation*}
\mc F_n(q_n) (\xi) = -\tfrac{i}{2}  \sum_{x \in \Z} e^{ 2i \pi x\frac{ \xi}{n}} \iint_{[-\frac{n}{2}, \frac{n}{2}]^2} e^{-2i \pi x\frac{ (k+\ell) }{n}}  \Omega \big( \tfrac{k}{n}, \tfrac{\ell}{n}\big)  \mc F_n(h_n) (k,\ell) \, dk d\ell.
\end{equation*}
We use now Lemma \ref{lem:cov} (proved ahead) and the inverse Fourier transform relation to get
\[
\mc F_n(q_n) (\xi) =- \tfrac{in}{2} \int_{[-\frac{n}{2}, \frac{n}{2}]} \Omega  \big( \tfrac{\xi -\ell}{n}, \tfrac{\ell}{n}\big)\;  \mc F_n(h_n) (\xi- \ell,\ell) \, d\ell.
\]
By the explicit expression \eqref{eq:fourierofh} of $\mc F_n(h_n)$ we obtain that
\begin{equation*}
\mc F_n({q_n}) (\xi) = \frac{(1+\gamma_n\kappa_n)^2}{4} \; n^{\frac32}\; K(\tfrac{\xi}{n})\; \mc F_n(f) (\xi),
\end{equation*}
where
\begin{equation}
K(y):=- \int_{[-\frac{1}{2}, \frac{1}{2}]} i\, (\Omega \Theta)\, (y-x,x) dx. 
\label{eq:defK}
\end{equation}
Therefore, we have proved the identity
\begin{equation*}\mc D_nh_n\big(\tfrac x n \big) = \int_{[-\frac n 2,\frac n 2]} e^{-2i \pi x \frac \xi n}\;  n^{\frac32} G_n\big(\tfrac \xi n\big) \mc F_n(f) (\xi) \; d\xi,\end{equation*}
where $G_n$ is the 1-periodic function defined for any $y \in [-\frac12,\frac12]$ by 
\begin{equation}\label{eq:defgn}
G_n(y)=\frac{(1+\gamma_n\kappa_n)^2}4 \; K(y).
\end{equation}
The proof is now reduced to prove that $n^{\frac32}G_n(\frac \xi n)$ is close to $G_0(\xi)$, in the following sense:

\begin{lem} \label{lem:est}
For any $\xi \in [-\frac n 2,\frac n 2]$, 
$$ \Big|n^\frac32\;G_n\Big(\frac \xi n \Big) - G_0(\xi) \Big| \; \lesssim \;  \gamma_n |\xi|^{\frac32} + \cfrac{|\xi|^2}{\sqrt n}.$$
\end{lem}
Lemma \ref{lem:est} is proved in Section \ref{app:gn}. Let us now prove \eqref{eq:limdh}: we have
\begin{align*}
\mc D_nh_n\big(\tfrac x n\big) + \tfrac 1 4 {\bf L}f\big(\tfrac x n \big) & = q_n\big(\tfrac x n\big)-q\big(\tfrac x n\big) \\
& = \int_{|\xi| \geq \frac n 2} e^{-2i\pi x \frac \xi n} \; G_0(\xi) \; \mc F(f)(\xi)\; d\xi \\
& \quad + \int_{|\xi|\leq \frac n 2} e^{-2i\pi x \frac \xi n} \; G_0(\xi) \; \Big(\mc F(f)(\xi)-\mc F_n(f)(\xi)\Big)\; d\xi \\
& \quad + \int_{|\xi| \leq \frac n 2} e^{-2i\pi x \frac \xi n} \; \Big( G_0(\xi) - n^{\frac32}\; G_n\big(\tfrac \xi n\big)\Big) \; \mc F_n(f)(\xi)\; d\xi.
\end{align*}
Therefore, we bound from the Parseval-Plancherel identity as
\begin{align}
\tfrac 1 n \sum_{x\in\Z} \big[q\big(\tfrac x n\big)-q_n\big(\tfrac x n\big) \big]^2 & \leq \tfrac 1 n \sum_{x\in\Z} \bigg|\int_{|\xi| \geq \frac n 2} e^{-2i\pi x \frac \xi n} \; G_0(\xi) \; \mc F(f)(\xi)\; d\xi\bigg|^2 \vphantom{\Bigg(} \label{eq:bb1}\\
& \quad +   \int_{|\xi|\leq \frac n 2}  \big| G_0(\xi)\big|^2  \; \big|\mc F(f)(\xi)-\mc F_n(f)(\xi)\big|^2\; d\xi 
  \vphantom{\Bigg(}  \label{eq:bb2} \\
  & \quad +  \int_{|\xi| \leq \frac n 2}  \Big| G_0(\xi) - n^{\frac32}\; G_n\big(\tfrac \xi n\big)\Big|^2 \; \big|\mc F_n(f)(\xi)\big|^2\; d\xi.
  \vphantom{\Bigg(}  \label{eq:bb3}
\end{align}
We treat each term \eqref{eq:bb1}, \eqref{eq:bb2} and \eqref{eq:bb3} separately. For the first one \eqref{eq:bb1}, we perform an integration by parts and then we use two facts: first the Fourier transform $\mc F(f)$ of $f$ is in the Schwartz space, and second, the functions $G_0$ and $G_0'$ grow at most polynomially. This implies that \eqref{eq:bb1} is bounded by a constant times
\[
\frac 1 n \sum_{x\in\Z} \frac{n^2}{|x|^2}\Big\{   \big| G_0\; \mc F(f) \big|^2 \big(\pm \tfrac n 2 \big) + \Big| \int_{|\xi| \geq \frac n 2} \big| \tfrac{d}{d\xi}[G_0\; \mc F(f)](\xi) \big| \; d\xi \Big|^2 \Big\} \; \lesssim \; \; n^{-p},   \] from some $p>0$ and it vanishes as $n\to \infty$. The second term \eqref{eq:bb2} can be bounded from above by 
\[
\int_{[-\frac n 2,\frac n 2]} |\xi|^3 \; \big|\mc F(f)(\xi)-\mc F_n(f)(\xi)\big|^2\; d\xi. 
\]
Performing the change of variables $y=\frac \xi n$ and using the fact that $f$ is in the Schwartz space (together with Lemma \ref{lem:sfp}), we can prove that \eqref{eq:bb2} does not contribute to the limit $n\to +\infty$. 
Finally, the contribution of \eqref{eq:bb3} is estimated by using Lemma \ref{lem:est}: it is bounded by a sum of terms of the form
\[\frac{1}{n^{\beta}} \int_{[-\frac n 2,\frac n 2]} |\xi|^\alpha \; \big|\mc F_n(f)(\xi)\big|^2\; d\xi \; \lesssim \; \frac{1}{n^\beta} \int_{[-\frac n 2,\frac n 2]} \frac{|\xi|^\alpha}{(1+|\xi|^p)^2}\; d\xi,\]
for some $\alpha,\beta >0$. The last inequality above is a consequence of Lemma \ref{lem:sfp}, and therefore it vanishes as $n \to + \infty$ after choosing $p$ such that $2p > \alpha +1$.

\subsection{Proof of Proposition \ref{prop_for_q_2}: the dynamical argument}
 \label{subsec:blabla1}

As in \cite{BGJ} the idea of the proof consists in using once again the differential equation \eqref{eq:Q1prime} after solving a new Poisson equation, with a different right hand side. For that purpose,  let $ v_n:\frac{1}{n}\Z\times \frac{1}{n}\Z \to \R$ be the symmetric solution of 
\begin{equation} \label{eq:v}
\mathbf{L}_n v_n= 2n^{-\frac12}\widetilde \cD_nh_n, \ \text{ with } \ \mathbf{L}_n={n}^\frac12 (1+\kappa_n \gamma_n) A_n + n^{-\frac12}\Delta_n.
\end{equation}
From \eqref{eq:Q1prime}, the term that we want to estimate is now equal to 
\begin{equation*}
\begin{split}
 \int_0^T \bQ_t^{2,n}\big(n^{-\frac{1}{2}}\widetilde{\cD}_n h_n\big) dt =& \bQ_T^{2,n}(v_n) - \bQ_0^{2,n}(v_n) + 4 \int_0^T\bE_t^n(\cD_nv_n)dt\\
  +&  \gamma_n\int_0^T \bE_t^{4,n}(\cD_n v_n) dt
 - 2 \int_0^T\bQ_t^{2,n}\big(n^{-\frac12} \widetilde\cD_nv_n\big)dt \\
 -& 2\gamma_n \int_0^T\bQ_t^{4,n}(n B_n v_n) dt- \gamma_n\kappa_n \int_0^T\bQ_t^{2,n}(\sqrt n\, \nabla_n v_n)dt.
\label{eq:dyna}
 \end{split}
 \end{equation*}
It turns out that the work becomes easier, since the rough Cauchy-Schwarz inequality will be enough to control all these terms but one, thanks to the following proposition whose proof is given in Appendix \ref{sec:poisson}.

\begin{prop}[$\ell^2$-norms involving the solution of the Poisson equation \eqref{eq:v}]\label{prop:estimatevn}
If $v_n:\frac1n \Z \times \frac1n \Z \to \R$ is the symmetric solution to \eqref{eq:v}, then
\begin{align}
&\sum_{x,y\in\bb Z}v_n^2\Big(\frac{x}{n},\frac{y}{n}\Big) = \mc O(\sqrt n),\label{eq:est0-vn}\\
&\sum_{x\in\Z}v_n^2\Big(\frac{x}{n},\frac{x}{n}\Big) =\mc O(1),\label{eq:est0x-vn}\\
&\sum_{x\in\Z}\Big[\cD_nv_n\Big(\frac{x}{n}\Big)\Big]^2=\mc O(1),\label{eq:est1-vn}\\
&\sum_{x\in\Z}\Big[v_n\Big(\frac{x+1}{n},\frac{x+1}{n}\Big)- v_n\Big(\frac{x}{n},\frac{x}{n}\Big)\Big]^2=\mc O(n^{-2}),\label{eq:est4-vn}\\
& \sum_{x\in\Z}\Big[\widetilde\cD_n v_n\Big(\frac{x}{n},\frac{x+1}{n}\Big)\Big]^2 = \mc O(n^{2}).\label{eq:est5-vn}
\end{align}
\end{prop}

From \eqref{eq:CS-Q} and  \eqref{eq:est0-vn} we get that
\[
\Big|\bQ_T^{2,n}(v_n)\Big|   \xrightarrow[n\to\infty]{} 0, \qquad
\Big|\bQ_0^{2,n}(v_n)\Big|  \xrightarrow[n\to\infty]{} 0, \] independently of $\gamma_n$. From  \eqref{eq:CS-S4} and  \eqref{eq:est1-vn} we get that
\[
{\color{red}\gamma_n}\int_0^T \Big|\bE_t^{4,n}(\mathcal D_nv_n)\Big| dt  =\mc O(\gamma_n n^{-\frac12}). \]
Analogously, from \eqref{eq:CS-Q} and  \eqref{eq:est5-vn} we have that

\[
\Big|\bQ_T^{2,n}(n^{-\frac 12}\widetilde{\mathcal{D}}_n v_n)\Big|   \xrightarrow[n\to\infty]{} 0, \] independently of $\gamma_n$. Finally, from  \eqref{eq:CS-Q} and \eqref{eq:est4-vn} we have that
\[
\Big|\gamma_n\bQ_T^{2,n}(\sqrt n\, \nabla_n v_n)\Big|   \xrightarrow[n\to\infty]{} 0, \] independently of $ \gamma_n$.  
Moreover, we prove below in Section \ref{sec:new_estimate} the following

\begin{prop} \label{prop_for_q_4_new} { For $v_n$ solution of \eqref{eq:v}, we have that }
\begin{equation}\label{eq:new_estimate}
 \lim_{n\to\infty} \bb E\bigg[\bigg(\gamma_n \int_0^T \bQ_t^{4,n}\big(  n B_n v_n\big) dt \bigg)^2\bigg] =0.
\end{equation}
\end{prop}
This ends the proof of Proposition \ref{prop_for_q_2}.

\subsection{Proof of Proposition \ref{prop_for_q_4}}
\label{subsec:blabla2}
Thanks to the Kipnis-Varadhan inequality \eqref{eq:estim} together with Lemma \ref{lem:two} we have reduced the problem to control the $\bb L^2(\bb P)$-norm of  \eqref{eq:decomp6} to estimating the behavior w.r.t.~$n$ of the integral  \eqref{eq:h-1-two} with $z= n^{-\frac{3}{2}}$ and $\varphi$ suitably defined.  We can make a change of variables in \eqref{eq:h-1-two} in order to get: 
\[
\frac{1}{n^2}\iint_{[-\frac{n}2,\frac{n}2]^2}\; \frac{|{\widehat{\varphi}}(\tfrac{u}{n},\tfrac{v}{n})|^2}{ z +  4\sin^{2} (\pi \tfrac{u}{n}) + 4 \sin^{2} (\pi \tfrac{v}{n}) } dudv.
\]
Now, we observe that
\[\bQ_t^{4,n}\big(  2 n B_n h_n-  2(1+\gamma_n\kappa_n) \nabla_n f \otimes \delta\big)={\mathbb E} \bigg[ {\mc E}_0^n (g) \; \bigg\{   \sum_{x\neq y} \Phi_n(x,y) H_{3\delta_x+\delta_y}(\omega(tn^{\frac32})) \bigg\} \bigg]  \]
where 
\begin{multline*}
\Phi_n(x,y)=\\ \begin{cases} 0 & \text{ if } y=x,\\ \displaystyle \vphantom{\frac12} 2\sqrt{n}\; h_n\big(\tfrac{x-1}{n},\tfrac{x+1}{n}\big)  -n\; (1+\gamma_n\kappa_n)\big(f\big(\tfrac{x+1}{n}\big)-f\big(\tfrac{x}{n}\big)\big)   & \text{ if } y=x+1,\\
 \displaystyle  \vphantom{\frac12}-2\sqrt{n}\; h_n\big(\tfrac{x+1}{n},\tfrac{x-1}{n}\big)-n\;(1+\gamma_n\kappa_n) \big(f\big(\tfrac{x}{n}\big)-f\big(\tfrac{x-1}{n}\big)\big)   & \text{ if } y=x-1,\\
  \displaystyle \vphantom{\frac12}2 \sqrt{n}\big(h_n\big(\tfrac{x-1}{n},\tfrac{y}{n}\big) - h_n\big(\tfrac{x+1}{n},\tfrac{y}{n}\big) \big) & \text{ otherwise.} 
\end{cases}
\end{multline*}
Then, compute the Fourier transform: for  $(k,\ell) \in \big[-\tfrac{1}{2},\tfrac{1}{2}\big]^2$,
\begin{align*}
\widehat{\Phi}_n(k,\ell)& = \sum_{(x,y)\in\Z^2} \Phi_n(x,y) e^{2i\pi(kx+\ell y)}\\
& = -(1+\gamma_n\kappa_n) n \sum_{x\in\Z} \big[f\big(\tfrac{x+1}{n}\big)-f\big(\tfrac{x}{n}\big)\big] e^{2i\pi x(k+\ell)} \big(e^{2i\pi\ell}+e^{2i\pi k}\big)\\
& \quad + 2n^{\frac{1}{2}}\; \sum_{j\geq 1} \sum_{u\in\Z}  h_n\big(\tfrac{u}{n},\tfrac{u+j}{n}\big) e^{2i\pi (k+\ell) u} \big(e^{2i\pi k}-e^{-2i \pi k}\big) \big(e^{2i\pi \ell j}+e^{2i\pi k j}\big)\\
& \quad - 2 n^{\frac{1}{2}}\; \sum_{u\in\Z} h_n\big(\tfrac{u}{n},\tfrac{u+1}{n}\big) e^{2i\pi u (k+\ell)} \big(e^{2i \pi (k+\ell)}-1\big).
\end{align*}
We have, for any $m\in\Z$, (see Lemma \ref{lem:cov} below)
\[
\sum_{u\in\Z} h_n\big(\tfrac{u}{n},\tfrac{u+m}{n}\big) e^{2i\pi(k+\ell)u} = n\int_{[-\frac{n}{2},\frac{n}{2}]} \cF_n(h_n)(n(k+\ell)-u,u)e^{-2i\pi m \tfrac{u}{n}} \; du.
\]
Therefore,  we have
\begin{align*}
\widehat{\Phi}_n(k,\ell)  =&  (1+\gamma_n\kappa_n) \; n^2 \cF_n(f)(n(k+\ell)) \; i\Omega(k,\ell) \\
&  + 2n^\frac32 \big( e^{2i\pi k}-e^{-2i \pi k} \big) \int_{[-\frac{n}{2},\frac{n}{2}]} \cF_n(h_n)(n(k+\ell)-u,u) \\
& \qquad \qquad \qquad \qquad \qquad \qquad \times \sum_{j\geq 1}\big[ e^{-2i\pi j \big(\tfrac{u}n - \ell\big)}+ e^{-2i\pi j\big(\tfrac{u}n -k\big)} \big]\; du\\
&  - 2n^\frac32 \big(e^{2i\pi(k+\ell)}-1\big)\int_{[-\frac{n}{2},\frac{n}{2}]} \cF_n(h_n)(n(k+\ell)-u,u) e^{-2i\pi\tfrac{u}n} \; du.
\end{align*}
Observe that $\cF_n(h_n)(\xi-u,u) = \cF_n(h_n)(u,\xi-u)$ for any $\xi,u \in \R$. Therefore, one can perform the change of variables $v=n(k+\ell) -u$ in one of the integrals in the second term above, and one gets
\begin{align*}
\widehat{\Phi}_n&(k,\ell)  = (1+\gamma_n\kappa_n) \; n^2 \cF_n(f)(n(k+\ell)) \; i\Omega(k,\ell) \\
&  + 2n^\frac32 \big( e^{2i\pi k}-e^{-2i \pi k} \big) \int_{[-\frac{n}{2},\frac{n}{2}]} \cF_n(h_n)(n(k+\ell)-u,u) \times 2 \text{Re}\left[ \Big(e^{2i\pi  \big(\tfrac{u}n - \ell\big)}-1\Big)^{-1} \right]\; du\\
&  - 2n^\frac32 \big(e^{2i\pi(k+\ell)}-1\big)\int_{[-\frac{n}{2},\frac{n}{2}]} \cF_n(h_n)(n(k+\ell)-u,u) e^{-2i\pi\tfrac{u}n} \; du.
\end{align*}
Note that
\[
2 \text{Re}\left[ \Big(e^{2i\pi  \big(\tfrac{u}n - \ell\big)}-1\Big)^{-1} \right]=-1.
\]
Using \eqref{eq:fourierofh}, we obtain
\[ \widehat{\Phi}_n(k,\ell) = n^2 (1+\gamma_n\kappa_n)  \cF_n(f)(n(k+\ell))   R_n(k,\ell) \]
where
\[R_n(k,\ell)=i\Omega(k,\ell) + (1+\gamma_n\kappa_n) \;  \; \Big\{ I(k+\ell) \big(e^{-2i\pi k}-e^{2i \pi k}\big) + J(k+\ell)\Big\}\]
and  $I,J$ are defined as (see \eqref{eq:theta} for the definition of $\Theta$) 
\begin{align}
\label{eq:defI}
I(y)&:=\int_{[-\frac 1 2, \frac 1 2]} \Theta(y-x,x)\; dx,\\
J(y)&:=\int_{[-\frac12,\frac12]}\Theta(y-x,x) e^{-2i\pi x}(1-e^{2i\pi y})\ dx. \label{eq:defJ}
\end{align}
Finally, from \eqref{eq:estim} and Lemma \ref{lem:two}, the square of the $\mathbb{L}^2(\mathbb{P})$-norm of 
\[
  \gamma_n \int_0^T \bQ_t^{4,n}\big(  n B_n h_n- (1+\gamma_n\kappa_n) \nabla_n f \otimes \delta\big) dt.
\]
is bounded from above by (recall that $z=n^{-\frac32}$)
\begin{align}
&C \gamma_n^2\; n^{\frac52}\; \iint_{[-\frac{1}{2},\frac{1}{2}]^2} \frac{|R_n(k,\ell)|^2}{z+\sin^2(\pi k)+\sin^2(\pi \ell)}\big|\cF_n(f)(n(k+\ell))\big|^2 \; dk d\ell \notag\\ = \; & C \gamma_n^2\; n^{\frac52} \int_{[-\frac{1}{2},\frac{1}{2}]} U_n(\xi) \big|\cF_n(f)(n\xi)\big|^2 \; d\xi,  \label{eq:toreplace}
\end{align}
where
\[U_n(\xi):=\int_{[-\frac{1}{2},\frac{1}{2}]} \frac{|R_n(k,\xi -k)|^2}{z+\sin^2(\pi k)+\sin^2(\pi (\xi-k))} \; dk.  \]
We can bound $|R_n(k,\xi-k)|$ from above as follows:
\[|R_n(k,\xi -k)|^2\leq 2|\Omega(k,\xi-k)|^2 + 2(1+\gamma_n\kappa_n)^2 \big| 2i\sin(2\pi k) I(\xi)  + J(\xi)\big|^2. \]
Let us first estimate the contribution coming from $\Omega(k,\xi-k)$, namely: 
\begin{align*}
\int_{[-\frac{1}{2},\frac{1}{2}]} \frac{|\Omega (k,\xi-k)|^2 }{z+ \sin^2 (\pi k) + \sin^2 (\pi (\xi-k) ) }\; \; dk 
& \le  \int_{[-\frac{1}{2},\frac{1}{2}]} \frac{4 \sin^2 (\pi \xi)}{z+ \sin^2 (\pi k) + \sin^2 (\pi (\xi-k) ) }\; \; dk\\
& \le 4 n^{\frac32} \sin^2 (\pi \xi).
\end{align*}
Therefore, the contribution of this term in \eqref{eq:toreplace} gives
\[
C\gamma_n^2  \int_{[-\frac{1}{2},\frac{1}{2}]}  n^4\sin^2 (\pi \xi)  \big|\mc F_n(f)(n\xi)\big|^2\; d\xi
\] and by Lemma \ref{lem:sfp} it vanishes since $n \gamma_n^2 \to 0$. 
Moreover, we now use Lemma \ref{lem:integral_residues} in order to estimate the remaining contribution, namely
\begin{align*}
\int_{[-\frac{1}{2},\frac{1}{2}]} \frac{\big| 2i\sin(2\pi k) I(\xi)  + J(\xi)\big|^2}{z+ \sin^2 (\pi k) + \sin^2 (\pi (\xi-k) ) }\;  dk & \le C \int_{[-\frac{1}{2},\frac{1}{2}]} \frac{ \sin^2(2\pi k) |I(\xi)|^2  + |J(\xi)|^2}{z+ \sin^2 (\pi k) + \sin^2 (\pi (\xi-k) ) }\; dk \\
& \leq C  \int_{[-\frac{1}{2},\frac{1}{2}\big]} \frac{ \sin^2(2\pi k) |\sin(\pi \xi)| + |\sin(\pi \xi)|^3}{z+ \sin^2 (\pi k) + \sin^2 (\pi (\xi-k) ) }\;  dk \\
&\le C\Big[  |\sin(\pi \xi)| + n^{\frac32} |\sin(\pi \xi)|^3\Big].
\end{align*}
Therefore, the contribution of this term in \eqref{eq:toreplace} gives
\[
C\gamma_n^{2} n^{\frac52}  \int_{[-\frac{1}{2},\frac{1}{2}]}  \Big[  |\sin(\pi \xi)| + n^{\frac32} |\sin(\pi \xi)|^3\Big] \big|\mc F_n(f)(n\xi)\big|^2\; d\xi
\] and by Lemma \ref{lem:sfp} it vanishes if $\sqrt{n} \gamma_n^2 \to 0$. 
This ends the proof of Proposition \ref{prop_for_q_4}.

\subsection{Proof of Proposition \ref{prop_for_q_4_new}}
\label{sec:new_estimate}
Here we follow closely the proof of Proposition \ref{prop_for_q_4} and for that reason we skip some steps of it. First we observe that
\[\bQ_t^{4,n}\big(  n B_n v_n\big)={\mathbb E} \bigg[ {\mc E}_0^n (g) \; \bigg\{   \sum_{x\neq y} \Psi_n(x,y) H_{3\delta_x+\delta_y}(\omega(tn^{\frac32})) \bigg\} \bigg]  \]
where 
\[
\Psi_n(x,y)= \begin{cases} 0 & \text{ if } y=x,\\  \displaystyle \vphantom{\frac12}  \sqrt n\; v_n\big(\tfrac{x-1}{n},\tfrac{x+1}{n}\big)    & \text{ if } y=x+1,\\
 \displaystyle\vphantom{\frac12} -\sqrt n\; v_n\big(\tfrac{x+1}{n},\tfrac{x-1}{n}\big)  & \text{ if } y=x-1,\\
  \displaystyle\vphantom{\frac12}  \sqrt n\big(v_n\big(\tfrac{x-1}{n},\tfrac{y}{n}\big) - v_n\big(\tfrac{x+1}{n},\tfrac{y}{n}\big) \big) & \text{ otherwise.} 
\end{cases}
\]
The Fourier transform of $\Psi$ is given on  $(k,\ell) \in \big[-\tfrac{1}{2},\tfrac{1}{2}\big]^2$ by
\begin{align*}
\widehat{\Psi}_n(k,\ell)& = \sum_{(x,y)\in\Z^2} \Psi_n(x,y) e^{2i\pi(kx+\ell y)}\\
& =  n^{\frac{1}{2}}\; \sum_{j\geq 1} \sum_{u\in\Z}  v_n\big(\tfrac{u}{n},\tfrac{u+j}{n}\big) e^{2i\pi (k+\ell) u} \big(e^{2i\pi k}-e^{-2i \pi k}\big) \big(e^{2i\pi \ell j}+e^{2i\pi k j}\big)\\
& \quad - n^{\frac{1}{2}}\; \sum_{u\in\Z} v_n\big(\tfrac{u}{n},\tfrac{u+1}{n}\big) e^{2i\pi u (k+\ell)} \big(e^{2i \pi (k+\ell)}-1\big).
\end{align*}
\begin{align*}
\widehat{\Psi}_n(k,\ell)  =&\;  n^\frac32 \big( e^{2i\pi k}-e^{-2i \pi k} \big) \int_{[-\frac{n}{2},\frac{n}{2}]} \cF_n(v_n)(n(k+\ell)-u,u) \\
& \qquad \qquad \qquad \qquad \qquad \qquad \times \sum_{j\geq 1}\Big[ e^{-2i\pi j \big(\tfrac{u}n - \ell\big)}+ e^{-2i\pi j\big(\tfrac{u}n -k\big)} \Big]\; du\\
&  - n^\frac32 \big(e^{2i\pi(k+\ell)}-1\big)\int_{[-\frac{n}{2},\frac{n}{2}]} \cF_n(v_n)(n(k+\ell)-u,u) e^{-2i\pi\tfrac{u}n} \; du.
\end{align*}
The proof of this proposition is very similar to the previous one. Indeed,  what we really need is the expression of the Fourier transform of $v_n$, solution to the Poisson equation \eqref{eq:v}.  For any $(u,v) \in [-\frac{n}{2},\frac{n}{2}]^2$, we have 
\begin{equation}
\label{eq:fourierofv}
\cF_n(v_n)(u,v)=-\frac{2}{n} \frac{e^{2i\pi\frac u n }+e^{2i\pi\frac v n} }{(1+\kappa_n\gamma_n) \Lambda(\tfrac{u}{n},\tfrac{v}{n})-i\Omega(\tfrac{u}{n},\tfrac{v}{n})} \cF_n(w_n)(u+v),
\end{equation}
where $w_n : \frac1n \Z \to \R$ is defined by 
\begin{equation}\label{eq:w}
w_n\big(\tfrac x n \big):=h_n\big(\tfrac x n,\tfrac{x+1}n\big)-h_n\big(\tfrac x n,\tfrac x n\big)
\end{equation}
and therefore, for any $\xi \in [-\frac n 2,\frac n 2]$, we deduce from \eqref{eq:fourierofh} that
\begin{equation}\label{eq:fourierofw}
\mc F_n(w_n)(\xi) =\int_{[-\frac n 2,\frac n 2]} \mc F_n(h_n)(\xi-\ell,\ell) \big(e^{-2i\pi\frac\ell n}-1\big)\; d\ell = -\tfrac{\sqrt n}{2} L\big(\tfrac \xi n\big) \mc F_n(f)(\xi),
\end{equation}
where 
\begin{equation}
\label{eq:defL}
L(y):=\int_{[-\frac 1 2, \frac 1 2]} (1-e^{-2i\pi x}) \Theta(y-x,x)\; dx.
\end{equation}
From the previous computations we conclude that
\begin{equation}
\label{eq:fourierofv_new}
\cF_n(v_n)(u,v)=\frac{1}{\sqrt n} \frac{e^{2i\pi\frac u n }+e^{2i\pi\frac v n} }{(1+\kappa_n\gamma_n) \Lambda(\tfrac{u}{n},\tfrac{v}{n})-i\Omega(\tfrac{u}{n},\tfrac{v}{n})}L\big(\tfrac {u+v} {n}\big) \mc F_n(f)(u+v),
\end{equation}
Since $v_n$ is symmetric, $\cF_n(v_n)(\xi-u,u) = \cF_n(v_n)(u,\xi-u)$ for any $\xi,u \in \R$, and similar computations as before give
\begin{align*}
\widehat{\Psi}_n(k,\ell)  = & \; n^\frac32 \big( e^{2i\pi k}-e^{-2i \pi k} \big) \int_{[-\frac{n}{2},\frac{n}{2}]} \cF_n(v_n)(n(k+\ell)-u,u) \; du\\
&  - n^\frac32 \big(e^{2i\pi(k+\ell)}-1\big)\int_{[-\frac{n}{2},\frac{n}{2}]} \cF_n(v_n)(n(k+\ell)-u,u) e^{-2i\pi\tfrac{u}n} \; du.
\end{align*}
Using \eqref{eq:fourierofv_new}, we obtain
\[ \widehat{\Psi}_n(k,\ell) = n^2L(k+\ell) \cF_n(f)(n(k+\ell))  \widetilde R_n(k,\ell) \]
where
\[\widetilde R_n(k,\ell)=O(k+\ell) \big(e^{-2i\pi k}-e^{2i \pi k}\big) + \big(1-e^{2i \pi (k+\ell)}\big)  M(k+\ell)\]
with $O$ and $M$ given respectively by 
\begin{align}
\label{eq:defO}
O(y)&:=\cfrac{1}{1-e^{-2i \pi y}}\; \int_{[-\frac12,\frac12]} \Theta (y-x,x) \, dx  
\end{align}
 and \begin{align}
\label{eq:defM}
M(y)&:= \cfrac{1}{1-e^{-2i\pi y}} \; \int_{[-\frac12,\frac12]} e^{-2i \pi x} \, \Theta (y-x,x) \, dx.
\end{align}
Finally, from \eqref{eq:estim} and Lemma \ref{lem:two}, the square of the $\mathbb{L}^2(\mathbb{P})$-norm of 
\[
  \gamma_n \int_0^T \bQ_t^{4,n}\big(  n B_n v_n\big) dt
\]
is bounded from above by 
\begin{equation}
 C \gamma_n^2 n^{\frac52} \int_{[-\frac{1}{2},\frac{1}{2}]} \widetilde U_n(\xi) \big |L(\xi)|^2|\cF_n(f)(n\xi)\big|^2 \; d\xi,  \label{eq:toreplace_new}
\end{equation}
where
\[\widetilde U_n(\xi):=\int_{[-\frac{1}{2},\frac{1}{2}]} \frac{|\widetilde R_n(k,\xi -k)|^2}{z+\sin^2(\pi k)+\sin^2(\pi (\xi-k))} \; dk.  \]
We can bound $|\widetilde R_n(k,\xi-k)|$ from above as follows:
\[|\widetilde R_n(k,\xi -k)|^2\leq 2|\sin(\pi\xi)|^2|M(\xi)|^2 + 2 |\sin(2\pi k)|^2 |O(\xi)|^2. \] 
Let us first estimate the contribution coming from $O(\xi)$, namely: 
\[
\int_{[-\frac{1}{2},\frac{1}{2}]} \frac{|\sin(2\pi k) O(\xi)|^2 }{z+ \sin^2 (\pi k) + \sin^2 (\pi (\xi-k) ) }\; \; dk  \le C|O(\xi)|^2.
\]
Therefore, from Lemma \ref{lem:integral_residues} \eqref{eq:O} the contribution of this term in \eqref{eq:toreplace_new} gives
\[
C\gamma_n^2  \int_{[-\frac{1}{2},\frac{1}{2}]}  n^{\frac52}\sin^2 (\pi \xi)  \big|\mc F_n(f)(n\xi)\big|^2\; d\xi
\] and by Lemma \ref{lem:sfp} it vanishes  independently of $\gamma_n$, since $ \gamma_n^2/\sqrt n \to 0$. 
Now using again Lemma \ref{lem:integral_residues} we estimate the contribution coming from $M(\xi)$, namely: 
\begin{align*}
\int_{[-\frac{1}{2},\frac{1}{2}]} \frac{|\sin(\pi \xi)|^2|M(\xi)|^2 }{z+ \sin^2 (\pi k) + \sin^2 (\pi (\xi-k) ) }\; \; dk 
& \le C|M(\xi)|^2 n^{\frac32}|\sin(\pi \xi)|^2
\end{align*}
Therefore, the contribution of this term in \eqref{eq:toreplace_new} gives
\[
C\gamma_n^2  \int_{[-\frac{1}{2},\frac{1}{2}]}  n^{4}\sin^4 (\pi \xi)  \big|\mc F_n(f)(n\xi)\big|^2\; d\xi
\] and by Lemma \ref{lem:sfp} it vanishes  independently of $\gamma_n$, since $ \gamma_n^2/n \to 0$. 
This ends the proof.

\section*{Acknowledgements}
This work benefited from the support of the project EDNHS  ANR-14-CE25-0011 of the French National Research Agency (ANR) and of the PHC Pessoa Project 37854WM, and also in part by the International Centre for Theoretical Sciences (ICTS) during a visit for participating in the program Non-equilibrium statistical physics (Code: ICTS/Prog-NESP/2015/10).

C.B.~thanks the French National Research Agency (ANR) for its support through the grant ANR-15-CE40-0020-01 (LSD). 

P.G.~thanks  FCT/Portugal for support through the project UID/MAT/04459/2013.  

M.J.~thanks CNPq for its support through the grant 401628/2012-4 and FAPERJ for its support through the grant JCNE E17/2012. M.J.~was partially supported by NWO Gravitation Grant 024.002.003-NETWORKS.  
 
 M.S.~thanks the Labex CEMPI (ANR-11-LABX-0007-01) for its support.
 
This project has received funding from the European Research Council (ERC) under  the European Union's Horizon 2020 research and innovative programme (grant agreement   No 715734).

\appendix

\section{Nonlinear fluctuating hydrodynamics predictions} \label{app:prediction}

The aim of this Appendix is to show that if we perturb the noisy linear Hamiltonian lattice field model by an even potential $V$ then the nonlinear fluctuating hydrodynamics theory predicts that for a zero tension $\tau=0$ the model belongs to the universality class of the harmonic case. We consider an even potential $V$, namely $V(u)=V(-u)$ and we assume moreover that $V$ is non-negative and continuous, with at most polynomial growth at infinity.

\subsection{Cumulants and rules of derivation}

We still denote by $\nu_{\beta,\tau,\gamma}$ the product measure
\begin{equation}\label{eq:nu}
d\nu_{\beta,\tau,\gamma}(\omega):= \prod_{x \in \Z}\frac{e^{-\beta [ e_\gamma (\omega_x) + \tau \omega_x]}}{Z_\gamma (\beta,\tau)} d\omega_x,\qquad \tau \in \R, \beta >0,
\end{equation}
where the energy is now
\[
e_\gamma (u):= \frac{u^2}{2}+\gamma V(u).
\] 
Note that, if $\tau=0$, the density of the marginal of \eqref{eq:nu} at site $x$ with respect to the Lebesgue measure $d\omega_x$ is an even function, and every local function $f$ which is odd has a zero average with respect to $\nu_{\beta,0,\gamma}$. 

Let us denote by $\langle f \rangle_{\beta,\tau,\gamma}$ (resp. $\langle f \, ;\, g \, ; \, h \ldots \rangle_{\beta,\tau,\gamma}$) the average of $f$ (resp. the cumulants between $f,g,h \ldots$) with respect to $\nu_{\beta,\tau,\gamma}$,  and define 
\begin{align*}
{\mf e}_\gamma (\beta,\tau):=\big\langle e_\gamma (\omega_x) \big\rangle_{\beta,\tau,\gamma}\quad \textrm{and}\quad 
{\mf v}_\gamma (\beta,\tau):= \big\langle \omega_x \big\rangle_{\beta,\tau,\gamma}. 
\end{align*}

For any $\gamma>0$ sufficiently small the application \[(\beta, \tau) \in (0,\infty) \times \RR \mapsto ({\mf e}_\gamma (\beta,\tau), {\mf v}_\gamma (\beta,\tau)) \in {\mc U}_\gamma\] is one-to-one for some open subset ${\mc U}_\gamma$. 
In the sequel, we denote $({\mf e}, {\mf v})$ for $({\mf e}_{\gamma} (\beta,\tau)$, ${\mf v}_\gamma (\beta,\tau))$ and by $(\beta, \tau):=(\beta ({\mf e}, {\mf v}), \tau ({\mf e}, {\mf v}))$ the inverse application. 
If $\gamma=0$ then
$$\log Z_0 (\beta, \tau)= \cfrac{1}{2} \log (2\pi) -\cfrac12 \log \beta +\cfrac{\tau^{2} \beta }{2}.$$
If $\gamma$ is small one can easily check that
\begin{equation*}
Z_\gamma (\beta, \tau)= Z_0 (\beta, \tau) \Big(1- \gamma \beta \big\langle V(\omega_x)\big\rangle_{\beta,\tau,0} \Big) + o(\gamma).
\end{equation*}
We can also easily compute
\begin{align*}
\frac{\partial_\gamma Z_\gamma(\beta,\tau)}{Z_{\gamma}(\beta,\tau)} & =-\beta \big\langle V(\omega_0) \big\rangle_{\beta,\tau,\gamma} \qquad
 \frac{\partial_\tau Z_\gamma(\beta,\tau)}{Z_{\gamma}(\beta,\tau)}  =-\beta \big\langle \omega_0 \big\rangle_{\beta,\tau,\gamma} \\ \frac{\partial_\beta Z_\gamma(\beta,\tau)}{Z_{\gamma}(\beta,\tau)}&=- \big\langle e_\gamma(\omega_0)+\tau\omega_0 \big\rangle_{\beta,\tau,\gamma}.
\end{align*}
The computational rules explained in \cite[Appendix 3]{SS}, may be slightly generalized into
\begin{align}
\partial_\gamma \big\langle  {\mc A} (\omega_0) \big\rangle_{\beta,\tau,\gamma} &= -\beta \big\langle {\mc A} (\omega_0) \, ; \,V(\omega_0) \big\rangle_{\beta,\tau,\gamma},\label{eq:crules1}\\
\partial_\tau \big\langle  {\mc A} (\omega_0) \big\rangle_{\beta,\tau,\gamma} &= -\beta \big\langle {\mc A} (\omega_0) \, ; \, \omega_0 \big\rangle_{\beta,\tau,\gamma}, \label{eq:crules2}\\
\partial_\beta \big\langle  {\mc A} (\omega_0) \big\rangle_{\beta,\tau,\gamma} &= - \big\langle {\mc A} (\omega_0) \, ; \, e_{\gamma} (\omega_0) + \tau \omega_0  \big\rangle_{\beta,\tau,\gamma}.
\label{eq:crules3}\end{align}
These derivation rules can be extended to higher-order cumulants as follows: 
\begin{align}
\partial_\gamma \big\langle  {\mc A} (\omega_0) ; \, {\mc B} (\omega_0)  \big\rangle_{\beta,\tau,\gamma}& = -\beta \big\langle {\mc A} (\omega_0)   ; \, {\mc B} (\omega_0) \, ; \, V(\omega_0) \big\rangle_{\beta,\tau,\gamma},\label{eq:crules11}\\
\partial_\tau \big\langle  {\mc A} (\omega_0) ; \, {\mc B} (\omega_0)  \big\rangle_{\beta,\tau,\gamma} &= -\beta \big\langle {\mc A} (\omega_0)  \, ; \, {\mc B} (\omega_0) \, ; \, \omega_0 \big\rangle_{\beta,\tau,\gamma},\label{eq:crules22}\\
\partial_\beta \big\langle  {\mc A} (\omega_0) ; \, {\mc B} (\omega_0)  \big\rangle_{\beta,\tau,\gamma}& = -\big\langle {\mc A} (\omega_0) \, ; \, {\mc B} (\omega_0) \, ; \, e_\gamma (\omega_0) + \tau \omega_0 \big\rangle_{\beta,\tau,\gamma}, \label{eq:crules33}
\end{align} and so on.

The microscopic energy current $j^{e}_{x,x+1}$ of the Hamiltonian part of the dynamics is given by 
$$j^e_{x,x+1}=- \big[ \omega_{x+1} + \gamma V'(\omega_{x+1})\big]  \big[ \omega_{x} + \gamma V'(\omega_{x})\big]$$    
and the microscopic volume current $j^{v}_{x,x+1}$ of the Hamiltonian part of the dynamics is given by
$$j^{v}_{x,x+1} = -\big[ \omega_{x+1} + \gamma V'( \omega_{x+1})\big]  -\big[ \omega_{x} + \gamma V'(\omega_{x})\big].$$
Their averages at equilibrium  are equal to
\[
J^{e} ({\mf e},{\mf v})= \langle j^e_{0,1} \rangle_{\beta,\tau,\gamma}=-{\tau}^2, \qquad J^v ({\mf e},{\mf v})= \langle j^v_{0,1} \rangle_{\beta,\tau,\gamma}=2 \tau.
\]
Observe that for $\gamma=0$, $\tau (\mf e, \mf v)= - \mf v$. For $\gamma>0$, we do not have an explicit formula  for $\tau$ in terms of ${\mf e}$ and ${\mf v}$.

We use the results of \cite{SS} and we refer the reader to this paper for more explanations. When $\tau=0$, it is not difficult to check that \begin{enumerate}
\item the \textit{sound mode} is proportional to the volume field;
\item the  \textit{heat mode} is proportional to the energy field, because we have $\partial_{\mf e} \tau=0$ (see (8.3) of \cite{SS}, or Lemma \ref{lem:tau} below).
\end{enumerate}

\subsection{Computations of coupling constants for any $\gamma \geq 0$}
\label{ssec:expansion}
In this section, we compute some \textit{coupling constants} which are introduced in \cite{SS} and are fundamental to predict the universality classes to which the model belongs. 

To simplify the exposition, we redefine $Y:=\omega_0$ which is distributed according to the probability law
\[ \frac{1}{Z_\gamma(\beta,\tau)}\exp\Big(-\beta \frac{\omega^2}{2} -\beta\gamma V(\omega)-\beta\tau\omega\Big) d\omega.\]
The following lemma is straightforward:
\begin{lem}\label{lem:simp}
Assume that $f(u)$ is an odd function, and that $g(u),h(u)$ are even functions defined on $\bb R$. Then, for any $\beta,\gamma >0$,
\begin{align*}
\big\langle f(Y) \big\rangle_{\beta,0,\gamma} & = 0, \\
\big\langle f(Y)\; ; \; g(Y) \big\rangle_{\beta,0,\gamma} & = 0, \\
\big\langle f(Y)\; ; \; g(Y) \; ; \; h(Y) \big\rangle_{\beta,0,\gamma} & = 0. 
\end{align*} 
Moreover, if $f,g,h$ are all odd functions on $\bb R$, then 
\[\big\langle f(Y)\; ; \; g(Y) \; ; \; h(Y) \big\rangle_{\beta,0,\gamma}  = 0.\]
\end{lem}

\begin{rem} \label{rem:gauss}
Note that, if $\tau=0$, and $\gamma=0$, then $Y=\frac{1}{\sqrt \beta}G$, where $G$ is a standard Gaussian variable.
\end{rem}

\begin{lem}[Continuity] \label{lem:continuity}
For any $f_1, f_2, \dots, f_k$ functions defined on $\bb R$, the map
\begin{align*}
\phi: \bb R_+ & \to \bb R \\
\gamma & \mapsto \big\langle f_1(Y) \; ; \; f_2(Y) \; ;\; \dots \; ; \; f_k(Y) \big\rangle_{\beta,0,\gamma}
\end{align*}
is continuous. 
\end{lem}

\begin{proof}
This is an easy consequence of the dominated convergence Theorem. 
\end{proof}

\subsubsection{Gamma function} 
Let us define
\begin{align*}
&\Gamma:=\Gamma^\gamma (\beta, \tau)= \beta \Big( \big\langle Y \; ;\; Y \big\rangle_{\beta,\tau,\gamma} \big\langle e_\gamma (Y)\; ; \; e_\gamma (Y) \big\rangle_{\beta,\tau,\gamma} -  \big\langle Y\; ;\; e_\gamma (Y) \big\rangle_{\beta,\tau,\gamma}^2 \Big)\\
& = \tfrac \beta 4 \Big( \big\langle Y\, ; \, Y \big\rangle_{\beta,\tau,\gamma} \big\langle Y^2+2\gamma V(Y)\, ; \,  Y^2+2\gamma V(Y)\big\rangle_{\beta,\tau,\gamma} - \big\langle Y\, ; \, Y^2+2\gamma V(Y)\big\rangle_{\beta,\tau,\gamma}^2\Big).
\end{align*}
To simplify notation from now on we write $e_\gamma$ for $e_\gamma (Y)$.
From the computational rules \eqref{eq:crules22} and \eqref{eq:crules33}, the derivatives of $\Gamma^\gamma$ are given by 
\begin{align*}
\partial_\tau \Gamma^\gamma(\beta,\tau) = & \; \tfrac{\beta^2}{4}\Big( - \big\langle Y\, ; \, Y \, ; \, Y \big\rangle_{\beta,\tau,\gamma} \big\langle Y^2+2 \gamma V(Y) \, ; \, Y^2+2 \gamma V(Y) \big\rangle_{\beta,\tau,\gamma} \\& \qquad
- \big\langle Y\, ; \, Y \big\rangle_{\beta,\tau,\gamma} \big\langle Y^2+2\gamma V(Y) \, ; \, Y^2+2\gamma  V(Y) \, ; \, Y \big\rangle_{\beta,\tau,\gamma} \vphantom{\bigg(} \\
&\qquad + 2 \big\langle Y\, ; \, Y^2+2\gamma V(Y) \, ; \, Y \big\rangle_{\beta,\tau,\gamma} \big\langle Y\, ; \, Y^2+2\gamma V(Y)\big\rangle_{\beta,\tau,\gamma} \Big),\vphantom{\bigg(}
\\
\vphantom{\bigg(}\partial_\beta \Gamma^\gamma(\beta,\tau) =  &\; \beta^{-1}\; \Gamma^\gamma(\beta,\tau) \\
 + \tfrac{\beta}{8}\Big( - \big\langle Y\, ; \, &Y \, ; \, Y^2 +2\gamma V(Y) {+2\tau Y} \big\rangle_{\beta,\tau,\gamma} \big\langle Y^2+2\gamma  V(Y) \, ; \, Y^2+2\gamma V(Y) \big\rangle_{\beta,\tau,\gamma} \\ 
- \big\langle Y\, ; \, &Y \big\rangle_{\beta,\tau,\gamma} \big\langle Y^2+2\gamma V(Y) \, ; \, Y^2+2\gamma V(Y) \, ; \, Y^2+2\gamma V(Y) {+2\tau Y}\big\rangle_{\beta,\tau,\gamma}  \vphantom{\bigg(} \\
+ 2 \big\langle Y\, &; \, Y^2+2\gamma V(Y) \, ; \, Y^2+2\gamma V(Y)  {+2\tau Y}\big\rangle_{\beta,\tau,\gamma} \big\langle Y\, ; \, Y^2+2\gamma V(Y)\big\rangle_{\beta,\tau,\gamma} \Big).
\end{align*}
In particular, for the value $\tau=0$, these three expressions simplify significantly thanks to Lemma \ref{lem:simp} into
\begin{align*}
\Gamma^\gamma(\beta,0) & = \tfrac \beta 4  \big\langle Y\, ; \, Y \big\rangle_{\beta,0,\gamma} \big\langle Y^2+2\gamma V(Y)\, ; \, Y^2+2 \gamma V(Y)\big\rangle_{\beta,0,\gamma},\\\partial_\tau \Gamma^\gamma(\beta,0) & \equiv 0, \vphantom{\Bigg\{} \\
\partial_\beta \Gamma^\gamma(\beta,0)& =  \beta^{-1}\;\Gamma^\gamma(\beta,0) \\
&  + \tfrac{\beta}{8}\Big( - \big\langle Y\, ; \, Y \, ; \, Y^2+2\gamma V(Y) \big\rangle_{\beta,0,\gamma} \big\langle Y^2+2\gamma V(Y) \, ; \, Y^2+2\gamma  V(Y) \big\rangle_{\beta,0,\gamma} \\
& 
- \big\langle Y\, ; \, Y \big\rangle_{\beta,0,\gamma} \big\langle Y^2+2\gamma V(Y) \, ; \, Y^2+2\gamma V(Y)\, ; \, Y^2+2\gamma  V(Y) \big\rangle_{\beta,0,\gamma} \Big). \vphantom{\Bigg(} 
\end{align*}
In particular Lemma \ref{lem:simp} and Lemma \ref{lem:continuity} directly imply the following

\begin{lem}[Derivative of $\Gamma$ and its inverse at $\tau=0$] 
\label{lem:gamma}
The map $\gamma \mapsto \Gamma^\gamma(\beta,\tau)$ is continuous, and moreover
\[
\partial_\tau \Gamma^\gamma (\beta,0)  \equiv 0, \qquad 
 \partial_\tau \Big(\frac{1}{\Gamma^\gamma}\Big) (\beta,0) \equiv 0. \]

\end{lem}

\subsubsection{Tension and its derivatives}
Now let us compute the derivatives of $\tau$ with respect to ${\mf e}$ and $\mf v$ when $\tau({\mf e},{\mf v}) = 0$. 
Equation (8.3) of \cite{SS} gives
\begin{align}
\label{eq:dertau0}
\partial_{\mf v} \tau& = - \frac{1}\Gamma \big\langle e_\gamma(\omega_0)  \, ; \, e_\gamma(\omega_0)  + \tau \omega_0 \big\rangle_{\beta, \tau,\gamma}, \\ 
\label{eq:dertau2.0}
\partial_{\mf e} \tau& = \frac 1 \Gamma \big\langle \omega_0 \, ; \, e_\gamma(\omega_0) + \tau \omega_0 \big\rangle_{\beta, \tau,\gamma},
\end{align} 
which read with our notations
\begin{align}
\label{eq:dertau}
\partial_{\mf v} \tau& = - \frac{1}{4\Gamma^\gamma(\beta,\tau)} \big\langle Y^2+2\gamma V(Y)  \, ; \, Y^2+2 \gamma V(Y)  +2 \tau Y\big\rangle_{\beta, \tau,\gamma}, \\ 
\label{eq:dertau2}
\partial_{\mf e} \tau& = \frac 1 {2\Gamma^\gamma(\beta,\tau)} \big\langle Y \, ; \, Y^2+2\gamma V(Y) + 2\tau Y \big\rangle_{\beta, \tau,\gamma}.
\end{align}
For $\tau({\mf e}, {\mf v})=0$, it is then trivial that \begin{align*} \partial_{\mf e} \tau \big|_{\tau=0} &\equiv 0, \\
\partial_{\mf v}\tau\big|_{\tau=0} & = -\frac{1}{4\Gamma^\gamma(\beta,0)}\big\langle Y^2+ 2 \gamma V(Y)  \, ; \, Y^2+2\gamma V(Y)  \big\rangle_{\beta, 0,\gamma}.\end{align*}
One can even goes further, and compute the second order derivatives of $\tau$ as follows: 
\begin{lem} \label{lem:tau}
If $({\mf e},{\mf v})$ are such that $\tau({\mf e},{\mf v})=0$, we have that
\begin{align}
\tau=0, \quad  \partial_{\mf e} \tau\big|_{\tau=0} \equiv 0, \quad 
 \partial_{\mf v}^2 \tau\big|_{\tau=0} \equiv 0,\quad 
 \partial_{\mf e}^2 \tau\big|_{\tau=0} \equiv 0. \notag 
\end{align}
\end{lem}

\begin{proof}
The second derivatives of $\tau$ are evaluated as explained in \cite[Appendix 3.2]{SS}. First, we compute the derivatives of $\mf v$ and $\mf e$ with respect to $\tau$ and $\beta$, and then we use the Jacobian inversion. 
Recall that, in our notation
\[{\mf v} (\beta, \tau) = \big\langle Y \big\rangle_{\beta,\tau,\gamma}, \qquad
{\mf e} (\beta, \tau) = \tfrac12\big\langle Y^2 + 2\gamma  V(Y)  \big\rangle_{\beta,\tau,\gamma}.
\]
Therefore  at $\tau=0$, $\mf v(\beta,0)  \equiv 0$.
We deduce from the derivation rules \eqref{eq:crules2}  that 
\begin{align*}
\partial_{\tau} {\mf v}&= - \beta \big\langle Y \, ; \, Y\big\rangle_{\beta,\tau,\gamma}, \\ 
\partial_\beta {\mf v}&=- \tfrac12 \big\langle Y\, ; \, Y^2+2\gamma V(Y) + 2\tau Y\big\rangle_{\beta,\tau,\gamma}  
\end{align*}
and from Lemma \ref{lem:simp} we directly deduce
\begin{align}  
\partial_\tau{\mf v\big|_{\tau=0}} & \neq 0, \label{eq:dtauv}\\
\partial_\beta{\mf v}\big|_{\tau=0}  & \equiv 0. \label{eq:dbetav}\end{align}
In the same way 
\begin{align*}
\partial_{\tau} {\mf e}&=-\beta \big\langle Y\, ; \, Y^2+2\gamma V(Y) \big\rangle_{\beta,\tau,\gamma },
\\
\partial_\beta {\mf e}&=-\tfrac14\big\langle Y^2+2\gamma V(Y)  \, ;\, Y^2+2\gamma  V(Y)  +2\tau Y \big\rangle_{\beta,\tau,0} 
\end{align*}
and, in particular, for $\tau=0$ we get 
\begin{align}
\partial_{\tau} {\mf e} \big|_{\tau=0} & \equiv 0, \label{eq:dtaue}\\
\partial_{\beta} {\mf e} \big|_{\tau=0}& \neq 0. \label{eq:dbetae}
\end{align}
By using \eqref{eq:dertau}, \eqref{eq:dertau2}, \eqref{eq:crules2}, Lemma \ref{lem:simp} and Lemma \ref{lem:gamma}, we get that 
\begin{equation*}
\begin{split}
\partial_\tau (\partial_{\mf v} \tau) \big|_{\tau=0} &= -\frac14 \partial_\tau\Big(\frac{1}{\Gamma^\gamma}\Big)(\beta,0)\big\langle Y^2+2\gamma V(Y)  \, ;\, Y^2+2\gamma V(Y)  \big\rangle_{\beta,0,\gamma} \\
& \quad + \frac{\beta}{4\Gamma^\gamma(\beta,0)}  \big\langle Y \, ; \, Y^2+2\gamma V(Y)  \, ;\, Y^2+2\gamma V(Y)  \big\rangle_{\beta, 0,\gamma}
\equiv 0. \vphantom{\Bigg(}
\end{split}
\end{equation*}
Similarly,
\begin{align} 
\partial_\tau (\partial_{\mf e} \tau) \big|_{\tau =0}&= 
\frac12\partial_\tau\Big(\frac{1}{\Gamma^\gamma}\Big)(\beta,0) \big\langle Y\, ;\, Y^2+2\gamma V(Y)\big\rangle_{\beta,0,\gamma} + \frac{1}{\Gamma^\gamma(\beta,0)} \big\langle Y\, ; \, Y\big\rangle_{\beta,0,\gamma} \notag  \\ 
&  \quad - \frac{\beta}{2\Gamma^\gamma(\beta,0)} \big\langle Y\, ; \, Y^2+2\gamma V(Y) \, ; \, Y\big\rangle_{\beta,0,\gamma}. \label{eq:dertauderetau}
\end{align}
Finally, combining \eqref{eq:crules3}, Lemma \ref{lem:simp} and Lemma \ref{lem:gamma} we have
\begin{align}
\partial_\beta(\partial_{\mf v} \tau)\big|_{\tau=0} & = -\frac14\partial_\beta\Big(\frac{1}{\Gamma^\gamma}\Big)(\beta,0) \big\langle Y^2+2\gamma V(Y)\, ; \, Y^2+2\gamma V(Y) \big\rangle_{\beta,0,\gamma}\notag \\ & \quad + \frac1{8\Gamma^\gamma(\beta,0)} \big\langle Y^2+2\gamma V(Y)\, ; \,Y^2+2\gamma V(Y) \, ; \, Y^2+2\gamma V(Y) \big\rangle_{\beta,0,\gamma} \label{eq:derbetadervtau}
\end{align}
and 
\begin{align*}
\partial_\beta (\partial_{\mf e} \tau)\big|_{\tau=0}  
& = \frac{1}{2}\partial_\beta\Big(\frac{1}{\Gamma^\gamma}\Big)(\beta,0)\big\langle Y\, ;\, Y^2+2\gamma V(Y)\big\rangle_{\beta,0,\gamma} \\ 
& \quad - \frac{1}{4\Gamma^\gamma(\beta,0)} \big\langle Y\, ; \, Y^2+ 2\gamma V(Y)\, ; \, Y^2+2\gamma V(Y) \big\rangle_{\beta,0,\gamma}  \equiv 0.\vphantom{\Bigg(}
\end{align*}
We deduce from \eqref{eq:dtauv} and \eqref{eq:dbetae} that
\[
\frac{1}{\partial_\tau{\mf v}}\bigg|_{\tau=0}  \neq 0, \qquad
\frac{1}{\partial_\beta{\mf e}}\bigg|_{\tau=0}  \neq 0.
\]
Recall the values \eqref{eq:dtauv},\eqref{eq:dbetav}, \eqref{eq:dtaue}, \eqref{eq:dbetae}. Then, the values of the second derivatives of $\tau$ are given by 
\begin{align*}\begin{pmatrix} 
\partial_{\mf v}^2 \tau\\
\partial_{\mf e \mf v}^2 \tau
\end{pmatrix}\bigg|_{\tau=0}
&=
\begin{pmatrix}
\partial_{\tau}\mf v & \partial_{\tau} {\mf e}\\
\partial_{\beta} {\mf v} & \partial_\beta {\mf e} 
\end{pmatrix}\bigg|_{\tau=0}^{-1} 
\begin{pmatrix}
\partial_\tau (\partial_{\mf v} \tau)\\
\partial_\beta (\partial_{\mf v} \tau)
\end{pmatrix}\bigg|_{\tau=0} \\ 
& = \begin{pmatrix}
\Big[\partial_{\tau}\mf v\big|_{\tau=0}\Big]^{-1} & 0 \\
0 & \Big[\partial_\beta \mf e\big|_{\tau = 0}\Big]^{-1}
\end{pmatrix} \begin{pmatrix}
0 \\\partial_\beta (\partial_{\mf v} \tau)\big|_{\tau=0}
\end{pmatrix},  \notag \\
~\\ 
\begin{pmatrix}\partial_{\mf v \mf e}^2 \tau\\
\partial_{\mf e}^2 \tau
\end{pmatrix}\bigg|_{\tau=0} & 
= \begin{pmatrix}
\partial_{\tau}\mf v & \partial_{\tau} {\mf e}\\
\partial_{\beta} {\mf v} & \partial_\beta {\mf e} 
\end{pmatrix}\bigg|_{\tau=0}^{-1} 
\begin{pmatrix}
\partial_\tau (\partial_{\mf e} \tau)\\
\partial_\beta (\partial_{\mf e} \tau)
\end{pmatrix}\bigg|_{\tau=0} \\  & = 
\begin{pmatrix}
\Big[\partial_{\tau}\mf v\big|_{\tau=0}\Big]^{-1} & 0 \\
0 & \Big[\partial_\beta \mf e\big|_{\tau = 0}\Big]^{-1}
\end{pmatrix}
\begin{pmatrix}
\partial_\tau (\partial_{\mf e} \tau)\big|_{\tau=0}\\
0
\end{pmatrix} 
\end{align*}
and therefore $ \partial_{\mf v}^2 \tau\big|_{\tau=0} \equiv 0, $ and $
\partial_{\mf e}^2 \tau\big|_{\tau=0} \equiv 0$.
\end{proof}

From \cite{SS}, the sound mode (mode 1) has velocity 
\[c(\beta,\tau,\gamma)= -\frac{2}{\Gamma^\gamma(\beta,\tau)}\big\langle \tau\omega_0 + e_\gamma(\omega_0) \; ; \; \tau\omega_0 + e_\gamma(\omega_0) \big\rangle_{\beta,\tau,\gamma}= 2(\partial_{\mf v} - \tau \partial_{\mf e}) \tau,\]  and the heat mode (mode 2) has velocity $0$. The coupling constants $G^1_{\alpha\alpha'}$ and $G^2_{\alpha \alpha'}$, $\alpha, \alpha' \in \{1,2\}$ determine the universality class of the model. In the case considered here, we have that
\begin{equation*}
G^1_{11}=-\Big(\cfrac{-2}{\beta c}\Big)^{\frac12}\; \big [ \partial_{\mf v} - \tau \partial_{\mf e} \big]^2 \tau.
\end{equation*}
When $\tau=0$, the sound mode has velocity $c(\beta,0,\gamma):= 2\partial_{\mf v} \tau\big|_{\tau=0} $
and from Lemma \ref{lem:tau} we have
$$G_{11}^1 =- \Big(\cfrac{-2}{\beta c}\Big)^\frac12\; \partial_{\mf v}^2 \tau\big|_{\tau=0} \equiv 0.$$
Therefore, according to \cite[Section 2.2]{SS}, there are four possibilites: 
\begin{enumerate}[1.]
\item If $G_{22}^1=0$ and $G_{11}^2 \neq 0$, the sound mode is diffusive and the heat mode is L\'evy with exponent $\frac{3}{2}$;
\item If $G_{22}^1=0$ and $G_{11}^2 = 0$, the sound mode and the heat mode are diffusive;
\item If $G_{22}^1\neq 0$ and $G_{11}^2 \neq 0$, the sound mode and the heat mode are Gold-L\'evy;
\item If $G_{22}^1\neq 0$ and $G_{11}^2 = 0$, the sound mode is L\'evy with exponent $\frac{3}{2}$ and the heat mode is diffusive.
\end{enumerate}

In the following section we prove that our model belongs to the first case.

\subsection{Coupling matrices}
Let us introduce some definitions and notations, taken from \cite{SS}.

\begin{de} \label{def:coupling}
~

\begin{enumerate}[1.]
\item \textsc{Constants:} let us denote
\begin{align*}
c&:=2(\partial_{\mf v}-\tau \partial_{\mf e})\tau, \\
Z_1&:=-\sqrt{\tfrac{-\beta c}{2}} \quad \text{and} \quad Z_2:=\sqrt{\tfrac{-c}{2\Gamma}}, \\
\widetilde Z_1 & := \sqrt{\tfrac{-c}{2\beta}} \quad \text{and} \quad \widetilde Z_2 := \sqrt{\tfrac{-\Gamma c}{2}}.
\end{align*}
\item \textsc{Vectors:} let us denote
\[ \psi_1:=\frac{1}{Z_1}\begin{pmatrix}
1\\-\tau
\end{pmatrix} \quad \text{and}\quad \psi_2:=\frac{1}{Z_2}\begin{pmatrix}
\partial_{\mf e} \tau \\ -\partial_{\mf v}\tau
\end{pmatrix}. \]
\item \textsc{Matrices:} let us denote
\begin{align*}
R&:=\begin{pmatrix}
(\widetilde Z_1)^{-1}\; \partial_{\mf v}\tau & (\widetilde Z_1)^{-1}\; \partial_{\mf e}\tau \\
(\widetilde Z_2)^{-1}\; \tau & (\widetilde Z_2)^{-1}
\end{pmatrix}, \\
H_{\mf v}&:=2\begin{pmatrix}
\partial^2_{\mf v}\tau & \partial_{\mf v}\partial_{\mf e} \tau \\
\partial_{\mf v}\partial_{\mf e}\tau & \partial_{\mf e}^2\tau
\end{pmatrix},\\
H_{\mf e}&:=-\tau H_{\mf v} - 2 \begin{pmatrix}
(\partial_{\mf v}\tau)^2 & \partial_{\mf v}\tau \; \partial_{\mf e}\tau \\ 
\partial_{\mf v}\tau \; \partial_{\mf e}\tau & (\partial_{\mf e}\tau)^2
\end{pmatrix}.
\end{align*}
\end{enumerate}
\end{de}

With Definition \ref{def:coupling}, we are now able to define the \textit{coupling matrices} $G^1$ and $G^2$ as follows: 
\begin{align}
G^1_{\alpha\alpha'}&:= \frac{1}{2}\Big( R_{11} (\psi_{\alpha}^T \cdot H_{\mf v} \psi_{\alpha'}) + R_{12} (\psi_{\alpha}^T \cdot H_{\mf e} \psi_{\alpha'})\Big) \label{eq:G1}\\
 G^2_{\alpha\alpha'}&:= \frac{1}{2}\Big( R_{21} (\psi_{\alpha}^T \cdot H_{\mf v} \psi_{\alpha'}) + R_{22} (\psi_{\alpha}^T \cdot H_{\mf e} \psi_{\alpha'})\Big) \label{eq:G2}
\end{align}

A first corollary of Lemma \ref{lem:continuity} is the following:

\begin{cor}[Continuity of the coupling constants]\label{cor:continuity}
All the constants that are defined in Definition \ref{def:coupling} and also in \eqref{eq:G1} and \eqref{eq:G2} are continuous functions of $\gamma \in \bb R_+$. 
\end{cor}

We now give the values of each quantity that appears in \eqref{eq:G1} and \eqref{eq:G2}, taken first at $\tau=0$ and $\gamma = 0$.

\begin{prop}[Without anharmonicity, $\gamma=0$]
\label{prop:taylor} If $(\mf e,\mf v)$ are such that $\tau(\mf e,\mf v)=0$, and if $\gamma=0$ then we have
\begin{enumerate}[1.]
\item \textsc{Constants:}
\begin{align*}
c&=-2,\\
Z_1 & = -\sqrt\beta \quad \textrm{and}  \quad Z_2=\sqrt{2} \beta,\\
\widetilde Z_1 & = -\frac{1}{\beta}Z_1 =  -\frac{1}{\sqrt\beta} \quad \textrm{and} \quad  
\widetilde{Z}_2=\frac{1}{\sqrt 2 \beta}.
\end{align*}
\item \textsc{Vectors:}
\[
\psi_1= \begin{pmatrix}
-\frac{1}{\sqrt \beta} \\ 0
\end{pmatrix} \quad \textrm{and}\quad  
\psi_2 = \begin{pmatrix}
0 \\ \frac{1}{\sqrt 2 \beta}
\end{pmatrix}.
\]
\item \textsc{Matrices:}
\begin{align*}
R= \begin{pmatrix}
-\sqrt \beta & 0 \\
  0 & \sqrt 2 \beta 
\end{pmatrix},  \qquad
H_{\mf v} & = \begin{pmatrix}
0 & 0 \\
0 & 0
\end{pmatrix},\qquad
H_{\mf e}  = \begin{pmatrix}
-2 & 0 \\
0 & 0 
\end{pmatrix}.
\end{align*} 
\end{enumerate}
\end{prop}

\begin{proof}
This proposition follows from easy computations, using Remark \ref{rem:gauss}, and the following straightforward lemma: 
\begin{lem}
\label{lem:cumul}
Let $G$ be a standard Gaussian variable of mean zero and variance 1. Then we have that all the odd moments of $G$ are zero and
\[
\langle G^{2n} \rangle = \frac{(2n)!}{n! \; 2^n}.
\]
We also have
\[
\langle G^{p_1} \; ; \; G^{p_2} \; ;\; \cdots \; ; \; G^{p_k} \rangle = 0 \qquad \text{ as soon as }  \quad p_1+\cdots +p_k \text{ is odd}. 
\]
and
\begin{equation*}
\begin{split} & \langle G \, ; \, G \rangle= 1, \quad \qquad \qquad \langle G^2 \, ; \, G^2\rangle=2, \\ 
 &\langle G \, ; \, G\, ; \, G^2 \rangle=2, \qquad \quad \langle G^2 \, ; \, G^2\, ; \, G^2 \rangle=8. \end{split}
\end{equation*}
\end{lem}
\end{proof}

Therefore, from Proposition \ref{prop:taylor} we conclude that
\[ G_{11}^2 \big|_{\tau=0,\gamma=0} = - \sqrt 2 \neq 0.\] 
Since the map 
$\gamma \mapsto G_{11}^2$ is continuous in $\gamma \geq 0$, we conclude that  there exists $\gamma_0 >0$ such that, for any $\gamma \leq \gamma_0$, $G_{11}^2 \neq 0$. 

It remains to compute $G_{22}^1$, which is given by the following:

\begin{prop}
Assume that $(\mf e,\mf v)$ are such that $\tau(\mf e,\mf v)=0$. Then for any $\gamma \geq 0$, 
\[ G_{22}^1 = 0.\] 
\end{prop}

\begin{proof}
We let the reader check, using the proof of Lemma \ref{lem:tau}, that for any $\gamma>0$, 
\begin{enumerate}[1.]
\item the diagonal coefficients of $H_{\mf v}$ are equal to 0;
\item the off-diagonal coefficients of $H_{\mf e}$ are equal to 0, as well as the second diagonal coefficient. In other words, the unique non-zero coefficient of $H_{\mf e}$ is the first one which is equal to $-2 (\partial_{\mf  v}\tau)^2$;
\item the off-diagonal of $R$ are equal to 0;
\item the second component of $\psi_1$ is equal to 0;
\item the first component of $\psi_2$ is equal to 0.
\end{enumerate}
Then, by computing the matrix product appearing in \eqref{eq:G1}, the result follows.
\end{proof}

\section{Estimates on the Poisson equation}
\label{sec:poisson}

In this appendix we prove Proposition \ref{prop:estimate2} and Proposition \ref{prop:estimatevn}.  Several times we will use the following change of variable property proved in \cite{BGJ} and that we recall here.

\begin{lem}
\label{lem:cov}
Let $f: \R^2 \to {\mathbb C}$ be a $n$-periodic function in each direction of $\R^2$. Then we have
\begin{equation*}
\int\int_{[-\frac{n}{2}, \frac{n}{2}]^2} f (k,\ell) \, dk d\ell \; = \; \int \int_{[-\frac{n}{2}, \frac{n}{2}]^2} f(\xi- \ell, \ell) \, d\xi d\ell.
\end{equation*}
\end{lem}

Another useful lemma is
\begin{lem}
\label{lem:sfp}
If $g \in {\mc S} (\R)$, then for any $p\geq 1$, there exists a constant $C:=C(p,g)$ such that for any $|y| \leq \frac12$,
\begin{equation*}
\big| \cF_n(g) (ny) \big|^2 \leq \frac{C}{1+ (n|y|)^{p}}.
\end{equation*}
\end{lem}

The following result is an easy corollary of the previous lemma.
\begin{cor}\label{cor:fourier}
If $g \in {\mc S}(\R)$, then, for any $p\geq 0$,
\[ \lim_{n\to\infty} \int_{[-\frac{n}{2},\frac{n}{2}]} |\xi|^p\, \left\vert \cF_n(g)(\xi)-(\mc F g)(\xi)\right\vert^2 {d}\xi =0,\]
and there exists a constant $C>0$ such that
\[ \int_{[-\frac{n}{2},\frac{n}{2}]} \vert \xi\vert^p \,\vert \cF_n(g)(\xi) \vert^2 {d}\xi \leq C.\]
\end{cor}

We start with some estimates concerning the solution $h_n$ of \eqref{eq:poisson}, and then we treat the solution $v_n$ of \eqref{eq:v}. All technical estimates involving integral calculus are detailed in the next section, see Appendix \ref{app:integral}.

\subsection{Proof of Proposition \ref{prop:estimate2}} 

Let us recall the explicit expression for the Fourier transform of  $h_n$ given in  \eqref{eq:fourierofh}. 

\subsubsection{Proof of  \eqref{eq:est0}}
\label{ssec:l2ofh}
From the Parseval-Plancherel relation and from \eqref{eq:fourierofh}, we have  that
\begin{align*}
\tfrac{1}{n^2}\sum_{x,y}&h_n^2\big(\tfrac x n,\tfrac y n\big) 
		\\&=\iint_{[-\frac{n}{2}, \frac{n}{2}]^2} | \mc F_n(h_n) (k, \ell) |^2 dk d\ell \\
		& = \frac{(1+\kappa_n\gamma_n)^4}{4n} \iint_{[-\frac{n}{2}, \frac{n}{2}]^2}\frac{\Omega^2\big( \tfrac{k}{n}, \tfrac{\ell}{n}\big) \, |{\mc F_n(f)} (k+\ell)|^2}{\left(1+\kappa_n\gamma_n\right)^2\Lambda^2 \big( \tfrac{k}{n}, \tfrac{\ell}{n}\big) + \Omega^2 \big( \tfrac{k}{n}, \tfrac{\ell}{n}\big)} \; dk d\ell \\
		&\lesssim \frac{1}{n} \int_{[-\frac{n}{2}, \frac{n}{2}]} \big|1 - e^{\tfrac{2i\pi\xi}{n}}\big|^2\; \big|\mc F_n(f) (\xi)\big|^2 \; \left[ \int_{[-\frac{n}{2}, \frac{n}{2}]} \frac{d\ell}{\Lambda^2 \big( \tfrac{\xi - \ell}{n}, \tfrac{\vphantom{\xi}\ell}{n}\big ) + \Omega^2 \big( \tfrac{\xi - \ell}{n}, \tfrac{\vphantom{\xi}\ell}{n}\big)} \right] \; d\xi \\
		&\approx n \int_{[-\frac{1}{2}, \frac{1}{2}]^2} \sin^2 (\pi y) |\mc F_n(f) (ny) |^2\; W(y) dy,
\end{align*}
where for the last equality we performed the changes of variables $y=\frac{\xi}{n}$ and $x= \frac{\ell}{n}$. The function $W$ is defined by
\begin{equation}
\label{eq:W}
W(y) =  \int_{[-\frac{1}{2}, \frac{1}{2}]}  \frac{dx}{\Lambda^2(y-x,x) + \Omega^2(y-x, x)}.
\end{equation}
It is proved in \cite[Lemma F.5]{BGJ} that $W(y) \leq C |y|^{-\frac32}$ on $[-\frac{1}{2}, \frac{1}{2}]$. Hence, we get, by using  Lemma \ref{lem:sfp} with $p=3$ and the elementary inequality $\sin^2 (\pi y) \leq \pi^2 y^2$, that
\[
\iint_{[-\frac{n}{2}, \frac{n}{2}]^2} | \mc F_n(h_n) (k, \ell) |^2 dk d\ell \; \lesssim\; n \int_{[-\frac{1}{2}, \frac{1}{2}]} \cfrac{|y|^{\frac12}}{1+ (n|y|)^3} dy \; \lesssim \; n^{-\frac12},
\]
which proves \eqref{eq:est0}.

\subsubsection{Proof of  \eqref{eq:est0x}}
\label{ssec:diag}

 The Plancherel-Parseval equality gives 
\[
\tfrac{1}{n}\sum_{x\in\Z}h_n^2\big(\tfrac{x}{n},\tfrac{x}{n}\big)=\int_{[-\frac{n}{2},\frac{n}{2}]} \big| \cF_n({{\overline h}}_n)(\xi) \big|^2 d\xi,
\]
where 
$
{\overline h}_n\big(\tfrac x n\big)=h_n\big(\tfrac x n,\tfrac{x}{n}\big), $ and then $\cF_n(\overline{h}_n)(\xi)=\tfrac{1}{n}\sum_{x\in\Z}h_n\big(\tfrac x n,\tfrac{x}{n}\big)e^{2i\pi x\frac\xi n}.
$
By definition,
\begin{align*}
\cF_n(\overline{h}_n)(\xi)&=\frac{1}{n}\sum_{x\in\Z}\bigg( \iint_{[-\frac{n}{2},\frac{n}{2}]^2}dkd\ell \; 
\cF_n(h_n)(k,\ell)e^{-2i\pi x \frac{k+\ell}{n}} e^{2i\pi x\frac{\xi}{n}}\bigg)\\
& = \int_{[-\frac{n}{2},\frac{n}{2}]}d\ell \; \cF_n(h_n)(\xi-\ell,\ell).
\end{align*}
By \eqref{eq:fourierofh}, we compute: 
\begin{align}
\cF_n(\overline{h}_n)(\xi) & = \frac{(1+\kappa_n\gamma_n)^2}{2\sqrt n} \int_{[-\frac{n}{2},\frac{n}{2}]}d\ell \; \frac{i\Omega(\tfrac{\xi-\ell}{n},\tfrac{\ell}{n})}{((1+\gamma_n\kappa_n)\Lambda-i\Omega)(\tfrac{\xi-\ell}{n},\tfrac{\ell}{n})}\cF_n(f)(\xi)\notag\vphantom{\Bigg(}\\
& = \vphantom{\Bigg(}\frac{(1+\kappa_n\gamma_n)^2}{2}\; \sqrt n\; I\Big(\frac{\xi}{n}\Big) \cF_n(f)(\xi),\label{eq:fourier1}
\end{align}
where $I$ has already been defined in \eqref{eq:defI}. 
From Lemma \ref{lem:integral_residues}, we have that
$|I(y)| \; \lesssim \; |\sin(\pi y)|^{\frac12},$
and consequently, from \eqref{eq:fourier1} we get
\[
\big|\cF_n({{\overline h}}_n)(\xi)\big|^2  \; \lesssim \;  n \frac{|\xi|}{n} \big|\cF_n(f)(\xi)\big|^2,
\]
and using Corollary \ref{cor:fourier}
\[
\tfrac{1}{n}\sum_{x\in\Z}h_n^2\big(\tfrac{x}{n},\tfrac{x}{n}\big)  \; \lesssim \; \int_{[-\frac{n}{2},\frac{n}{2}]}  d\xi\; |\xi|\; \big| \cF_n(f)(\xi) \big|^2 \; \lesssim \; 1,
\]
which proves \eqref{eq:est0x}.

\subsubsection{Proof of \eqref{eq:est1}.}
Recall that $q_n : {\frac{1}{n}} \Z \to \R$ is the function defined by
\[q_n \big(\tfrac{x}{n}\big) = {{\mc D}}_n h_n \,  \big(\tfrac{x}{n}\big).
\]
We already proved in Section \ref{ssec:main} that
\begin{equation*}
\mc F_n({q_n}) (\xi) = \frac{(1+\gamma_n\kappa_n)^2}{4} \; n^{\frac32}\; K\Big(\frac{\xi}{n}\Big)\; \mc F_n(f) (\xi),
\end{equation*}
where $K$ is defined in \eqref{eq:defK}. 
From Lemma \ref{lem:integral_residues} we get
\[\big|\mc F_n(q_n)(\xi) \big|^2 \; \lesssim \; |\xi|^3 \big|\mc F_n(f)(\xi) \big|^2 \] 
and therefore
\[\tfrac{1}{n} \sum_{x\in\bb Z} \big[\mc D_nh_n\big(\tfrac{x}{n}\big)\big]^2 \; \lesssim \; \int_{[-\frac{n}{2},\frac{n}{2}]}  d\xi\; |\xi|^3\; \big| \cF_n(f)(\xi) \big|^2 \; \lesssim \; 1,
\]
using again Corollary \ref{cor:fourier}. This proves \eqref{eq:est1}.

\subsubsection{Proof of  \eqref{eq:est4}.}
\label{ssec:grad} 

By Parseval-Plancherel's relation we have that
\begin{equation*}
\tfrac{1}{n} \sum_{x \in \ZZ} \big[h_n\big(\tfrac{x+1}{n},\tfrac{x+1}{n}\big)- h_n\big(\tfrac{x}{n},\tfrac{x}{n}\big)\big]^2= \int_{[-\frac{n}{2},\frac{n}{2}]} \big|1- e^{2i\pi \frac{\xi}{n}}\big|^2 \big|\cF_n({{\overline h}}_n)(\xi) \big|^2 d\xi
\end{equation*}
so that by \eqref{eq:fourier1}, together with Lemma \ref{lem:integral_residues} and Corollary \ref{cor:fourier}, we get that
\begin{align*}\tfrac{1}{n} \sum_{x \in \ZZ} &\big[h_n\big(\tfrac{x+1}{n},\tfrac{x+1}{n}\big)- h_n\big(\tfrac{x}{n},\tfrac{x}{n}\big)\big]^2\\
& =(1+ \kappa_n \gamma_n)^4 n \int_{[-\frac{n}{2},\frac{n}{2}]}  \sin^{2} \big( \tfrac{\xi}{n} \big) I^2  \big( \tfrac{\xi}{n} \big) \big| {\mc F}_n (f) (\xi)\big|^2 \, d\xi \\
& \; \lesssim \; n^{-2} \, \int_{[-\frac{n}{2},\frac{n}{2}]}  |\xi|^3 \big| {\mc F}_n (f) (\xi)\big|^2 \, d\xi \; \lesssim \; n^{-2}.
\end{align*}
This proves \eqref{eq:est4}.

\subsection{Proof of Proposition \ref{prop:estimatevn}: estimates on $v_n$}

In this section we will use the explicit expression of the Fourier transform of $v_n$, given in \eqref{eq:fourierofv_new}, and then repeat the same computations as in the proof of Proposition \ref{prop:estimate2}.

\subsubsection{Proof of \eqref{eq:est0-vn}}

The same computation of Section \ref{ssec:l2ofh} gives that 
\begin{align*}
\tfrac{1}{n^2}\sum_{x,y\in \bb Z} v_n^2\big(\tfrac x n,\tfrac y n\big) & \; \lesssim \; n \int_{[-\frac12,\frac12]} \big|\mc F_n(f)(ny)\big|^2 \big|L(y)\big|^2\; W(y)\; dy \\
& \; \lesssim \; n \int_{[-\frac12,\frac12]} \big|\mc F_n(f)(ny)\big|^2 \big|\sin(\pi y)\big|^{\frac32}\; dy.
\end{align*}
The last inequality comes from Lemma \ref{lem:integral_residues} below, and using Lemma \ref{lem:sfp}, this proves  \eqref{eq:est0-vn}.

\subsubsection{Proof of \eqref{eq:est0x-vn}}

Similarly to Section \ref{ssec:diag} we have
\[ 
\tfrac1n \sum_{x\in\Z} v_n^2\big(\tfrac x n , \tfrac x n\big) = \int_{[-\frac n 2,\frac n 2]} \big|\mc F_n(\overline{v}_n(\xi)\big|^2  \; d\xi,
\]
where \[
\overline{v}_n\big(\tfrac x n\big)  = v_n\big(\tfrac x n,\tfrac x n\big),\qquad \text{and}\qquad 
\mc F_n(\overline{v}_n)(\xi)  = \int_{[-\frac n 2,\frac n 2]} \mc F_n(v_n)(\xi-\ell,\ell)\; d\ell.
\]
From \eqref{eq:fourierofv_new} we get
\begin{align*}
\mc F_n(\overline{v}_n)(\xi) & = \frac{1}{\sqrt n}\int_{[-\frac n 2,\frac n 2]} \frac{e^{2i\pi \frac{\xi-\ell}n}+e^{2i\pi \ell} }{((1+\gamma_n\kappa_n)\Lambda - i \Omega)(\frac{\xi-\ell}{n},\frac \ell n)} \; d\ell \; L\Big(\frac \xi n\Big) \mc F_n(f)(\xi) \vphantom{\Bigg(} \\
& = \sqrt n \; O\Big(\frac\xi n\Big)L\Big(\frac \xi n\Big)\; \mc F_n(f)(\xi) , \vphantom{\Bigg(} 
\end{align*}
where $O$ has been defined in \eqref{eq:defO}. 
From Lemma \ref{lem:integral_residues} we get
\[
\tfrac1n \sum_{x\in\Z} v_n^2\big(\tfrac x n , \tfrac x n\big) \; \lesssim \; n^{-1} \int_{[-\frac n 2,\frac n 2]}  |\xi|^2\; \big|\mc F_n(f)(\xi)\big|^2\; d\xi.
\]
From Corollary \ref{cor:fourier} we get \eqref{eq:est0x-vn}.

\subsubsection{Proof of \eqref{eq:est1-vn}}

A straightforward computation gives that 
\[ \mc F_n(\mc D_n v_n)(\xi) = -2n\big(1-e^{2i\pi\frac \xi n}\big) \mc F_n(w_n)(\xi) M\big(\tfrac \xi n \big),
\]
where $w_n:\frac1n \Z \to \R$ has been defined in \eqref{eq:w} and $M$ has been defined in \eqref{eq:defM}. 
Finally from \eqref{eq:fourierofw} we get 
\[
\mc F_n(\mc D_nv_n)(\xi)=n^{\frac32}\big(1-e^{2i\pi\frac \xi n}\big) \mc F_n(f)(\xi) L\big(\tfrac \xi n \big) M\big(\tfrac \xi n\big),
\]
where $L$ is defined in \eqref{eq:defL}. Therefore, using Lemma \ref{lem:integral_residues} below, 
\begin{align*}
\big\| \mc D_nv_n\big\|_{2,n}^2& \lesssim \;  n^3 \int_{[-\frac n 2, \frac n 2]} \sin^2\big(\tfrac{\pi \xi}{n}\big) \big|\mc F_n(f)(\xi)\big|^2 \big|L\big(\tfrac \xi n \big)\big|^2 \big| M\big(\tfrac \xi n \big)\big|^2 \; d\xi \\
& \lesssim \;  n^4 \int_{[-\frac12,\frac12]} \sin^4(\pi y) \big|\mc F_n(f)(ny)\big|^2\; dy,
\end{align*}
and  from Lemma \ref{lem:sfp} this proves \eqref{eq:est1-vn}.

\subsubsection{Proof of \eqref{eq:est4-vn}}
As in Section \ref{ssec:grad}, using Lemma \ref{lem:integral_residues} we have
\begin{align*}
 \sum_{x \in \Z} \big[v_n\big(\tfrac{x+1} n,\tfrac{x+1}n\big)&-v_n\big(\tfrac x n,\tfrac x n\big)\big]^2  = n\int_{[-\frac n 2,\frac n 2]} \big|1-e^{2i\pi\frac \xi n}\big|^2 \; \big|\mc F_n(\overline{v}_n)(\xi)\big|^2\; d\xi \\
& = 4n^2 \int_{[-\frac n 2,\frac n 2]} \sin^2\big(\tfrac \xi n\big) \; \big|O \big(\tfrac \xi n\big)\big|^2 \;  \big|L \big(\tfrac \xi n\big)\big|^2 \; \big|\mc F_n(f)(\xi)\big|^2\; d\xi \\
&\lesssim \;  n^{-2}  \int_{[-\frac n 2,\frac n 2]}  |\xi|^4\; \big|\mc F_n(f)(\xi)\big|^2\; d\xi,
\end{align*}
which proves \eqref{eq:est4-vn}, from Corollary \ref{cor:fourier}.

\subsubsection{Proof of \eqref{eq:est5-vn}}
Let $\theta_n:\frac1n\Z\to\R$ be defined by
\[ \theta_n\big(\tfrac x n\big) = v_n\big(\tfrac x n,\tfrac{x+1} n\big)-v_n\big(\tfrac x n, \tfrac x n\big), \]
so that 
\[ \tfrac{1}{\sqrt n} \widetilde{\mc D}_nv_n\big(\tfrac x n,\tfrac y n \big) = n^{\frac32} \begin{cases} \theta_n(\frac x n); & y=x+1 \\ 
\theta_n(\frac{x-1}n); & y=x-1 \\ 
0 & \text{otherwise}.\end{cases}\]
Moreover, 
\[ \mc F_n(\theta_n)(\xi)  = \int_{[-\frac n 2,\frac n 2]} \mc F_n(v_n)(\xi-\ell,\ell)\big(e^{-2i\pi\frac \ell n}-1\big)\; d\ell \\
 = \sqrt n \mc F_n(f)(\xi) L\big(\tfrac \xi n\big) N\big(\tfrac \xi n \big), \]
where $L$ has been defined in \eqref{eq:defL} and $N$ is given by 
\begin{equation}\label{eq:defN}
N(y) =\cfrac{1}{1- e^{-2 i \pi y}} \, \int_{[-\frac 1 2,\frac 12]} (e^{-2 i \pi x} -1) \, \Theta (y-x,x)\,  dx. 
\end{equation}
Therefore, using Lemma \ref{lem:integral_residues} below, we get
\begin{align*}
n^{-2} \big\| \widetilde{\mc D}_nv_n\big\|_{2,n}^2 = n \big\|\theta_n\big\|_{2,n}^2 & \lesssim \; n^2 \int_{[-\frac n 2, \frac n 2]}  \big|\mc F_n(f)(\xi)\big|^2 \big|L\big(\tfrac \xi n \big)\big|^2 \big| N\big(\tfrac \xi n \big)\big|^2 \; d\xi \\
& \lesssim \;  \frac{1}{n^2} \int_{[-\frac n2,\frac n2]} |\xi|^4 \big|\mc F_n(f)(\xi)\big|^2\; d\xi
\end{align*}
and by Corollary \ref{cor:fourier} the proof ends.
\section{Technical integral estimates}
\label{app:integral}

Recall the definitions of $I,J,K,L,M,N,O$ given in \eqref{eq:defI}, \eqref{eq:defJ}, \eqref{eq:defK}, \eqref{eq:defL}, \eqref{eq:defM}, \eqref{eq:defN}, \eqref{eq:defO}, respectively. Note that we have the relations 
\begin{equation} 
 M(y) = \frac{1}{|w-1|^2} \, J(y), \qquad
N(y)  = \frac{w}{1-w} \, L(y), \qquad O(y) = \frac{w}{w-1} \, I(y) \label{eq:relation} \end{equation} 
and also 
\begin{equation}
L(y)=I(y) - \frac{1}{1-w} J(y). \label{eq:relation2} 
\end{equation}

\subsection{Uniform bounds}
In this section we prove the following:

\begin{lem} 
\label{lem:integral_residues}
For any $y\in\big[-\tfrac12,\tfrac12\big]$
\begin{align}
|I(y)| & \lesssim  |\sin(\pi y)|^{\frac12}, \label{eq:I}\\
|J(y)|&  \lesssim |\sin(\pi y)|^{\frac32},\label{eq:J} \\
|K(y) | & \lesssim |\sin(\pi y)|^{\frac32},\label{eq:K} \\
|L(y)| &  \lesssim  |\sin(\pi y)|^{\frac32}, \label{eq:L}
\end{align}
and therefore \eqref{eq:relation} implies also that \begin{align}
|M(y)| &  \lesssim  |\sin(\pi y)|^{-\frac12}, \label{eq:M} \\
|N(y)| &  \lesssim  |\sin(\pi y)|^{\frac12}, \label{eq:N} \\
|O(y)| &  \lesssim  |\sin (\pi y)|^{-\frac12}. \label{eq:O}
\end{align}
\end{lem}

\begin{proof} The proof consists in using the residue theorem to estimate each integral. For any $y \in \big[-\frac{1}{2}, \frac{1}{2}\big]$ we denote by $w:=w(y)$ the complex number $w=e^{2i\pi y}$. We denote by $\cC$ the unit circle positively oriented, and $z:=e^{2i\pi x}$ is the complex integration variable in $\cC$. With these notations we have
\begin{equation}\label{eq:lambomeg}
\Lambda ( y - x, x ) = 4 - z ( \bar w + 1 ) - \bar z ( w + 1 ) , \quad 
i \Omega( y - x, x ) = z ( 1 - \bar w ) + \bar z ( w - 1 ).
\end{equation}
Some quantities are going to appear many times, therefore for the sake of clarity we introduce some notations.   Hereafter, for any complex number $z$, we denote by $\sqrt z$ its principal square root, with positive real part. Precisely, if $z=re^{i\varphi}$, with $r \ge 0$ and $\varphi \in (-\pi,\pi]$, then the principal square root of $z$ is $\sqrt z=\sqrt r e^{i\varphi/2}.$
We introduce the degree two complex polynomial:
\begin{equation} P_w(z):=z^2-\frac{4(1+\gamma_n\kappa_n)z}{(1+\kappa_n\gamma_n)(1+\bar w)+1-\bar w}+w=(z-z_-)(z-z_+), \label{eq:pz}\end{equation}
where $|z_-|<1$ and $|z_+|>1$. The useful identities are \[
z_-z_+ =w, \qquad
z_-+z_+=\frac{4(1+\kappa_n\gamma_n)}{(1+\kappa_n\gamma_n)(\bar w+1)+1-\bar w}.\]
Finally, we denote 
\begin{align*}
a_n(w):&= (1+\kappa_n\gamma_n)(1+\bar w)+1-\bar w=2+\kappa_n\gamma_n(1+\bar w)\\
\delta_n(w)  :&= (1-w)(2+\kappa_n\gamma_n)^2+(1-\bar w)(\kappa_n\gamma_n)^2 \\
&= 4(1+\kappa_n\gamma_n)(1-w)+(\kappa_n\gamma_n)^2(2-w-\bar w),
\end{align*}
so that the discriminant of $P_w$ is $4 \delta_n(w) / a_n^2 (w)$ and
\[
z_+ = \frac{2(1+\kappa_n\gamma_n)+\sqrt{\delta_n(w)}}{a_n(w)},\qquad  z_-=\frac{2(1+\kappa_n\gamma_n)-\sqrt{\delta_n(w)}}{a_n(w)}.
\]
Note the following identity:
\begin{equation}
\label{eq:star}
 w a_n^2(w) + \delta_n(w) = 4(1+\kappa_n\gamma_n)^2.
\end{equation}
Moreover
\begin{equation}
\label{eq:delta}
\delta_n(w)= (2+\kappa_n\gamma_n)^2 (1-w) \bigg[1-\frac{(\kappa_n\gamma_n)^2 \bar w}{(2+\kappa_n\gamma_n)^2}\bigg].
\end{equation}
With these notations, note also that \eqref{eq:lambomeg} gives two very useful identities:
\begin{align*}
&i\Omega(y-x,x)  = (z+w\bar z)(1-\bar w)\\
&\big[(1+\kappa_n\gamma_n)\, \Lambda-i\Omega \big] (y-x,x)  = -\frac{a_n(w)}{z}P_w(z).
\end{align*}
Let us now give some estimates as $n\to \infty$. For $n$ large enough, we have
\[ \bigg|\frac{(\kappa_n\gamma_n)^2 \bar w}{(2+\kappa_n\gamma_n)^2}\bigg|\leq \frac{1}{2}.\]
Therefore
\begin{equation}\label{eq:esti-delta}
\bigg|\frac{1}{\sqrt{\delta_n(w)}}\bigg|  = \bigg|\frac{(2+\kappa_n\gamma_n)^{-1}}{\sqrt{1-w}}\bigg|\; \bigg|1- \frac{(\kappa_n\gamma_n)^2 \bar w}{(2+\kappa_n\gamma_n)^2}\bigg|^{-\frac12} \lesssim  \big|\sin(\pi y)\big|^{-\frac12}
\end{equation}
since 
\[
\sqrt{1-w}=\big|\sin(\pi y)\big|^{\frac12}\; e^{\frac{i\pi y}{2}} \big(1-i\;\mathrm{sgn}(y) \big),
\]
and $|1+z|^{-\frac12}  \lesssim 1$ for all $|z| \leq \frac12$. 
Let us resume in the following lemma the estimates that shall be needed:

\begin{lem}\label{lem:estidelta}
For any $w=e^{2i\pi y}$ with $y \in [-\frac12,\frac12]$, and any $n$ large enough,
\[
\bigg|\frac{1}{\sqrt{\delta_n(w)}}\bigg|  \lesssim  \big|\sin(\pi y)\big|^{-\frac12}.
\]
We also have
\[1   \lesssim  |a_n(w)|  \lesssim  1. \]
\end{lem}
\bigskip
\begin{enumerate}[(i)]
\item \textsc{Proof of \eqref{eq:I}.} We have
\[
I(y)= \frac{1}{2 i \pi} \frac{1-w}{w} \frac{1}{a_n(w)}\oint_{\cC} dz\; f_{w} (z),
\]
with the function $f_{w}$ defined by
\begin{equation*}
f_w (z)= \frac{(z^2 +w)}{z (z-z_+)(z-z_-)}
\end{equation*}
where $z_\pm$ are the two complex solutions of $P_w(z)=0$ defined in \eqref{eq:pz}.
Since $|z_-| <1$ and $|z_+|>1$, by the residue theorem we have
\[
I(y) =\frac{1-w}{w\; a_n(w)} \Big[ {\rm{Res}} (f_w,0) +{\rm {Res}} (f_w,z_-)\Big].
\]
It is easy to see that
\begin{equation*}
{\rm{Res}} (f_w,0) = 1, 
\end{equation*}
and 
\[
{\rm {Res}} (f_w,z_-) =\frac{z_-^2+w}{z_-(z_--z_+)}=-\frac{2(1+\gamma_n\kappa_n)}{\sqrt{\delta_n(w)}}.
\]
It follows by Lemma \ref{lem:estidelta} that
\begin{align*}
|I(y)| & =  \bigg|\frac{w-1}{w \; a_n(w)} \bigg[ 1 -\frac{2(1+\gamma_n\kappa_n)}{\sqrt{\delta_n(w)}}\bigg]\bigg|  = \frac{2|\sin(\pi y)|}{|a_n(w)|} \bigg|1 -\frac{2(1+\gamma_n\kappa_n)}{\sqrt{\delta_n(w)}}\bigg| \\
&  \lesssim  \big|\sin(\pi y) \big| \Big[ 1 + |\delta_n(w)|^{-\frac12}\Big]  \lesssim  \big|\sin(\pi y)\big|^{\frac12} .
\end{align*}
\bigskip

\item \textsc{Proof of \eqref{eq:J}.}
In the same way, we have
\begin{equation*}
\begin{split}
J(y)&= \frac{(1-w)^2}{w\; a_n(w)} \, \cfrac{1}{2i \pi} \oint g_w (z) dz= \frac{(1-w)^2}{w\; a_n(w)} \Big[ {\rm{Res}} (g_w,0) +{\rm {Res}} (g_w,z_-)\Big],
\end{split}
\end{equation*}
with the function $g_{w}$ defined by
\begin{equation*}
g_w (z)= \frac{(z^2 +w)}{z^2 (z-z_+)(z-z_-)}.
\end{equation*}
Here we have
\begin{equation*}
{\rm{Res}} (g_w,0) =\cfrac{z_+ + z_-}{w}=  \frac{4(1+\gamma_n\kappa_n)}{w\; a_n(w)}, 
\end{equation*}
and 
\[
{\rm {Res}} (g_w,z_-) =\frac{z_-^2+w}{z_-^2(z_--z_+)}=-\frac{2(1+\gamma_n\kappa_n)}{z_-\sqrt{\delta_n(w)}}.
\]
By \eqref{eq:star} it follows that
\begin{align*}
J(y) & = \frac{2(1-w)^2(1+\gamma_n\kappa_n)}{w^2 \; a_n^2(w)} \bigg[1-\frac{2(1+\gamma_n\kappa_n)}{\sqrt{\delta_n(w)}}\bigg] \\ &  = 2(1+\gamma_n\kappa_n)\, \frac{1-w}{w\; a_n(w)}\,  I(y).
\end{align*}
Therefore we conclude that, for any $y\in[-\frac12,\frac12]$, 
$ |J(y)| \lesssim |\sin(\pi y)|^{\frac32}.$
\bigskip

\item \textsc{Proof of \eqref{eq:K}.} Again, similar computations give
\begin{equation*}
\begin{split}
K(y)&=   \frac{(1-w)^2}{w^2} \frac{1}{a_n(w)}\, \cfrac{1}{2i \pi} \oint k_w (z) dz\\
&   =\frac{(1-w)^2}{w^2} \frac{1}{a_n(w)} \Big[{\rm{Res}} (k_w,0) +{\rm {Res}} (k_w,z_-)\Big] ,
\end{split}
\end{equation*}
with the function $k_{w}$ defined by
\begin{equation*}
k_w (z)= \frac{(z^2 +w)^2}{z^2 (z-z_+)(z-z_-)}.
\end{equation*}
Here we have
\begin{equation*}
{\rm{Res}} (k_w,0) = z_-+z_+=\frac{4(1+\gamma_n\kappa_n)}{a_n(w)}, 
\end{equation*}
and 
\[
{\rm {Res}} (k_w,z_-) =\frac{(z_-^2+w)^2}{z_-^2(z_--z_+)}=-\frac{8(1+\gamma_n\kappa_n)^2}{a_n(w)\sqrt{\delta_n(w)}}.
\]
It follows that
\begin{align*}
K(y) & =  \frac{(w-1)^2}{w^2\; a_n^2(w)}4(1+\kappa_n\gamma_n) \bigg[1-\frac{2(1+\gamma_n\kappa_n)}{\sqrt{\delta_n(w)}}\bigg]\\
&=4(1+\gamma_n\kappa_n) \frac{1-w}{w\; a_n(w)} I(y) =2 J(y).
\end{align*} 
Therefore we conclude that, for any $y\in[-\frac12,\frac12]$, 
$ |K(y)| \lesssim |\sin(\pi y)|^{\frac32}.$
\bigskip

\item \textsc{Proof of \eqref{eq:L}.} 
Here we use the relation \eqref{eq:relation2}, and we obtain: 
\[L(y) = I(y) - \frac{1}{1-w}J(y)=I(y)\bigg[1-\frac{2(1+\gamma_n\kappa_n)}{wa_n(w)}\bigg].\]
It is easy to check that \[\bigg|1-\frac{2(1+\gamma_n\kappa_n)}{wa_n(w)}\bigg|=\bigg|\frac{(2+\gamma_n\kappa_n)(w-1)}{wa_n(w)}\bigg| \lesssim |\sin(\pi y)|.\]
Using  \eqref{eq:I} we conclude that, for any $y\in[-\frac12,\frac12]$, 
$ |L(y)| \lesssim  |\sin(\pi y)|^{\frac32}.$
%
\end{enumerate}
This concludes the proof of Lemma \ref{lem:integral_residues}.
\end{proof}

\subsection{Proof of Lemma \ref{lem:est}} \label{app:gn}
Recall that $G_n$ has been defined in \eqref{eq:defgn}, and it equals
\begin{align*} 
G_n(y)&=\frac{(1+\gamma_n\kappa_n)^2}{4}K(y)=(1+\gamma_n\kappa_n)^3 \frac{1-w}{w\; a_n(w)} I(y) \\
& =\frac{(1+\gamma_n\kappa_n)^3}{a_n^2(w)}\frac{(1-w)^2}{w^2}\bigg[1-\frac{2(1+\gamma_n\kappa_n)}{\sqrt{\delta_n(w)}}\bigg]\\
&=\cfrac{1+\kappa_n \gamma_n}{4} \; w (\bar w -1)^2\;  \bigg(\frac{w a_n^2 (w) }{4(1+\kappa_n \gamma_n)^2}\bigg)^{-1} \;\Bigg[ 1-\bigg( 1- \cfrac{w a_n^2 (w) }{4(1+\kappa_n \gamma_n)^2}\bigg)^{-\frac12}\Bigg]\\
&=\cfrac{1+\kappa_n \gamma_n}{4} \; w (\bar w -1)^2 g \bigg(\frac{w a_n^2 (w) }{4(1+\kappa_n \gamma_n)^2}\bigg) \\
& \qquad - \cfrac{1+\kappa_n \gamma_n}{4} \; w (\bar w -1)^2 \bigg( 1- \cfrac{w a_n^2 (w) }{4(1+\kappa_n \gamma_n)^2}\bigg)^{-\frac12}
\end{align*}
where $g: u \in \C -\{1\} \to u^{-1} (1- (1-u)^{-\frac12})+ (1-u)^{-\frac12}$. We have that $g(u)=\tfrac{1}{1+ \sqrt{1-u}}$. Since $\sqrt{1-u}$ has a positive real part, we deduce that the function $g$ is uniformly bounded. Therefore, 
\begin{equation*}
\left| G_n (y) +  \cfrac{1+\kappa_n \gamma_n}{4} \, w (\bar w -1)^2 \bigg( 1- \cfrac{w a_n^2 (w) }{4(1+\kappa_n \gamma_n)^2}\bigg)^{-\frac12} \right| \lesssim | w-1|^2.
\end{equation*}
Let us now observe that, since  $a_n(w)= (1+\kappa_n\gamma_n)(1+\bar w)+1-\bar w$, we have that
\begin{equation}
\label{eq:tournevis1}
1-\cfrac{w a_n^2 (w)}{4 (1 +\gamma_n \kappa_n)^2} = \bigg( 1+\cfrac{1}{(1+ \kappa_n \gamma_n)^2 } \bigg)\sin^2 (\pi y) - \cfrac{ i}{1+\kappa_n \gamma_n } \sin (2\pi y).
\end{equation}
Note also that
\begin{multline*}
{\rm{Arg}} \bigg( 1-\cfrac{w a_n^2 (w)}{4 (1 +\gamma_n \kappa_n)^2}\bigg) =-{\rm{sgn}} (y) \cfrac{\pi}{2} \\ + \arctan \left(\frac 12  \Big(1+\kappa_n \gamma_n+\frac{1}{1+\kappa_n\gamma_n}\Big)  \tan (\pi y)\right).\end{multline*}
Since $\sin (2\pi y) =2 \sin (\pi y) \cos (\pi y)$ and $\cos^2(\pi y)= 1- \sin^2(\pi y)$, we have that
\begin{equation}
\label{eq:tournevis2}
\bigg|  1-\cfrac{w a_n^2 (w)}{4 (1 +\gamma_n \kappa_n)^2}   \bigg|^2  =  \frac{4\sin^2 (\pi y)}{(1+\kappa_n\gamma_n)^2}  \left\{\left( \frac{1+\kappa_n\gamma_n}{2}-\frac{1}{2(1+ \kappa_n \gamma_n) } \right)^2 \sin^2 (\pi y)+{1} \right\}.
\end{equation}
Therefore, by observing that $(1-\bar w)^2 \;w = - 4 \sin^2 (\pi y)$, we have that
\begin{equation*}
\begin{split}
G_n (y) &= (1+ \kappa_n \gamma_n)^{\frac32} \,\cfrac{|\sin (\pi y)|^{\frac32}}{\sqrt 2\Big\{\big( \frac{1+\kappa_n\gamma_n}{2}-\frac{1}{2(1+ \kappa_n \gamma_n)} \big)^2 \sin^2 (\pi y)+{1} \Big\}^{\frac14}} \; e^{i \varphi_n (y)} + \mc O(|y|^2)\\
&=\cfrac{|\sin(\pi y)|^{\frac32}}{\sqrt 2} \, e^{i \varphi_n (y)} \; + \; \varepsilon_n (y) +\mc O(|y|^2)
\end{split}
\end{equation*}
where
\begin{equation*}
\varphi_n (y) = {\rm{sgn}} (y) \cfrac{\pi}{4} -\cfrac{1}{2} \, \arctan  \left(\frac 12  \Big(1+\kappa_n \gamma_n+\frac{1}{1+\kappa_n\gamma_n}\Big)  \tan (\pi y)\right)
\end{equation*}
and
$|\varepsilon_n (y)| \le \gamma_n |y|^{\frac32}.
$
Moreover we have that
\begin{align*}
\Big| e^{i \varphi_n (y)} -e^{i {\rm{sgn}} (y) \tfrac{\pi}{4}}\Big| &= \bigg| 1- e^{-\tfrac{i}{2} \arctan  \big(\frac 12  \big(1+\kappa_n \gamma_n+\frac{1}{1+\kappa_n\gamma_n}\big)  \tan (\pi y)\big) }\bigg|\\
&\lesssim \bigg| \arctan  \bigg(\frac 12  \Big(1+\kappa_n \gamma_n+\frac{1}{1+\kappa_n\gamma_n}\Big)  \tan (\pi y)\bigg) \bigg|\lesssim |y|.
\end{align*}
We conclude that
\begin{align*}
G_n (y) &= \cfrac{|\sin(\pi y)|^{\frac32}}{\sqrt 2} \, e^{i {\rm{sgn}} (y) \tfrac{\pi}{4} } \; + \; \mc O( \gamma_n |y|^{\frac32}) + \mc O(|y|^2)\\
&=\cfrac{\pi^{\frac32} |y|^{\frac32}}{\sqrt 2} \, e^{i {\rm{sgn}} (y) \tfrac{\pi}{4} } \; + \;  \mc O( \gamma_n |y|^{\frac32}) + \mc O(|y|^2).
\end{align*}

\section{Proof of Propositions \ref{prop:vol_decomposition} and \ref{prop:equadiff}}
\label{app:equadiff}
In this section we prove  
Propositions \ref{prop:vol_decomposition} and  \ref{prop:equadiff}. To simplify the notations we write $\mc L_n$, $\mc A_{n}$, $e_n$ for ${\mc L}_{\gamma_n}$, ${\mc A}_{\gamma_n}$ and $e_{\gamma_n}$. We start by showing Proposition \ref{prop:vol_decomposition}. 
Let $f:\Z\to\R$   be a function of finite support and let us define
\begin{align*}
V(f)&= \sum_{x\in\Z} f (x)  \omega_x, \quad V^3 (f) = \sum_{x\in\Z} f(x) (\omega_x^3-\kappa_n\omega_x)=\sum_{x\in \Z} f(x)  H_{3\delta_x}(\omega).
\end{align*}
Observe that 
$$\S V(f)=\sum_{x\in\Z}\Delta f(x) \; \omega_x = V(\Delta f),$$
where
$\Delta f(x)=f(x+1)+f(x-1)-2f(x)$  and
\begin{align*}
\A_nV(f)& =\sum_{x\in\Z} \Big[ \Delta f (x)   -2 \nabla f(x) \Big] \;  ( \omega_x+\gamma_n\omega_{x}^3)\\
&= V\Big( (1+\gamma_n \kappa_n) (\Delta f -2 \nabla f) \Big) +\gamma_n V^3 \Big( (\Delta f -2 \nabla f \Big),
\end{align*}
where  $\nabla f(x)=f(x+1)-f(x)$.
Therefore,
\begin{align*}
\cL_n  V(f)\; = \; V\Big( (2+\gamma_n \kappa_n) \Delta f -2 (1+\gamma_n \kappa_n) \nabla f \Big) +\gamma_n V^3 \Big( \Delta f -2 \nabla f \Big).
\end{align*}
From last computations it is simple to obtain Proposition \ref{prop:vol_decomposition}.

Now we prove Proposition \ref{prop:equadiff} and we start with  \eqref{eq:S}. 
 Let $f:\Z\to\R$ and $h:\Z^2\to\R$ be functions of finite support and define: 
\begin{align*}
\mathbf{E}_n(f)&=\sum_{x\in\Z} f(x)e_n(\omega_x)=\sum_{x\in\Z}f(x)\Big[\frac{\omega_x^2}{2}+\gamma_n\frac{\omega_x^4}{4}\Big], \\
{\mathbf{E}^4}(f)& = \sum_{x\in\Z}f(x)\omega_x^4,\\
 \mathbf{Q}^2(h)&=\sum_{x\neq y} h(x,y)\omega_x\omega_y=\sum_{x\neq y} h(x,y) H_{\delta_x+\delta_y}(\omega),\\
  \mathbf{Q}^4(h)&=\sum_{x\neq y} h(x,y)(\omega^3_x-\kappa_n\omega_x)\omega_y = \sum_{x\neq y} h(x,y) H_{3\delta_x+\delta_y}(\omega),\\
   \mathbf{Q}^6(h)&=\sum_{x\neq y} h(x,y)(\omega^3_x-\kappa_n\omega_x)(\omega^3_y-\kappa_n\omega_y)=\sum_{x\neq y} h(x,y) H_{3\delta_x+3\delta_y}(\omega).
\end{align*}
We define the symmetric (resp. antisymmetric) part $h^s$ (resp. $h^a$) of $h$ by $h^{s} (x,y)= \tfrac{1}{2} [ h(x,y) +h(y,x)]$ (resp. $h^a=h-h^s$). Observe that $\mathbf{Q}^2(h)=\mathbf{Q}^2(h^s)$ and $ \mathbf{Q}^6=\mathbf{Q}^6(h^s)$ depend only on the symmetric part $h^s$ of $h$ but that it is not the case for $ \mathbf{Q}^4(h)$.  Then, observe that 
\begin{equation*}
\S\mathbf{E}_n(f)=\sum_{x\in\Z}\Delta f(x) e_n(\omega_x),
\end{equation*}
and \begin{equation*}
\A_n\mathbf{E}_n(f)=-\sum_{x\in\Z}\nabla f(x)\big[\omega_x\omega_{x+1}+\gamma_n(\omega_{x}\omega_{x+1}^3+\omega_{x+1}\omega_x^3)+\gamma_n^2\omega_{x+1}^3\omega_x^3\big].
\end{equation*}
Therefore,
\begin{align*}
\cL_n\mathbf{E}_n(f)=&\mathbf{E}_n(\Delta f) -(1+\gamma_n\kappa_n)^2 \mathbf{Q}^2(\nabla f\otimes\delta)\\&-2\gamma_n(1+\gamma_n\kappa_n)\mathbf{Q}^4(\nabla f \otimes\delta)-\gamma_n^2 \mathbf{Q}^6(\nabla f\otimes\delta),
\end{align*}
where 
\begin{align*}
 \nabla f\otimes \delta(x,y) &=   \begin{cases} \displaystyle\big(f(x+1)-f(x)\big)/2 & \text{ if } y=x+1,\\
  \displaystyle\big(f(x)-f(x-1)\big)/2 & \text{ if } y= x-1,\\
 0 & \text{ otherwise}.\end{cases}
\end{align*}
From this it is easy to obtain \eqref{eq:S}. Now we prove \eqref{eq:Q}.
For any symmetric function $h$, we have that
\begin{align*}
\cL_n\mathbf{Q}^2(h)= \mathbf{Q}^2(\Delta h+\mathbf{A}h)&+2\sum_{x\in\Z}(h(x-1,x)-h(x+1,x))(\omega_x^2 + \gamma_n \omega_x^4) \\
& + 4\sum_{x\in\Z}(h(x,x+1)-h(x,x))\omega_x\omega_{x+1}\\
& + 2\gamma_n \sum_{x\neq y}\big[h(x-1,y)-h(x+1,y)\big]\omega_x^3\omega_y \\
&  + 2 \gamma_n \sum_{x\in\Z} \big[h(x+1,x+1)\omega_x^3\omega_{x+1}-h(x,x)\omega_{x+1}^3\omega_x\big].
\end{align*}
%
We now replace $\omega_x^3\omega_y$ by $(\omega_x^3-\kappa_n \omega_x)\omega_y$ in order to see the product of two orthogonal polynomials, and write the decomposition in the basis $\{H_\sigma\; ; \sigma \in \Sigma\}$ as
\begin{align*}
\cL_n\mathbf{Q}^2(h) = & \; \mathbf{Q}^2\big(\Delta h+(1+\gamma_n \kappa_n)\mathbf{A}h\big)  - 4 \mathbf{E}_n(\mathbf{D}h) - \gamma_n {\mathbf{E}}^4(\mathbf{D}h)+2\mathbf{Q}^2(\widetilde{\mathbf{D}}h)\notag\\
& + 2\gamma_n \sum_{x\neq y}\big[h(x-1,y)-h(x+1,y)\big](\omega_x^3-\kappa_n\omega_x)\omega_y\\
& + 2 \gamma_n \sum_{x\in\Z}h(x,x)\big[(\omega_{x}^3-\kappa_n \omega_x)\omega_{x+1}-(\omega_{x+1}^3-\kappa_n \omega_{x+1})\omega_{x}\big]\\
& + 2 \gamma_n \sum_{x\in\Z}\big[h(x+1,x+1)-h(x,x)\big] (\omega_{x}^3-\kappa_n\omega_x)\omega_{x+1}\\
& + 2 \gamma_n \kappa_n \sum_{x\in\Z}\big[h(x+1,x+1)-h(x,x)\big] \omega_x\omega_{x+1}.
\end{align*}
 We rewrite the previous  identity as
\begin{align*}
\cL_n\mathbf{Q}^2(h) = &\mathbf{Q}^2\big(\Delta h+(1+\gamma_n \kappa_n)\mathbf{A}h\big)  - 4 \mathbf{E}_n(\mathbf{D}h) - \gamma_n {\mathbf{E}}^4(\mathbf{D}h)+2\mathbf{Q}^2(\widetilde{\mathbf{D}}h)\notag\\
& + 2\gamma_n \sum_{x\neq y}\big[h(x-1,y)-h(x+1,y)\big]H_{3\delta_x+\delta_y}(\omega)\\
& + 2 \gamma_n \sum_{x\in\Z}h(x,x)\big[H_{3\delta_x+\delta_{x+1}}-H_{3\delta_{x+1}+\delta_x}\big](\omega)\\
& + 2 \gamma_n \sum_{x\in\Z}\big[h(x+1,x+1)-h(x,x)\big] \; \big[H_{3\delta_x+\delta_{x+1}}+\kappa_nH_{\delta_x+\delta_{x+1}}\big](\omega)\\
=&\mathbf{Q}^2\big(\Delta h+(1+\gamma_n \kappa_n)\mathbf{A}h\big)  - 4 \mathbf{E}_n(\mathbf{D}h) +2\mathbf{Q}^2(\widetilde{\mathbf{D}}h)\notag\\
& + 2\gamma_n \mathbf{Q}^4(\mathbf{B}h) - \gamma_n {\mathbf{E}}^4(\mathbf{D}h)+2 \gamma_n \kappa_n \mathbf{Q}^2(\nabla h),
\end{align*}
where 
\begin{align*}
\mathbf{D}h(x)&=h(x,x+1)-h(x-1,x).\\
 \nabla h(x,y)&=\begin{cases} (h(x+1,x+1)-h(x,x))/2 & \text{ if } y=x+1,\\
 (h(x,x)-h(x-1,x-1))/2 & \text{ if } y= x-1,\\
 0 & \text{ otherwise}.\end{cases}\\
\Delta h(x,y) &  =  h(x+1,y)+h(x-1,y) + h(x,y+1)+h(x,y-1)-4h(x,y),\\
 \mathbf{A}h(x,y)& =h(x-1,y) + h(x,y-1) -h(x+1,y) - h(x,y+1),\\
 \widetilde{\mathbf D}h(x,y)& = \begin{cases} h(x,x+1)-h(x,x) & \text{ if } y=x+1,\\
 h(x-1,x)-h(x-1,x-1) & \text{ if } y= x-1,\\
 0 & \text{ otherwise}.\end{cases}\\
 \mathbf{B}h(x,y) & =  \big( h(x-1,y) -h(x+1,y)\big)+ \big({\bf 1}_{y=x+1} -{\bf 1}_{y=x-1}\big) h(y,y).
\end{align*}
Remark that for any $f:\ZZ \to \RR$, $h:\ZZ^2 \to \RR$ the functions $\nabla h$, ${\widetilde{\mathbf D}} h$, $\nabla f \otimes \delta$ are always symmetric and that the operators $\Delta$ and ${\mathbf A}$ preserve the parity of functions.
From last computations it is easy to recover  \eqref{eq:Q}.

%
%
%

\end{document}